\documentclass[
reprint,
superscriptaddress,
 amsmath,amssymb,
 aps,
pra,
]{revtex4-2}

\usepackage{hyperref}
\hypersetup{colorlinks,linkcolor=blue,urlcolor=blue,citecolor=blue,hypertexnames=false}
\usepackage[whole]{bxcjkjatype}

\usepackage{graphicx}% Include figure files
\usepackage{dcolumn}% Align table columns on decimal point
\usepackage{bm}% bold math
\usepackage{physics}
\usepackage{amsthm}
\usepackage{comment}
\usepackage{newtxtext,newtxmath}
\usepackage{tikz}
\usetikzlibrary{positioning}
\usepackage{here}
\usepackage{algorithm}
\usepackage{algpseudocode}

\usepackage{varwidth}
\usepackage{xcolor}

\newtheorem{thm}{Theorem}
\newtheorem{lem}[thm]{Lemma}

\newtheorem{prop}[thm]{Proposition}

\newtheorem{cor}[thm]{Corollary}

\theoremstyle{definition}
\newtheorem{defi}{Definition}
\newtheorem*{defi*}{Definition}
\newtheorem{exa}{Example}

\begin{document}
\title{Efficient LOCC extraction of quantum information encoded in stabilizer codes}
%\title{Efficient decoding of stabilizer code by single-qubit local operations and classical communication}% Force line breaks with \\

\author{Koki Shiraishi}
 \email{kokishiraishi@g.ecc.u-tokyo.ac.jp}

\affiliation{
 Department of Physics,  Graduate School of Science, The University of Tokyo, Tokyo, Japan
}
\author{Hayata Yamasaki}
\affiliation{
Department of Physics, Graduate School of Science, The University of Tokyo, Tokyo, Japan
}
\author{Mio Murao}
\affiliation{
Department of Physics, Graduate School of Science, The University of Tokyo, Tokyo, Japan
}

\date{\today}% It is always \today, today,

\begin{abstract}
We construct a protocol for extracting distributed single-qubit quantum information encoded in a stabilizer code of multiple qubits, only by single-qubit local operations and classical communication (LOCC) without global operations or entanglement resources.
This protocol achieves efficient extraction within a polynomial time in terms of the number of physical qubits.
We apply this protocol to a setting of quantum information splitting where a subset of spatially separated parties cooperate by classical communication to extract quantum information shared among all the parties. For this task, our LOCC extraction protocol allows designing hierarchical information access structures among the parties, where the minimum number of parties required to cooperate depends on the location of extracting the shared quantum information.
Extraction of quantum information encoded in the stabilizer code appears widely in distributed quantum information processing, such as quantum secret sharing.
Our results allow such distributed quantum information processing without entanglement costs.

\end{abstract}

\maketitle

%\tableofcontents
\section{Introduction}
\label{sec:1}

Connecting multiple quantum computers by classical and quantum communication channels enables various types of distributed quantum tasks~\cite{QuantumInternet,doi:10.1126/science.aam9288}, such as quantum secret sharing~\cite{cleve1999share,gottesman2000theory,singh2005generalized}.
In general, distributed quantum tasks require global quantum operation between distant quantum computers.
Such global quantum operations are implemented by two-qubit gates across the distant qubits.
To perform global two-qubit gates, one has to implement (or simulate) quantum communication between the qubits in different quantum computers, for example, by coherently transporting physical qubits~\cite{Noiri2022}
%, by using auxiliary qubits to mediate long-range interaction~\cite{Litinski2019gameofsurfacecodes},
or by consuming preshared entanglement resources~\cite{Gottesman1999DQC,PhysRevA.76.062323,PhysRevA.62.052317,PhysRevLett.106.013602,Chou2018}.
However, such quantum communication is costly to realize especially when quantum computers to be connected are further apart from each other.

%The setting of distributed quantum information processing is relevant not only to the cooperation of multiple quantum computers but also critical to building a single quantum computer for performing large-scale quantum information processing.
%As the quantum device becomes larger, it becomes difficult to straightforwardly apply two-qubit gates between the qubits physically located far apart in the device.
%In this case, two-qubit gates in the device can be considered to be resources for such quantum communication, which is costly and thus should be reduced as well.

%いくつかのLOCC taskがあることを紹介する。
%quantum output taskであるところが新しい。
%2nd paragraphでclassical outputであることを言及しておく。
For the cases that quantum communication should be avoided as much as possible, it is important to understand what kind of quantum tasks are still achievable without quantum communication.
The class of quantum operations implementable without quantum communication is called local operations and classical communication (LOCC)~\cite{PhysRevLett.66.1119,horodecki2009quantum,chitambar2014everything} and played a fundamental role in the development of entanglement theory~\cite{PhysRevLett.83.436,PhysRevA.63.022301,horodecki2009quantum}.
%While the LOCC reflects the practical setting, it also plays a fundamental role in the development of entanglement theory, as the nature of entanglement has been revealed through the analysis of the achievability of quantum tasks by LOCC~\cite{PhysRevLett.83.436,PhysRevA.63.022301,horodecki2009quantum}.
There are some cases where distributed quantum tasks are achievable only by LOCC, such as distinguishing quantum states by LOCC~\cite{PhysRevLett.85.4972,doi:10.1080/09500340903477756} and recovering shared classical information by LOCC~\cite{PhysRevA.91.022330,PhysRevA.95.022320,yang2015quantum}.
These tasks aim to extract \textit{classical information} encoded in a quantum state.

%For multiple qubits, multipartite entangled states can be classified into equivalence classes in which states can be transformed into each other by LOCC with nonzero probability~\cite{PhysRevA.63.012307,PhysRevA.62.062314}.
%classができたこと、extractabilityの特徴付けができたことを強調する
%entanglementがいらない個別の例もあるけど...classはわからなかった

In this paper, we analyze a quantum task to extract single-qubit \textit{quantum information} distributed among multiple qubits to one of the qubits when we introduce parties each of whom holds only one qubit and can perform LOCC on each qubit~\cite{Yamasaki2019,Yamasaki2018}.
In this task, single-qubit quantum information is represented by an unknown single-qubit state $\alpha\ket{0}+\beta\ket{1}$ and distributed among multiple qubits as a quantum state $\alpha\ket{\psi_0}+\beta\ket{\psi_1}$, where $\ket{\psi_0}$ and $\ket{\psi_1}$ are mutually orthogonal quantum states of the multiple qubits distributed to the parties.
%In particular, we consider the case where the distributed quantum state $\alpha\ket{\psi_0}+\beta\ket{\psi_1}$ is encoded as a state of the stabilizer code~\cite{gottesman1997stabilizer,calderbank1997quantum} of the multiple qubits.
Extracting quantum information to a party is defined as a decoding transformation from a quantum state $\alpha\ket{\psi_0}+\beta\ket{\psi_1}$ on the multiple qubits to a single-qubit state $\alpha\ket{0}+\beta\ket{1}$ on a qubit at one of the parties.
For example, recovering quantum information in the quantum secret sharing~\cite{cleve1999share,gottesman2000theory,singh2005generalized} and decoding of the quantum codes~\cite{shor1995scheme,steane1996error,RevModPhys.87.307} can be considered as extracting quantum information.

Quantum information encoded in the generic quantum codes cannot be extracted only by LOCC without entanglement resources in general~\cite{Yamasaki2018,Yamasaki2019,yamasaki2019entanglement}.
However, there are several individual examples where extraction is possible only by LOCC.
If the logical basis $\{\ket{\psi_0},\ket{\psi_1}\}$ is represented by particular graph states, it has been already known how to extract quantum information by LOCC as a part of the protocol of quantum secret sharing~\cite{Markham2008,wang2010hierarchical,wang2010hierarchicalQIS,wang2011multiparty}.
There are also cases where extraction is possible without entanglement resources but it requires two-way LOCC~\cite{yamasaki2019entanglement,yamasaki2019spread}.

%ここの表現を改良する
We provide a class of cases where LOCC extraction is possible, rather than individual examples
by proving that single-qubit quantum information encoded in the stabilizer codes can always be extracted to one of the parties only by LOCC.
We construct a protocol to extract single-qubit quantum information to an assigned party with one-way LOCC without entanglement resources when quantum information is encoded in stabilizer codes.
Our algorithm to find such a protocol is based on a graph-theoretical representation of the decoding operations using graph states~\cite{Hein2006}. 
This algorithm works efficiently within a polynomial time in terms of the number of qubits, due to the efficiency of transforming the graphs.
%For a specific subclass of stabilizer codes whose logical basis is represented by particular graph states, it has been already known how to extract quantum information by LOCC as a part of the protocol of quantum secret sharing~\cite{Markham2008,wang2010hierarchical,wang2010hierarchicalQIS,wang2011multiparty}.
%In contrast, our algorithm is applicable to all the stabilizer codes, clarifying when it is possible to extract quantum information to a particular party for all the stabilizer codes and also giving the explicit procedure for extraction whenever possible. Thus, the LOCC extraction protocol can be a fundamental building block of various quantum information processing that requires access to quantum information encoded in the stabilizer code.

We further investigate the minimum number of cooperative parties required for LOCC extraction of quantum information encoded in stabilizer codes.
In general, some of the parties do not have to be cooperative; that is, LOCC involved only by a cooperative subset of the parties can achieve extraction of quantum information, irrespective to local operations performed by the other uncooperative parties.
We find a type of stabilizer codes where the minimum required number of cooperative parties for extraction can be identified.
The number of parties that need to cooperate for extraction depends on the assignment of the party to which the information is extracted.
The difference in the minimum number of parties for extraction implies a hierarchical structure of distributing quantum information.
The larger minimum number of parties indicates quantum information is ``more non-locally distributed'' as it requires the cooperation of more parties. 

This hierarchical structure can be applied to construct a particular type of quantum information splitting (QIS)~\cite{PhysRevA.59.1829,PhysRevA.62.012308,PhysRevA.72.044301,PhysRevLett.98.020503,PhysRevA.74.054303,PhysRevA.78.062333,wang2010hierarchical,wang2011multiparty,wang2010hierarchicalQIS}.
The QIS is a technique for sharing quantum information among multiple parties.
We consider a type of QIS where the parties cannot perform quantum communication, and the quantum information has to be recovered only with LOCC, which we call LOCC-QIS.
In the QIS shown in Refs.~\cite{wang2010hierarchical,wang2010hierarchicalQIS,wang2011multiparty}, the hierarchical structures of authorities to access quantum information are presented.
Our analysis generalizes the hierarchical structure to more complex ones among more parties in the LOCC setting and provides a stepping stone to understanding hierarchical QIS protocols in a unified manner using graphs for the graph states.

The rest of this paper is organized as follows. 
In Sec.~\ref{sec:2}, we introduce definitions of the task of extracting quantum information and present the preliminary materials on quantum information extracting tasks, stabilizer codes, and graph states.
In Sec.~\ref{sec:3}, we present our main result, Theorem~\ref{thm:thm10} for extracting single-qubit quantum information encoded in a stabilizer code on multiple qubits to a specified qubit by LOCC\@.
In Sec.~\ref{sec:4}, we provide applications of the protocol for LOCC-QIS obtained in Theorem~\ref{thm:thm10} to Theorem~\ref{thm:gQSS}.
Our conclusion is given in Sec.~\ref{sec:5}.

\section{settings and Preliminaries}
\label{sec:2}
In this section, we define the task of \textit{extraction of quantum information} and introduce preliminaries. In Sec.~\ref{subsec:2-A} definitions of our task are given. In Sec.~\ref{subsec:2-B} we state useful lemmas of preceding works for the analysis of the task. In Sec.~\ref{subsec:2-C} we summarize a definition of the stabilizer code, by which the quantum information is encoded in our setting; then, we also provide a summary of graph states, which are used in our analysis.

\subsection{Settings}
\label{subsec:2-A}
We introduce settings and notations about multipartite systems.
We consider $n$ parties labeled by an integer $j\in\{1,...,n\}$, and each party has a quantum system represented by a complex Hilbert space.
In this paper, $\mathcal{H}^{(j)}$ denotes the Hilbert space associated with party $j$ and the $n$-partite system shared among the $n$ parties is represented by the tensor product of $\mathcal{H}^{(j)}$, i.e., $\bigotimes_{j=1}^n\mathcal{H}^{(j)}$. 
Each party $j$ only performs local quantum operations on $\mathcal{H}^{(j)}$, and classical communication of the outcome of the measurements on $\mathcal{H}^{(j)}$ between the parties. Quantum operations on the systems held by the other parties and quantum communication between the parties are not allowed.  Such a set of quantum operations is called LOCC (local operations and classical communication)~\cite{chitambar2014everything}.
We consider the cases where $\mathcal{H}^{(j)}$ is a qubit, that is, $\mathcal{H}^{(j)}=\mathbb{C}^2$.
single-qubit quantum information encoded in a $n$-qubit system $\bigotimes_{j=1}^n\mathcal{H}^{(j)}$ is represented as a quantum state of a two-dimensional subspace $\mathcal{H}^{(S)}$ of the $n$-qubit Hilbert space.
We focus on the cases when $\mathcal{H}^{(S)}$ is a stabilizer code space specified by a stabilizer $S$, which will be defined in Sec.~\ref{subsec:2-C}.

Let $\mathcal{H}$ be a finite-dimensional Hilbert space, $\mathcal{B}(\mathcal{H})$ be the set of bounded linear operators on $\mathcal{H}$, and $\mathcal{D}(\mathcal{H})$ be the set of density operators. 
A quantum state is represented by an element in $\mathcal{D}(\mathcal{H})$.
A quantum state $\rho$ in a two-dimensional Hilbert space can be represented by using an orthonormal basis called the computational basis $\{\ket{0},\ket{1}\}$ as $\rho=\sum_{a,b=0,1}\rho_{ab}\ket{a}\bra{b}$.
Let $\{\ket{a}^{(S)}:a=0,1\}$ and $\{\ket{a}^{(j)}:a=0,1\}$ be the bases of $\mathcal{H}^{(S)}$ and $\mathcal{H}^{(j)}$, respectively.  The basis $\{\ket{a}^{(S)}:a=0,1\}$ of the code space is called the logical basis.
Then extraction of quantum information to party $j$ is defined as the task to transform $\rho^{(S)}=\sum_{a,b=0,1}\rho_{ab}\ket{a}\bra{b}^{(S)}$ shared over $n$ parties to $\rho^{(j)}=\sum_{a,b=0,1}\rho_{ab}\ket{a}\bra{b}^{(j)}$ of party $j$ by performing a quantum operation over the multipartite system. 
This quantum operation is represented by a completely positive and trace-preserving (CPTP) map from $\mathcal{B}(\bigotimes_{k=1}^n\mathcal{H}^{(k)})$ to $\mathcal{B}(\mathcal{H}^{(j)})$~\cite{nielsen2002quantum}.
By extraction of quantum information, the quantum information spread over the multipartite system is localized to party $j$, and the quantum information becomes accessible by local operations on $\mathcal{H}^{(j)}$.

Our interest is the condition when extraction of quantum information is achievable by LOCC\@.
In the following, a CPTP map that represents LOCC is called an LOCC map.
Now, we define the task of extracting quantum information to party $j$ by LOCC as follows.
\begin{defi}[LOCC extraction of quantum information]
\label{defi:1}
Given a party $j\in\{1,\ldots,n\}$,
let $U:\bigotimes_{k=1}^n\mathcal{H}^{(k)}\to\mathcal{H}^{(j)}$ be an isometry map which transforms the logical basis $\{\ket{a}^{(S)}:a=0,1\}$ of $\mathcal{H}^{(S)}$ to the computational basis $\{\ket{a}^{(j)}:a=0,1\}$ of $\mathcal{H}^{(j)}$.
LOCC extraction of quantum information to party $j$ is a task represented by an LOCC map $\mathcal{C}^{(S\to j)}$ which satisfies, for any quantum state $\rho^{(S)}\in\mathcal{D}(\mathcal{H}^{(S)})$,

\begin{equation}
\label{eq:2}
    \mathcal{C}^{(S\to j)}(\rho^{(S)})=U\rho^{(S)} U^\dagger\in\mathcal{D}(\mathcal{H}^{(j)}).
\end{equation}

\end{defi}
We give an example of LOCC extraction of single-qubit quantum information encoded in two qubits in Appendix~\ref{app:Ex1}.

\subsection{Preliminaries of extracting tasks}
\label{subsec:2-B}
To complete extraction of quantum information, we have to explicitly construct the LOCC map $\mathcal{C}$ that extracts an arbitrary input state $\rho^{(S)}\in\mathcal{D}(\mathcal{H}^{(S)})$. 
In this section, we describe how to simplify this construction of the LOCC map based on Ref.~\cite{Yamasaki2018}.

By purification of quantum states~\cite{nielsen2002quantum}, any quantum state $\rho\in\mathcal{D}(\mathcal{H})$ is represented by a pure state in a composite system represented by $\mathcal{H}^{(R)}\otimes \mathcal{H}$, where $\mathcal{H}^{(R)}$ represents a reference system and $\dim\mathcal{H}^{(R)}=\dim\mathcal{H}$.
The introduction of the reference system is useful for analyzing extraction of quantum information.
The following lemma shows that extraction of quantum information to party $j$ is equivalent to a task to transform the maximally entangled state between $\mathcal{H}^{(R)}$ and $\mathcal{H}^{(S)}$ to the maximally entangled state between $\mathcal{H}^{(R)}$ and $\mathcal{H}^{(j)}$.
Here, the maximally entangled state between $\mathcal{H}^{(R)}$ and $\mathcal{H}^{(S)}$ is
\begin{equation}
\label{eq:RS}
\ket{\Phi}^{(RS)}\coloneqq(\ket{0}^{(R)}\ket{0}^{(S)}+\ket{1}^{(R)}\ket{1}^{(S)})/\sqrt{2}
\end{equation}
and that between $\mathcal{H}^{(R)}$ and $\mathcal{H}^{(j)}$
\begin{equation}
\label{eq:Rk}
    \ket{\Phi}^{(Rj)}\coloneqq(\ket{0}^{(R)}\ket{0}^{(j)}+\ket{1}^{(R)}\ket{1}^{(j)})/\sqrt{2}.
\end{equation}
Due to the following lemma, it suffices to construct a map given by
\begin{equation}
\mathrm{id}^{(R)}\otimes \mathcal{C}^{(S\to j)}:~\ket{\Phi}\bra{\Phi}^{(RS)} \mapsto\ket{\Phi}\bra{\Phi}^{(Rj)},
\end{equation}
to seek an LOCC map extracting quantum information to party $j$. 
\begin{lem}State transformation equivalent to extraction of quantum information, a corollary of Proposition~3 in Ref. ~\cite{Yamasaki2018}: 
\label{lem:1}
LOCC extraction of quantum information to party $j$ defined as Definition~\ref{defi:1} is achievable if and only if we have an LOCC map $\mathcal{C}^{(S\to j)}$satisfying
\begin{equation}
    \mathrm{id}^{(R)}\otimes \mathcal{C} ^{(S\to j)}\left(\ket{\Phi}\bra{\Phi}^{(RS)}\right)=\ket{\Phi}\bra{\Phi}^{(Rj)},
\end{equation}
where $\mathrm{id}^{(R)}$ denotes the identity map for operators on $\mathcal{H}^{(R)}$.
\end{lem}
We give an example of the state transformation equivalent to the extraction of single-qubit quantum information encoded in two qubits in Appendix~\ref{app:Ex2}.

\subsection{Preliminaries for Stabilizer Code and Graph State}
\label{subsec:2-C}
We introduce the stabilizer code and the graph state in this subsection.
In the following, the Pauli operators and the identity on a qubit are denoted by $X, Y, Z$, and $I$, respectively. Their matrix representations in the basis $\{\ket{0},\ket{1}\}$ are,
\begin{equation}
X=\mqty(\pmat{1}),~Y=\mqty(\pmat{2}),~Z=\mqty(\pmat{3})~\text{and}~I=\mqty(1&0\\0&1).
\end{equation}
The Hadamard gate is denoted by $H$, which satisfies $H\ket{0}=\ket{+}$ and $H\ket{1}=\ket{-}$ where $\ket{\pm}$ is defined by $\ket{\pm} = \frac{1}{\sqrt{2}}(\ket{0} \pm \ket{1})$.

A stabilizer is defined as an abelian subgroup of a Pauli group, where the Pauli group on $n$ qubits is defined as the group generated by Pauli operators on each qubit.
For any stabilizer $S$ on $n$ qubits, there exists an element $\ket{\psi}$ in $\bigotimes_{j=1}^n\mathcal{H}^{(j)}$ which is stabilized by $S$, i.e., $s\ket{\psi}=\ket{\psi}$ for all $s\in S$.
If a stabilizer $S$ is generated by $n-m$ independent operators, the set of elements stabilized by $S$ is a $2^{m}$-dimensional linear subspace of $\bigotimes_{j=1}^n\mathcal{H}^{(j)}$ and called a stabilizer code. 
If the stabilizer $S$ on $n$ qubits is generated by $n$ independent generators, the pure state $\ket{S}$ stabilized by $S$ is uniquely determined since the subspace stabilized by $S$ is one-dimensional.
We call this pure state $\ket{S}$ the stabilizer state stabilized by $S$.

The graph state~\cite{Hein2006} can be used to analyze tasks of extracting quantum information encoded in the stabilizer code.
For a given graph $G=(V, E)$, the graph state $\ket{G}$ is defined as follows assuming $G$ to be a simple graph.
Let $G=(V,E)$ be a graph with $|V|=n$ and qubits $\mathcal{H}^{(j)}~(j=1,\cdots ,n)$ be associated with each vertex $j$ in $V=\{j\}_{j=1}^n$. 
We use the same index $j$ to represent the party, the qubit, and the vertex.
Graph state $\ket{G}$ is a pure state of $n$ qubits represented by
\begin{equation}
    \ket{G}=\prod_{\{j,k\}\in E} CZ^{(jk)}\bigotimes_{l=1}^n\ket{+}^{(l)},
\end{equation}
where %$\ket{+}^{(j)}=(\ket{0}^{(j)}+\ket{1}^{(j)})/\sqrt{2}$, and 
$CZ^{(jk)}$ is the controlled $Z$ gate on two qubits $j$ and $k$ acting as
\begin{equation}
    CZ^{(jk)}\ket{a}^{(j)}\ket{b}^{(k)}=
    \begin{cases}
-\ket{a}^{(j)}\ket{b}^{(k)}&(\text{if }a=b=1),\\
\ket{a}^{(j)}\ket{b}^{(k)}&(\text{otherwise}).
    \end{cases}
\end{equation}
The graph state $\ket{G}$ can also be defined as the unique pure state of $\bigotimes_{j=1}^n\mathcal{H}^{(j)}$ fixed by the stabilizer generated by $n$ operators
\begin{equation}
\label{eq:graph_state}
    S_j=X^{(j)}\otimes\bigotimes_{k\in N_j}Z^{(k)},~(j=1,\cdots,n),
\end{equation}
where a superscript of a Pauli operator corresponds to the superscript of the Hilbert space to which each Pauli operator acts, and $N_j$ is the set of vertices $k$ adjacent to $j$, i.e., those satisfying $\{k,j\}\in E$.
Therefore, a graph state is a stabilizer state. 
In addition, any stabilizer state can be transformed to a graph state by a local Clifford (LC) operation, which is a local unitary operation transforming a Pauli operator to a Pauli operator~\cite{Hein2006}.
We say two quantum states $\rho$ and $\sigma$ of a multiqubit system are \textit{LC equivalent} if there exists a local Clifford operation $\mathcal{E}$ which satisfies $\mathcal{E}(\rho)=\sigma$.
\begin{prop} LC equivalence between a stabilizer state and a graph state~\cite{Hein2006}:
\label{prop:2}
For any stabilizer state $\ket{S}$ of $n$ qubits, there exists a graph state $\ket{G}$ that is LC equivalent to $\ket{S}$, i.e., $\ket{S}\bra{S} = \mathcal{E}(\ket{G}\bra{G})$ for some LC operation $\mathcal{E}$. This LC operation and the graph state $\ket{G}$ can be calculated in $\order{n^2}$ time.
\end{prop}

In LOCC extraction of quantum information, a projective measurement in the eigenbasis of a Pauli operator on a party called a local Pauli measurement (LPM), plays an important role as will be shown in Sec.~\ref{sec:3}.
The LPMs on quantum states of an $n$-qubit stabilizer code can be classically simulated in a polynomial time in terms of $n$. This fact is known as the Gottesman-Knill theorem~\cite{gottesman1998heisenberg, Aaronson2004}.
One could use this property for seeking efficient extraction of quantum information encoded in a stabilizer code, but in this paper, we use the fact that the stabilizer state is LC equivalent to the graph state and that LPMs can be efficiently simulated by representing LPMs as a transformation of the graph such as eliminating vertices and local complementations~\cite{Hein2006}.

The graph obtained by eliminating vertex $j$ from graph $G=(V,E)$ is represented by $G-j=(V',E')$, where
\begin{equation}
    V'=V\setminus\{j\},~E'=E\setminus\{\{j,k\}:\{j,k\}\in E\}.
\end{equation}
The graph obtained by the local complementation on vertex $j$ of the graph $G=(V,E)$ is represented by $\tau_jG=(V,E')$, where
\begin{equation}
\begin{split}
     E'=&(E\setminus\{\{k,l\}:\{k,j\},\{l,j\}\in E,\{k,l\}\in E\})\\
     &\cup \{\{k,l\}:\{k,j\},\{l,j\}\in E,\{k,l\}\notin E\}.
\end{split}
\end{equation}
\begin{figure}
    \centering
    \includegraphics[keepaspectratio, scale=0.11, angle=0]{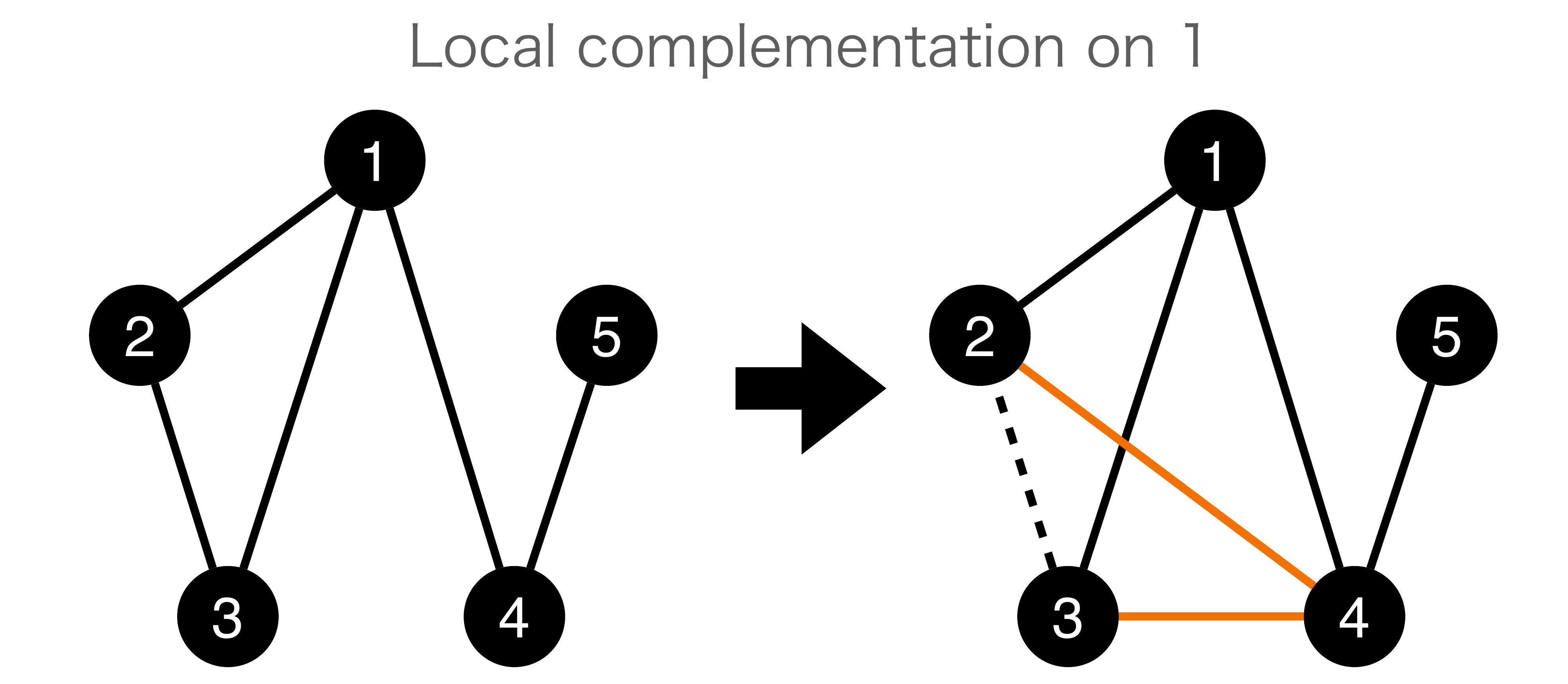}
    \caption{A graphical description of the local complementation. The orange lines represent the edges added after the local complementation on vertex $1$. The dashed line represents the eliminated edge after the local complementation.}
    \label{graph:Lcomp}
\end{figure}We give a graphical description of the local complementation in Fig.~\ref{graph:Lcomp}.

The LPMs on a graph state are represented as follows.
Let $|X , \pm\rangle^{(j)},~|Y, \pm\rangle^{(j)}$ and $|Z, \pm\rangle^{(j)}$ be the eigenstates of $X^{(j)},Y^{(j)}$ and $Z^{(j)}$ with eigenvalues $\pm1$, respectively, and $P_{X, \pm}^{(j)},~P_{Y, \pm}^{(j)}$ and $P_{Z, \pm}^{(j)}$ be the projection operators to the eigenstates $|X, \pm\rangle^{(j)},~|Y, \pm\rangle^{(j)}$ and $|Z, \pm\rangle^{(j)}$, respectively.
In particular, the measurement in $\{\ket{0}=|Z, + \rangle,\ket{1} = |Z, - \rangle  \}$ on a qubit of the graph states is simply represented by eliminating the vertex as shown in Fig.~\ref{graph:Zmeas}, and a measurement in $\{\ket{\pm}=|X, \pm\rangle \}$ on line-topology graph states is represented as shown in Fig.~\ref{graph:Xmeas}. 

\begin{prop}LPM on a graph state, Proposition 7 in Ref.~\cite{Hein2006} :
\label{prop:graph}
A measurement in the eigenbasis of $X, Y, Z$ on the qubit associated with vertex $j$ transforms a graph $G$ to a new graph state $\ket{G'}$ up to local unitaries $U_{j,A,\pm}~(A=X,Y,Z)$. 
Let $N_j$ be the neighborhood of $j$, and $\tau_j$ be the local complementation on vertex $j$.
The state after a measurement in each basis is given by
\begin{align}
    P_{Z, \pm}^{(j)}|G\rangle=&\frac{1}{\sqrt{2}}|Z, \pm\rangle^{(j)} \otimes U_{j,Z, \pm}|G-j\rangle,\\
    P_{Y, \pm}^{(j)}|G\rangle=&\frac{1}{\sqrt{2}}|Y, \pm\rangle^{(j)} \otimes U_{j,Y,\pm}\left|\tau_{j}(G)-j\right\rangle,\\
    \begin{split}
    P_{X, \pm}^{(j)}|G\rangle=&\frac{1}{\sqrt{2}}|X, \pm\rangle^{(j)} \\
    &\otimes U_{j,X, \pm}\left|\tau_{k}\left(\tau_{j} \circ \tau_{k}(G)-j\right)\right\rangle,
    \end{split}
\end{align}
where vertex $k$ can be any one of the elements of $N_j$, and the local unitaries $U_{j,A,\pm}$ are given by
\begin{align}
    U_{j,Z,+}&=1,\\
    U_{j,Z,-}&=Z^{(N_{j})}\coloneqq \bigotimes_{l\in N_j}Z^{(l)}\\
    U_{j,Y,+}&=\bigotimes_{l\in N_j}\sqrt{-i Z^{(l)}},\\
    U_{j,Y,-}&=\bigotimes_{l\in N_j}\sqrt{+i Z^{(l)}},\\
    U_{j,X,+}&=\sqrt{+i Y}^{(k)}\otimes\bigotimes_{l\in N_{j} \setminus\left(N_{k}\cup k\right)} Z^{(l)} ,\\
    U_{j,X,-}&=\sqrt{-i Y}^{(k)}\otimes\bigotimes_{l\in N_{k} \setminus\left(N_{j}\cup j\right)} Z^{(l)}.
\end{align}
If $N_j$ is an empty set, the obtained state after the measurement is always $\ket{G-j}$.
\end{prop}
\begin{figure}
    \centering
    \includegraphics[keepaspectratio, scale=0.12, angle=0]{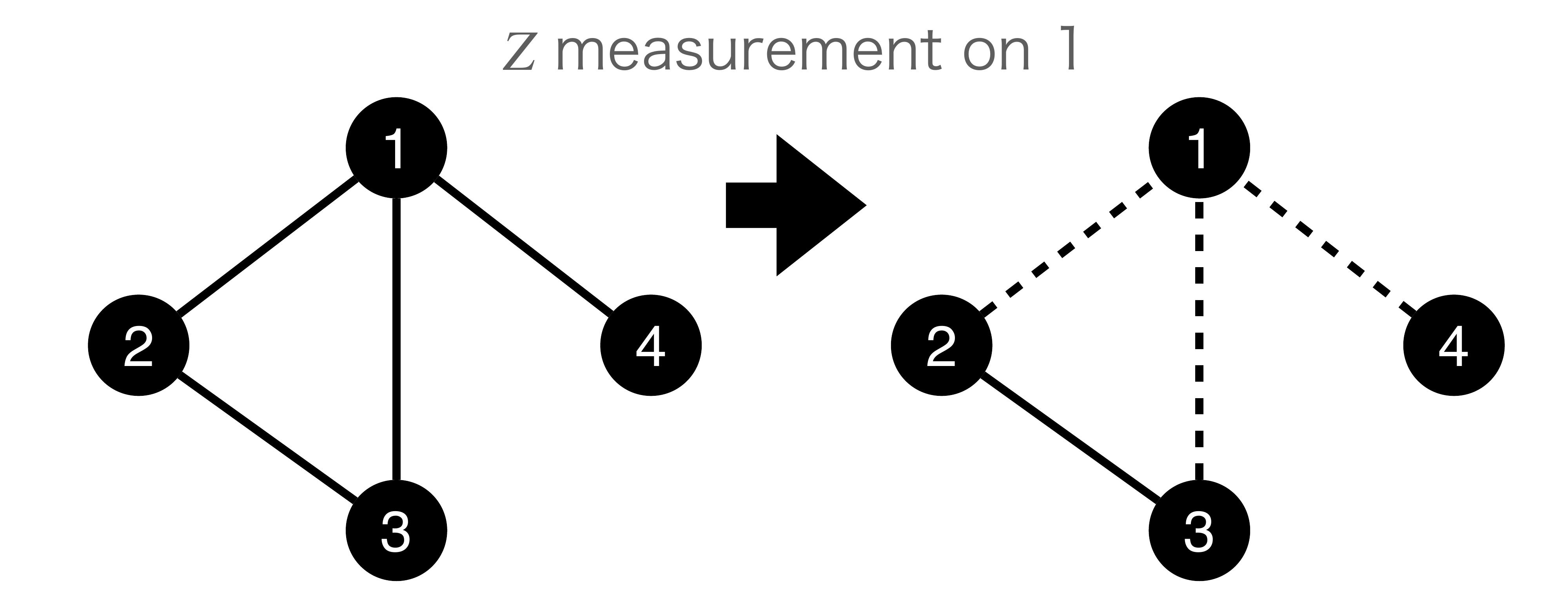}
    \caption{A measurement in $\{\ket{0},\ket{1}\}$ on a qubit on the graph state is represented by the elimination of the vertex corresponding to the measured qubit. The transformation of the graph represents the measurement in $\{\ket{0}^{(1)},\ket{1}^{(1)}\}$ on qubit $1$ and the dashed lines represent eliminated edges.}
    \label{graph:Zmeas}
\end{figure}
\begin{figure}
    \centering
    \includegraphics[keepaspectratio, scale=0.14, angle=0]{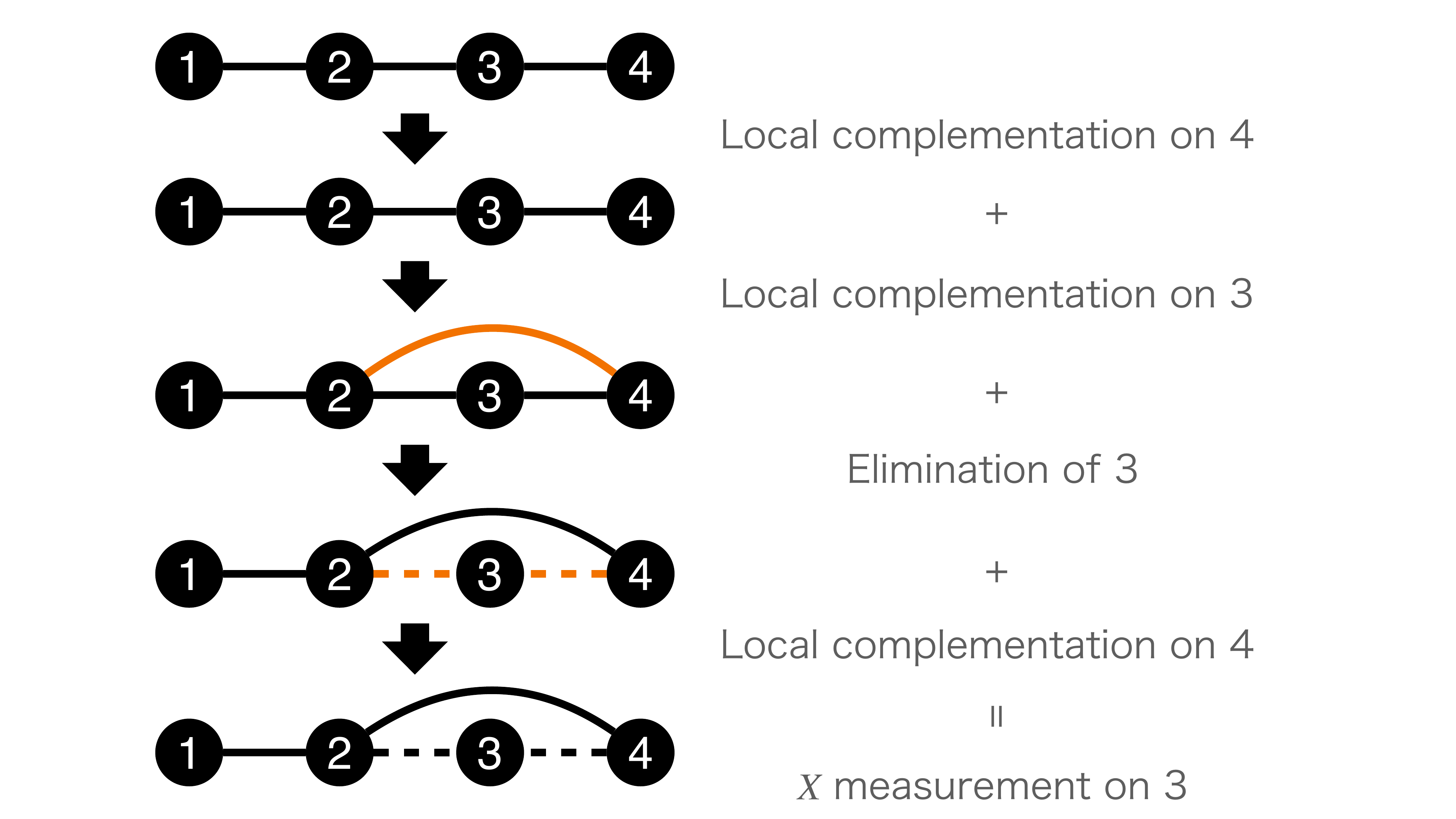}
    \caption{Measurement in $\{\ket{\pm}\}$ on a qubit of a line-topology graph state. These sequential transformations of graphs represent the measurement in $\{\ket{\pm}^{(3)}\}$ on qubit $3$ of the graph state corresponding to the line-topology graph with 4 vertices.
    The last graph at the bottom represents the state after the measurement in $\{\ket{\pm}^{(3)}\}$ on qubit $3$ up to the local Clifford operation.
    According to Proposition~\ref{prop:graph}, we first perform the local complementation on vertex $4$, which does not change the graph. Then, we perform the local complementation on vertex $3$, and the edge represented by the curved orange line is added as shown in the third graph. Finally, we eliminate vertex $3$, perform the local complementation on vertex 4, and obtain the graph in the last line in the figure. The dashed lines represent edges eliminated after eliminating vertex $3$.}
    \label{graph:Xmeas}
\end{figure}

Here, we give a corollary of Proposition~\ref{prop:graph} about a measurement in $\{\ket{\pm}\}$ on a line-topology graph states, which is illustrated in Fig.~\ref{graph:Xmeas} and is useful for describing our main result.
\begin{cor} Measurement in $\{\ket{\pm}\}$ on a qubit of a line-topology graph state:
\label{cor:linegraph}
Let $G=(V,E)$ be a line-topology graph given by $V=\{j\}_{j=1}^n$ and $E=\{\{j,{j+1}\}\}_{j=1}^{n-1}$. 
A measurement in $\ket{\pm}^{(j)}$ on $j$ maps $\ket{G}$ to $\ket{\pm}\otimes U_{j,X,\pm}\ket{G'}$, where $U_{j,X,\pm}$ is the local unitary defined as Proposition~\ref{prop:graph} and $G'=(V\setminus\{j\},(E\cup\{\{j-1,j+1\}\})\setminus\{\{j-1,j\},\{j,j+1\}\})$. 
\end{cor}

\section{LOCC extraction of stabilizer codes}
\label{sec:3}
In this section, we show the necessary and sufficient condition for extracting single-qubit quantum information encoded in a stabilizer code to a fixed party $j$ by LOCC\@.
This condition is that the code space $\mathcal{H}^{(S)}$ of the stabilizer code \textit{cannot} be expressed as
\begin{equation}
   \label{eq:condi} \mathcal{H}^{(S)}=\mathcal{H}^{(\bar{A})}\otimes\ket{\psi}^{(A)},
\end{equation}
where $A$ is a set of parties including $j$, $\bar{A}$ is the complement of $A$, $\mathcal{H}^{(\bar{A})}$ is a subspace of $\bigotimes_{k\in \bar{A}}\mathcal{H}^{(k)}$, and $\ket{\psi}^{(A)}$ is a state of $\bigotimes_{k\in {A}}\mathcal{H}^{(k)}$.
In this case, the quantum state of qubit $j$ is independent of the encoded quantum information. Therefore, the quantum information is not distributed to party $j$ at all.
It is interesting to note that quantum information shared among multiple qubits using any stabilizer code can always be extracted to some party by LOCC. This is in contrast to general cases (quantum information shared using a non-stabilizer code), as there are known examples in which extraction is not possible by LOCC in general~\cite{Yamasaki2019} and may require two-way LOCC~\cite{yamasaki2019spread}.
For stabilizer codes, there always exists a party included in $\bar{A}$, which satisfies the necessary and sufficient condition for LOCC extraction.  Thus it is always possible to extract quantum information encoded in the stabilizer code using LOCC to some party.

There are two key elements for the proof of the sufficient and necessary condition.
First, as shown in Lemma~\ref{lem:1}, to prove that arbitrary quantum information can be extracted, we will show that the maximum entangled state $\ket{\Phi}^{(RS)}$ of the reference system and code space can be mapped to the maximum entangled state $\ket{\Phi}^{(Rj)}$ of the reference system and qubit $j$ by LOCC without any operations on the reference system.
Second, using the fact that the code space is a stabilizer code, we will regard $\ket{\Phi}^{(RS)}$ as a graph state up to the LC equivalence; then, the measurement in the eigenbasis of the Pauli operators and the LC operations associated with the outcomes of the measurement can be regarded as a transformation of the corresponding graphs.

\begin{thm} 
\label{thm:thm10}
Let $S$ be a given stabilizer on $n$ qubits, $\bigotimes_{k=1}^{n}\mathcal{H}^{(k)}$, generated by $n-1$ independent generators and $\mathcal{H}^{(S)}$ be a stabilizer code stabilized by $S$.
Given a fixed party $j$, there exists an LOCC map which performs extraction of single-qubit quantum information encoded in $\mathcal{H}^{(S)}$ to party $j$, if and only if there exists no subset $A\subset\{1,\cdots,n\}$ which includes $j$ such that each state $\ket{a}^{(S)}~(a=0,1)$ in the logical orthonormal basis $\{\ket{a}^{(S)}:a=0,1\}$ of $\mathcal{H}^{(S)}$ is represented by
\begin{equation}
\label{eq:thm8}
    \ket{a}^{(S)}=\ket{a}^{(\bar{A})}\otimes\ket{\psi}^{(A)},
\end{equation}
where $\{\ket{a}^{(\bar{A})}\}$ is a set of mutually orthogonal states of $\bigotimes_{k\notin A}\mathcal{H}^{(k)}$, and $\ket{\psi}^{(A)}$ is a fixed pure state of $\bigotimes_{k\in A}\mathcal{H}^{(k)}$.
\end{thm}

\begin{proof}
   The LOCC extraction map for a given stabilizer $S$ can be obtained according to Algorithm~\ref{alg2}. See Appendix~\ref{app:thm5} for the details.
\end{proof}

Theorem~\ref{thm:thm10} is proved by reducing extraction of quantum information to the problem of considering a graph corresponding to the graph state LC equivalent to $\ket{\Phi}^{(RS)}$ and transforming this graph into a connected two-vertex graph representing the reference system and party $j$.
The transformation of a graph is composed of (i) eliminating the vertices not included in the path from vertex $R$ to vertex $j$, and (ii) the local complementations on the remaining vertices except for vertex $R$ and vertex $j$.
This graph transformation is implemented by LPMs and LC operations as shown in Algorithm~\ref{alg2}.
We note that the parties who need to cooperate in Algorithm~\ref{alg2} are only those who are in the path or adjacent to the path; therefore, not all parties need to cooperate on the LOCC extraction of quantum information.

\begin{figure*}
\begin{minipage}{\linewidth}
\begin{algorithm}[H]                      
\caption{An algorithm for LOCC extraction of quantum information to a fixed party}
\label{alg2}                
\begin{algorithmic}[1]               
\Require $\text{$S$: Stabilizer of the stabilizer code, $j$: the party to which information is extracted}$, $\rho^{(S)}\in\mathcal{D}(\mathcal{H}^{(S)})$: the input state in Definition~\ref{defi:1}
\State Calculate logical $X$ and $Z$ operators so as to obtain a classical description of $\ket{\Phi}^{(RS)}$ in terms of the stabilizer state.
\Comment{This can be done by using the standard form of the stabilizer code~\cite{gottesman1997stabilizer}. This can be calculated in $\order{n^3}$ time.}
\State Find a graph state $\ket{G}$ which is LC equivalent to $\ket{\Phi}^{(RS)}$ defined in Lemma~\ref{lem:1} and $\{ U^{(k)}_k \}$ satisfying $\ket{G}=\bigotimes_{k=R,1,\cdots,n}U^{(k)}_k\ket{\Phi}^{(RS)}$.
\Comment{The algorithm for finding a graph $G$ from a stabilizer state is given in Ref.~\cite{PhysRevA.69.022316}.}
\State Find path $P=(V_P,E_P)$ from vertex $R$ to vertex $j$ in graph $G$.
(If such a path does not exist, LOCC extraction of quantum information to $j$ is impossible.)
\Comment{This can be done in $\order{n+|E|}$ time by the breadth-first search~\cite{10009005306}.}
\State $\text{Each party $k$ who is in $V_P$ or adjacent to some vertex in $V_P$ applies $U_k^{(k)}$ to their local qubit.}$
\State $\text{Initialize $m_l=0$ for $l=1,2,\cdots,|V_P|$.}$
\Comment{These correspond to the vertices $v_l$ of $P$, where $v_1=j$, $v_{|V_P|}=R$ and $\{v_{l},v_{l+1}\}$'s are the edges of $P$.}
\For {$\text{$v$ in the set of the vertices adjacent to a vertex of $V_P$ but not included in $V_P$}$}
\State$\text{Perform measurement in $\{\ket{0}^{(v)},\ket{1}^{(v)}\}$ on party $v$.}$
\If {$\text{the outcome is $\ket{1}^{(v)}$}$}
\For {$l=1,2,\cdots,|V_P|$}
\If {$\text{$v$ is adjacent to $v_l$}$}
\State $m_l \gets m_l+1$
\EndIf
\EndFor
\EndIf
\EndFor
\For{$l=1,\cdots,|V_P|-1$}
\State $\text{The party $v_l$ performs $Z^{m_l}$.}$
\EndFor
\For{$l=2,\cdots,|V_P|-2$}
\State $\text{Party $v_l$ performs measurement in $\{\ket{\pm}^{(v_l)}\}$}$
\If {$\text{the outcome is $\ket{+}^{(v_l)}$}$}
\State Party $j$ performs $(iY^{(j)})^{-1/2}$.
\State Party $v_{l+1}$ performs $Z^{(v_{l+1})}$.
\Else \Comment{$\text{The outcome is $\ket{-}^{(v_l)}$}$}
\State Party $j$ performs $(-iY^{(j)})^{-1/2}$.
\EndIf
\EndFor
\State $\text{Party $v_{|V_P|-1}$ performs measurement in $\{\ket{\pm}^{(v_{|V_P|-1})}\}$}$
\If {$\text{the outcome is $\ket{+}^{(v_{|V_P|-1})}$}$}
\State $\text{Party $j$ performs $HZ^{m_{|V_P|} +1}(U_R^\dagger)^{T}(iY^{(j)})^{-1/2}$ on qubit $j$.}$
\Else \Comment{$\text{The outcome is $\ket{-}$}$}
\State $\text{Party $j$ performs $HZ^{m_{|V_P|}}(U_R^\dagger)^{T}(-iY^{(j)})^{-1/2}$ on qubit $j$.}$
\EndIf
\end{algorithmic}
\end{algorithm}
\end{minipage}
\end{figure*}

To see how Algorithm~\ref{alg2} works, we give an example of the five-qubit code~\cite{cleve1999share,gottesman2000theory} and  graphically explain the process of extraction of quantum information using graphs.

\begin{exa}[Five-qubit code]
\label{exa:5qcode}
We extract quantum information encoded in the five-qubit code~\cite{divincenzo1996fault}, which is known to be able to protect against an arbitrary single-qubit error.

Let us consider the task of extracting quantum information to party $1$ according to Algorithm~\ref{alg2}.
The stabilizer $S$ of the five-qubit code is generated by $4$ generators
\begin{equation}
\begin{split}
    X^{(1)}Z^{(2)}Z^{(3)}X^{(4)}I^{(5)},\\
    I^{(1)}X^{(2)}Z^{(3)}Z^{(4)}X^{(5)},\\
    X^{(1)}I^{(2)}X^{(3)}Z^{(4)}Z^{(5)},\\
    Z^{(1)}X^{(2)}I^{(3)}X^{(4)}Z^{(5)},
\end{split}
\end{equation}
where the tensor product symbol $\otimes$ is omitted for brevity.

The graph $G$ corresponding to the five-qubit code is given in Fig.~\ref{graph:5qubit}, which can be calculated according to Appendix~\ref{app:Ex3}.
The maximally entangled state between the stabilizer code and the reference is LC equivalent to $\ket{G}$ as
\begin{equation}
    \ket{\Phi}^{(RS)}=H^{(R)}H^{(4)}H^{(5)}\ket{G},
\end{equation}
where
\begin{equation}
\begin{split}
    \ket{0}^{(S)}=&\frac{1}{4}(\ket{00000}+\ket{11000}+\ket{01100}+\ket{00110}+\ket{00011}\\
    &+\ket{10001}-\ket{10100}-\ket{01010}-\ket{00101}-\ket{10010}\\
    &-\ket{01001}-\ket{11110}-\ket{01111}-\ket{10111}-\ket{11011}\\
    &-\ket{11101})\\
    \ket{1}^{(S)}=&\frac{1}{4}(\ket{11111}+\ket{00111}+\ket{10011}+\ket{11001}+\ket{11100}\\
    &+\ket{01110}-\ket{01011}-\ket{10101}-\ket{11010}-\ket{01101}\\
    &-\ket{10110}-\ket{00001}-\ket{10000}-\ket{01000}-\ket{00100}\\
    &-\ket{00010}).\\
\end{split}
\end{equation}

Then, the parties perform the sequence of operations to extract quantum information to party 1, according to Algorithm~\ref{alg2}.
This process is schematically illustrated in Fig.~\ref{graph:5qubit}, \ref{graph:5qubitb}, and \ref{graph:5qubitc}.

\begin{figure}
    \centering
    \includegraphics[keepaspectratio, scale=0.18, angle=0]{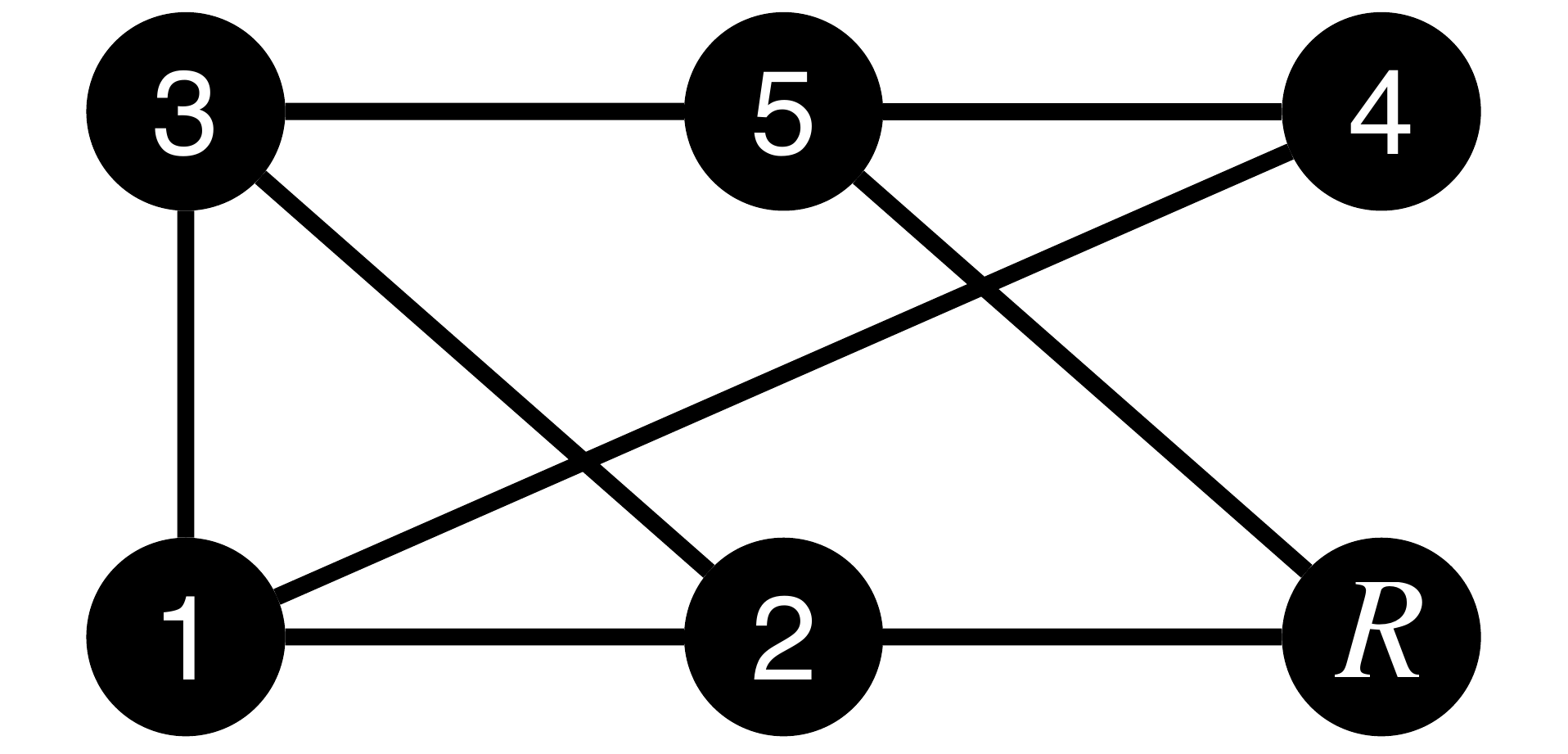}
    \caption{The graph corresponds to the graph state which is LC equivalent to the stabilizer state $\ket{\Phi}^{(RS)}$ of the five-qubit code.}
    \label{graph:5qubit}
\end{figure}
\begin{figure}
    \centering
    \includegraphics[keepaspectratio, scale=0.18, angle=0]{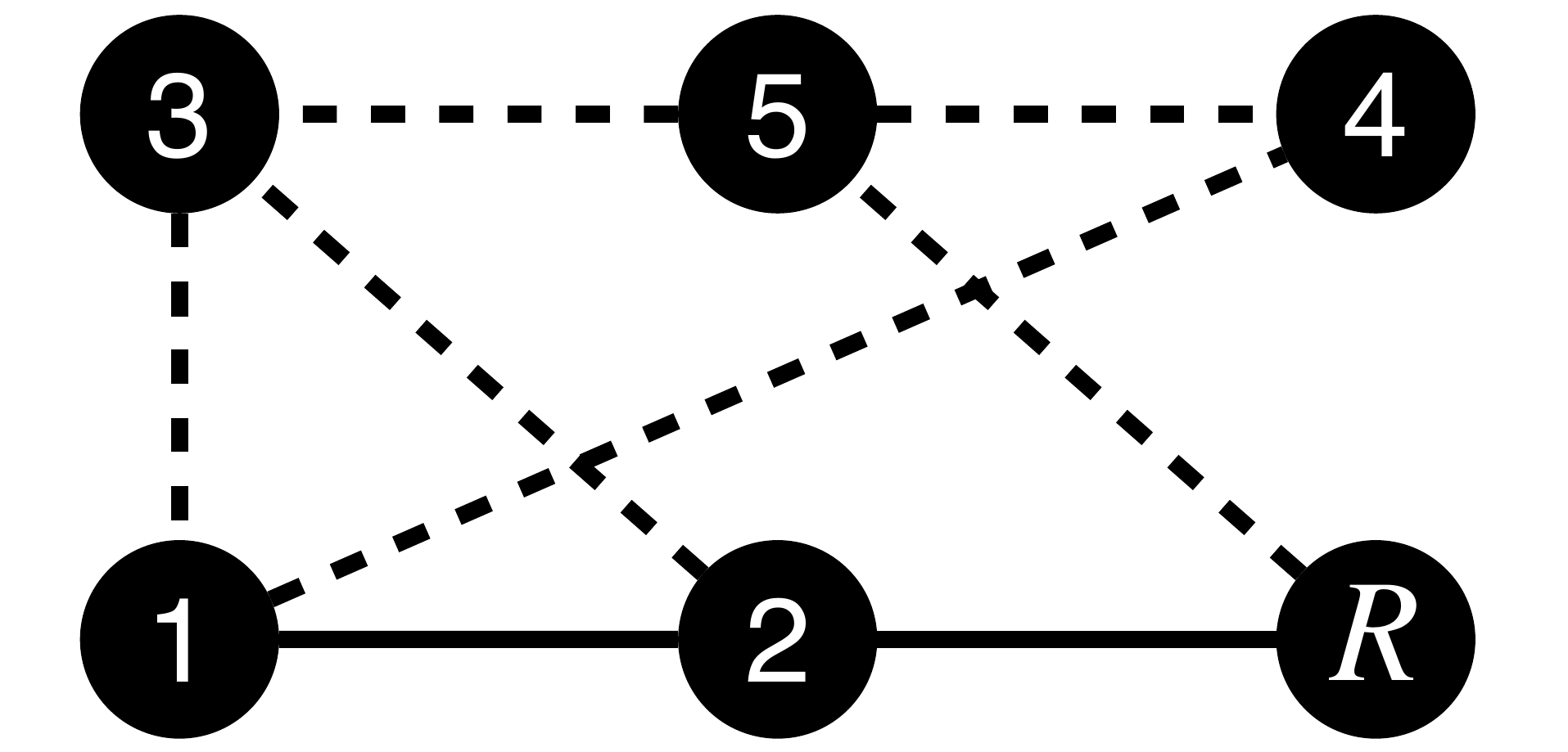}
    \caption{The graph transformation corresponding to the measurements on 3,4, and 5.
    The dashed lines represent the eliminated edges after the measurement.
    Find a path $P$ from 1 to $R$ via 2 in the graph $G_\text{5qubit}$. Parties 3, 4, and 5 perform measurements in $\{\ket{0},\ket{1}\}$ and communicate their outcomes to party 1 and 2.
    Parties 1 and 2 perform the local unitary operation based on the outcomes.
    The entire process above is represented by the elimination of vertices 3, 4, and 5.}
    \label{graph:5qubitb}
\end{figure}
\begin{figure}
    \centering
    \includegraphics[keepaspectratio, scale=0.18, angle=0]{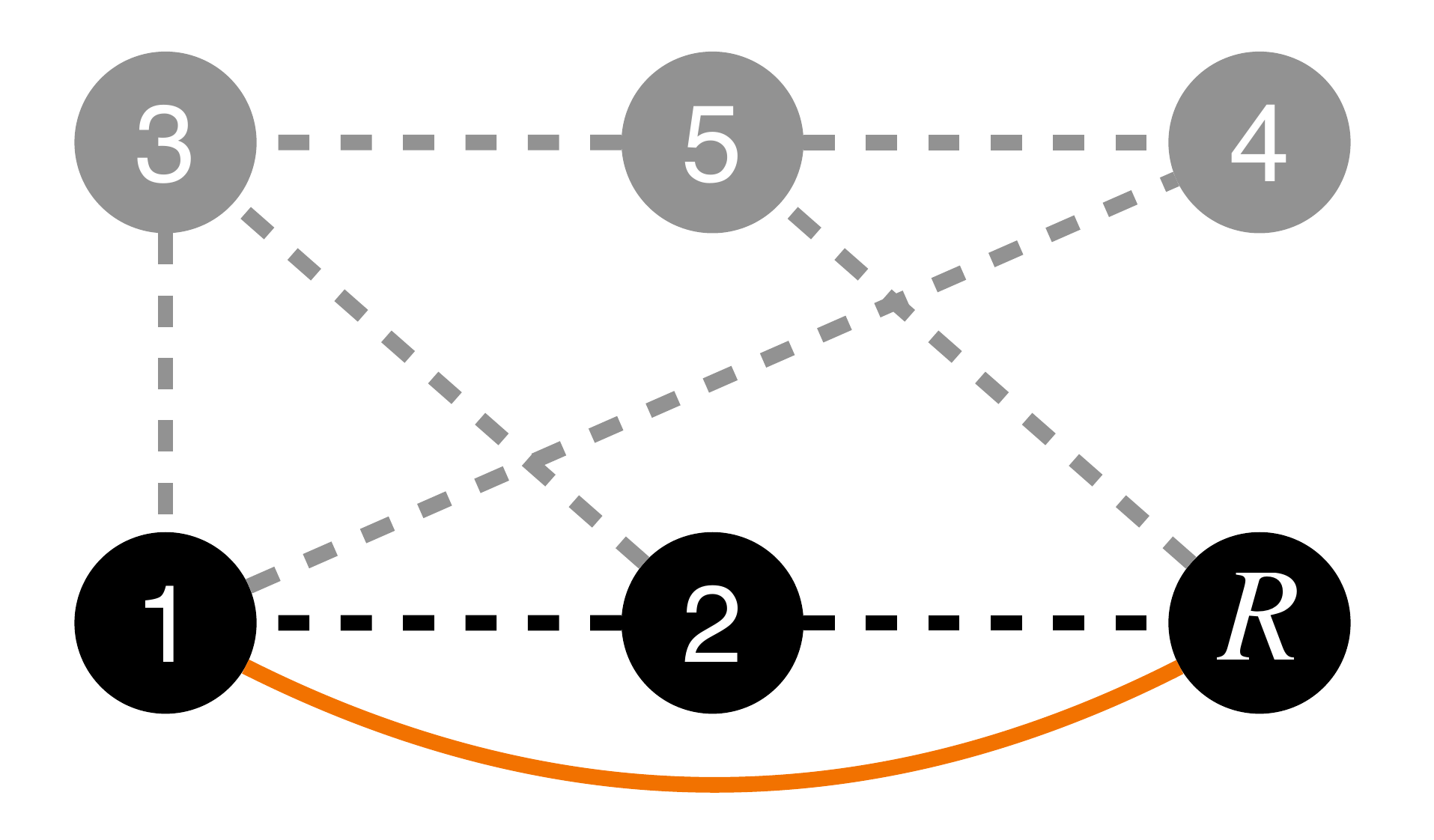}
    \caption{The graph transformation corresponding to the measurement on 2.
    The gray objects have been already eliminated in Fig.~\ref{graph:5qubitb}.
    The dashed lines represent the eliminated edges after the measurement on 2, and the orange line represents the added edge after the measurement on 2.
    Party $2$ performs the measurement in $\{\ket{\pm}\}$ and communicates its outcome to party 1.
     This process is represented by the local complementation on vertex 2. The graph state of the final graph in this figure is LC equivalent to $\ket{\Phi}^{(R1)}$. Therefore, party 1 performs a local unitary operation based on the outcomes, and the quantum information is extracted.}
    \label{graph:5qubitc}
\end{figure}

We first find a path in the graph from the reference party to a party where information is extracted.
We can find the path from vertex 1 to $R$ via vertex 2 and vertices 3, 4, and 5 are connected to the path.
Since parties in the path and adjacent to the path need to cooperate, all the parties need to cooperate to extract quantum information to party 1 in this case.
After parties 4 and 5 perform $H$ on their qubits and parties 3, 4, and 5 perform measurements in $\{\ket{0},\ket{1}\}$,
suppose that the parties 3, 4, and 5 obtain, e.g., the outcomes $\ket{1}$, $\ket{0}$, and $\ket{0}$, respectively.
Then, $\ket{\Phi}^{(RS)}$ is transformed to 
\begin{equation}
\label{eq:state5q-1}
    \frac{1}{2}\left[\ket{0}^{(R)}(\ket{11}^{(12)}+\ket{01}^{(12)})+\ket{1}^{(R)}(\ket{00}^{(12)}-\ket{10}^{(12)})\right].
\end{equation}
For these outcomes, the numbers of parties who obtained the outcome $\ket{1}$ and whose vertex is adjacent to party $l=1,2,R$ are given by
\begin{equation}
    m_1=1,~~m_2=1,~~m_R=0.
\end{equation}
Therefore, parties 1 and 2 perform $Z$ gate operation, and the state~(\ref{eq:state5q-1}) is transformed to
\begin{equation}
\label{eq:state5q-2}
    \frac{1}{2}\left[\ket{0}^{(R)}(\ket{11}^{(12)}-\ket{01}^{(12)})+\ket{1}^{(R)}(\ket{00}^{(12)}+\ket{10}^{(12)})\right].
\end{equation}

Next, party 2 performs the measurement in $\{\ket{\pm}\}$. If the outcome is $\ket{+}$, the state~(\ref{eq:state5q-2}) is transformed to
\begin{equation}
\label{eq:state5q-3}
    \frac{1}{2}\left[\ket{0}^{(R)}(\ket{1}-\ket{0})+\ket{1}^{(R)}(\ket{0}+\ket{1})\right].
\end{equation}
Since the outcome obtained by party 2 is $\ket{+}$, party 1 performs the local unitary represented as
\begin{equation}
    %H Z^{m_R+1}(U_R^\dagger)^\mathrm{T}(iY)^{-1/2}=
    HZH(iY)^{-1/2},
\end{equation}
and the state ~(\ref{eq:state5q-3}) is transformed to $\ket{\Phi}^{(R1)}$ up to the global phase.
This means that the quantum information encoded in the five-qubit code can be extracted to qubit 1 with the above operations.
\end{exa}

\section{Applications to the quantum information splitting}
\label{sec:4}
In this section, we consider applications of Algorithm~\ref{alg2} to the LOCC-QIS, i.e., to share quantum information among multiple parties so that any of the parties cannot extract the shared quantum information on their own but can extract the quantum information only by LOCC\@.
To split and share quantum information is the first step to the quantum secret sharing (QSS)~\cite{PhysRevA.59.1829,gottesman2000theory,singh2005generalized}, which aims to achieve secure quantum communication if more than a certain number of parties cooperate. The difference between QSS and QIS is that QSS additionally requires that information is kept completely unknown if less than a sufficient number of parties cooperate, while in QIS the information may be partially known. The QIS only requires that the information can be exactly extracted if a sufficient number of parties cooperate.
Note that in contrast to the ordinary QIS and QSS settings~\cite{PhysRevA.59.1829,gottesman2000theory,singh2005generalized}, quantum communication between parties is prohibited in our setting.

In Sec.~\ref{sec:3}, we showed that not every party needs to cooperate to extract quantum information.
The number of parties necessary to cooperate depends on to which party the shared quantum information is extracted.
In the following, using tree-topology graph states, we propose the LOCC-QIS schemes with a hierarchical structure of authorities to access quantum information shared among multiple parties.
The hierarchical structure under the ordinary QIS settings was proposed in Refs.~\cite{wang2010hierarchicalQIS,wang2010hierarchical,wang2011multiparty} and called hierarchical quantum information splitting (HQIS).
Therefore, we call the LOCC-QIS scheme introduced here \textit{LOCC-HQIS}.
We limited ourselves to cases where quantum communication is not possible, but instead, we can design a variety of hierarchical structures corresponding to trees in the graph theory.
We will also show that, in a special case, the LOCC-HQIS can be applied to the QSS where the parties can only perform LOCC~\cite{PhysRevA.91.022330,PhysRevA.95.022320,yang2015quantum}.

\subsection{Notations}
We consider LOCC-HQIS using trees $G=(V,E)$ in the following.
The set of vertices $V$ is divided into $V_P$ and $V_{k,l}~(l=1,2,\ldots,m_k)$ as follows.
As shown in Fig.~\ref{graph:G},
let $G_{k,l}=(V_{k,l},E_{k,l})~(k=1,\ldots,|V_P|,l=1,\ldots,m_k)$ denote the $l$th subgraph connected to vertex $k$, which has only one vertex $u_{k,l}$ adjacent to vertex $v_k\in V_P$ and is not included in $P$.
The path $P=(V_P,E_P)$ is explicitly represented by
\begin{align}
    V_P&=\{v_k\}_{k=1}^{|V_P|},~v_1=R,~v_{|V_P|}=j,\\
    E_P&=\{\{v_k,v_{k+1}\}\}_{k=1}^{|V_P|-1}.
\end{align}
These subgraphs are taken to satisfy the relations given by
\begin{equation}
V=V_P\cup\left(\bigcup_{k=1}^{|V_P|}\bigcup_{l=1}^{m_k}V_{k,l}\right),
\end{equation}
\begin{equation}
E=E_P\cup\left(\bigcup_{k=1}^{|V_P|}\bigcup_{l=1}^{m_k}E_{k,l}\cup \{\{u_{k,l},v_k\}\}\right),
\end{equation}
where no pair of $V_P$ and $G_{k,l}$ has intersection, and $m_k$ denotes the number of vertices adjacent to $k$ but not in $V_P$.

\begin{figure}
    \centering
    \includegraphics[keepaspectratio, scale=0.12, angle=0]{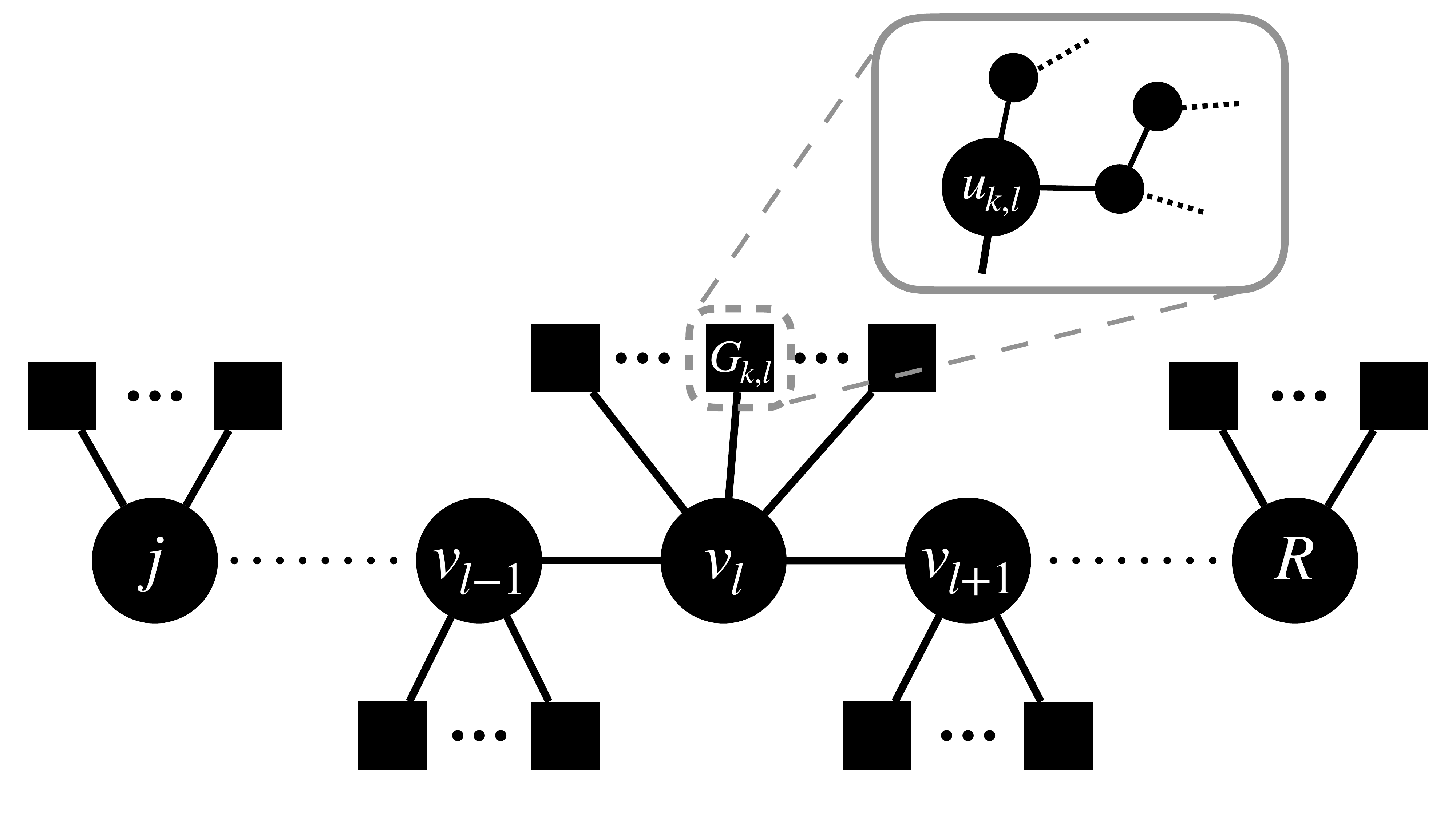}
    \caption{The graph corresponding to the tree-topology graph state. The circles represent the vertices of the path from vertex $j$ to vertex $R$ and the boxes represent subgraphs $G_{k,l}=(V_{k,l},E_{k,l})$ for $k=1,\ldots,|V_P|$ and $l=1,\ldots,m_k$. The dots represent that there are omitted vertices and edges.}
    \label{graph:G}
\end{figure}

\subsection{A property of protocol for extracting quantum information useful for HQIS}
\label{subsec:4A}
For the stabilizer codes represented by graph states corresponding to trees, Algorithm~\ref{alg2} is optimal for providing the minimal number of parties needed to cooperate.
That is, we can identify the hierarchical information access structure of an LOCC-HQIS scheme using such a stabilizer code as presented in the following theorem.

\begin{thm}
\label{thm:gQSS}
Let $G=(V, E)$ be a tree of a graph state $\ket{G}$ that is LC equivalent to the maximally entangled state $\ket{\Phi}^{(RS)}$ between the reference system and the stabilizer code.
The minimum number of parties that need to cooperate for LOCC extraction of quantum information to party $j$ equals to
\begin{equation}
    \left|\big\{k|k\in V,\{k,l\}\in E,l\in V_P\big\}\right|.
\end{equation}
This minimum number of parties is achievable by Algorithm~\ref{alg2}; that is, Algorithm~\ref{alg2} is optimal for a tree $G$ in terms of the number of parties needed for extracting quantum information.
\end{thm}

\begin{proof}
In the proof, we show the following two facts.
\begin{enumerate}
    \item If any one of the parties corresponding to a vertex in $V_P$ does not cooperate, the extraction of quantum information to party $j$ is impossible.
    \item If all the parties corresponding to the vertices in one of the subgraphs ${\{G_{m,l}\}}_{m,l}$ do not cooperate, the extraction of quantum information to party $j$ is impossible.
\end{enumerate}
We show the impossibility of the above two by using the monotonicity of distillable entanglement.
If these facts hold, the number of parties needed for the extraction of quantum information must be more than or equal to $N=\left|\big\{k|k\in V,\{k,l\}\in E,l\in V_P\big\}\right|$.
Then, Algorithm~\ref{alg2} assures the existence of such an algorithm performed with $N$ parties.

See Appendix~\ref{app:thm7} for the details.
\end{proof}

\subsection{Application of extraction of quantum information using tree states to HQIS}
\label{subsec:4B}
Based on the result shown in the previous subsection, we give four examples of LOCC-HQIS schemes using the stabilizer codes whose corresponding maximally entangled states (Eq.~\eqref{eq:RS}) are given by the tree-topology graph states.
First, in Example~\ref{exa:lin}, we consider a scheme using a line-topology graph state as the simplest LOCC-HQIS scheme.
Next, in Example~\ref{exa:tree}, we consider a scheme using tree-topology graph states.
We also describe that a multiparty HQIS~\cite{wang2011multiparty} and an $(n,n)$-threshold scheme using the  Greenberger\UTF{2013}Horne\UTF{2013}Zeilinger (GHZ) states~\cite{Markham2008} in Examples~\ref{exa:multi} and~\ref{exa:GHZ}, respectively, are special examples where the hierarchical structures do not change even when the quantum communication is prohibited.
Note that in general the hierarchical structure changes depending on whether quantum communication is allowed or not.

In the following example of the line-topology graph with $(n+1)$ vertices, the party labeled by $k(<n)$ needs other $k$ parties' cooperation to extract the shared quantum information.

\begin{exa}[LOCC-HQIS with the line-topology graph with $5$ vertices]
\label{exa:lin}
Let us consider a graph state on 5 qubits associated with the line-topology graph shown in Fig.~\ref{graph:line3}. 
\begin{figure}
    \centering
    \includegraphics[keepaspectratio, scale=0.12, angle=0]{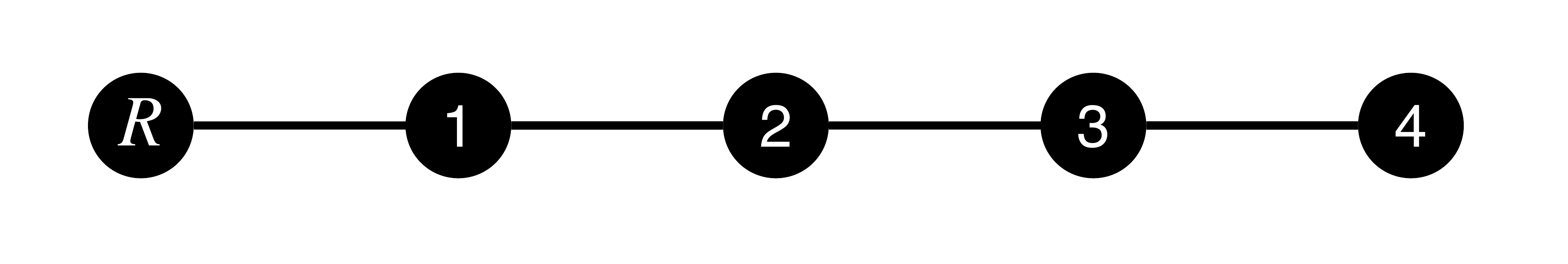}
    \caption{The line-topology graph with 5 vertices. In this case, parties $j=1,2,3$ need $j+1$ parties' cooperation and party 4 needs all 4 parties' cooperation to extract the quantum information. Therefore, in the LOCC-QSS corresponding to this graph state, the authorities of parties in terms of the number of parties needed to extract the shared quantum information are in the increasing order of the integers associated with parties.}
    \label{graph:line3}
\end{figure}
This graph state is also the stabilizer state of the stabilizer generated by the following 5 independent generators
\begin{equation}
\begin{split}
    X^{(R)}Z^{(1)}I^{(2)}I^{(3)}I^{(4)},\\
    Z^{(R)}X^{(1)}Z^{(2)}I^{(3)}I^{(4)},\\
    I^{(R)}Z^{(1)}X^{(2)}Z^{(3)}I^{(4)},\\
    I^{(R)}I^{(1)}Z^{(2)}X^{(3)}Z^{(4)},\\
    I^{(R)}I^{(1)}I^{(2)}Z^{(3)}X^{(4)}.
\end{split}
\end{equation}
This state can be represented by $|\Phi\rangle^{(RS)}$ by regarding it as the maximum entangled state between $\mathcal{H}^{(R)}$ and the stabilizer code by the stabilizer $S$ generated by $Z^{(1)}X^{(2)}Z^{(3)}I^{(4)}$, $I^{(1)}Z^{(2)}X^{(3)}Z^{(4)}$ and $I^{(1)}I^{(2)}Z^{(3)}X^{(4)}$, where its logical $Z$ operator is $X^{(1)}Z^{(2)}I^{(3)}I^{(4)}$.

Now, we consider the extraction of single-qubit quantum information encoded in the stabilizer code $\mathcal{H}^{(S)}$ of 4 qubits.
According to Theorem~\ref{thm:gQSS}, we can calculate the minimum number of parties needed to cooperate to extract quantum information to $j$.
In the case of extracting to party 1, it only needs the cooperation of party 2, regardless of what is performed on parties 3 and 4. 
Party 2 performs the measurement in $\{\ket{0}^{(2)},\ket{1}^{(2)}\}$ and communicates the measurement outcome to party 1, and then 1 performs a local Clifford operation conditioned on the outcome. 
By contrast, in the case of extracting to party $2$, party $2$ needs the cooperation of both party 1 and party 3.
Similarly, party $3$ needs the cooperation of all the other parties, and party $4$ needs it as well.
Therefore, this scheme has a hierarchical information access structure.
This scheme can be straightforwardly generalized to line-topology graphs with $n$ vertices for any $n=2,3,4,5,\ldots$.
\end{exa}

More complex information access structures can be considered by using trees as shown in the following example.
\begin{exa}[LOCC-HQIS with the tree with 5 vertices]
\label{exa:tree}
Let us consider the tree with 5 vertices shown in Fig.~\ref{graph:tree}.
\begin{figure}
    \centering
    \includegraphics[keepaspectratio, scale=0.15, angle=0]{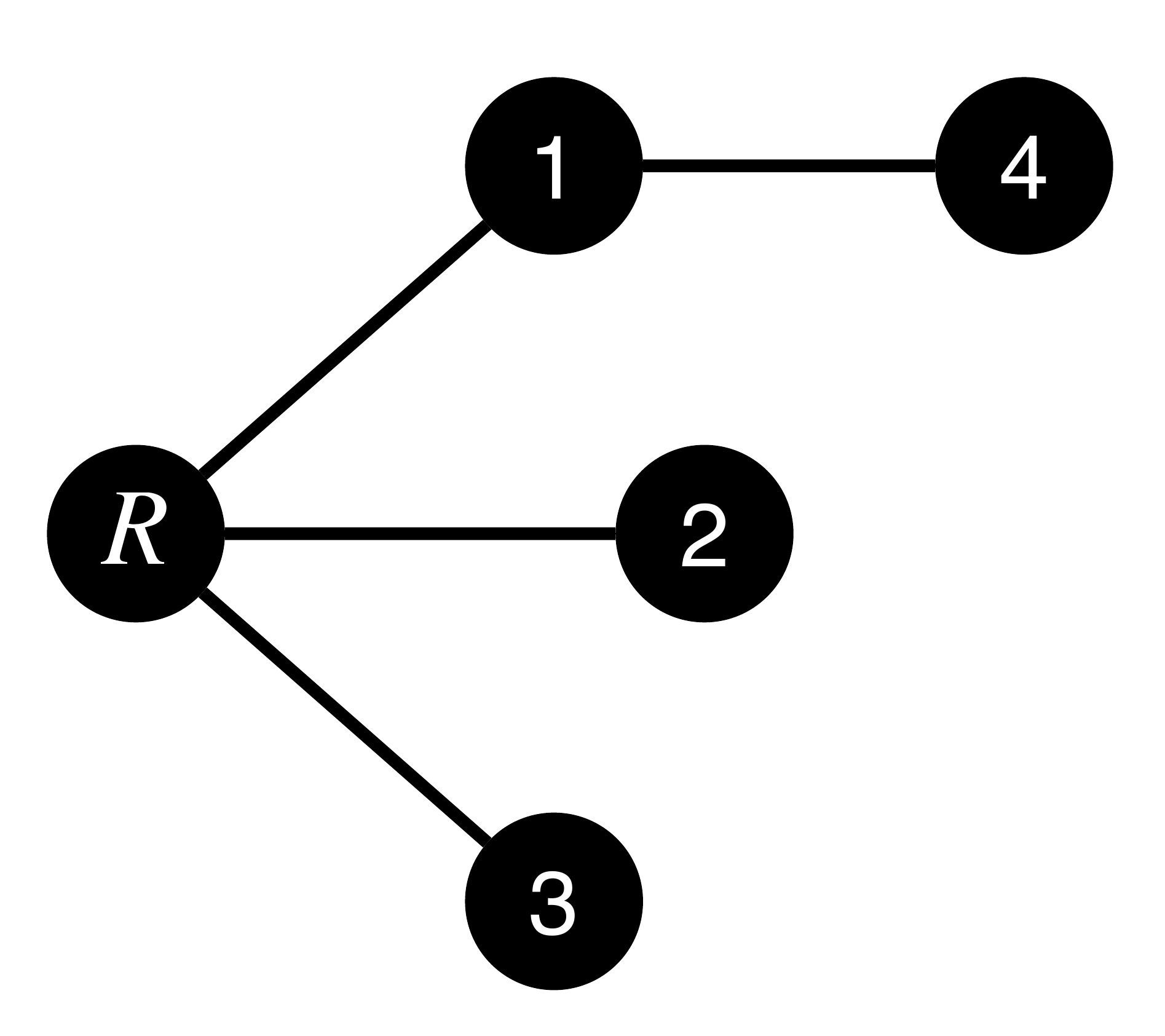}
    \caption{A tree with 5 vertices. In the LOCC-HQIS using this graph, to extract quantum information, parties $1$ and $4$ need the cooperation of all the other parties, and party $2$($3$) needs the cooperation of parties $3$($2$) and $1$.}
    \label{graph:tree}
\end{figure}
According to Theorem~\ref{thm:gQSS}, in the LOCC-HQIS using the graph represented by Fig.~\ref{graph:tree}, to extract quantum information, party $1$ or $4$ needs the cooperation of all the other parties, and party $2$($3$) needs the cooperation of parties $3$($2$) and $1$, respectively.
\end{exa}

The multiparty HQIS proposed in Ref.~\cite{wang2011multiparty} can also be considered as an LOCC-HQIS scheme corresponding to the tree shown in Fig.~\ref{graph:MHQIS}.

\begin{exa}[The multiparty HQIS among $3+2$ agents]
\label{exa:multi}
The tree shown in Fig.~\ref{graph:MHQIS} corresponds to the multiparty HQIS among $3+2$ agents presented in Ref.~\cite{wang2011multiparty}.
In this scheme, the parties $B_1,B_2$ and $B_3$ need only the other parties labeled $B$ to extract the quantum information while $C_1$ and $C_2$ need all the other 4 parties under the LOCC setting.
Such a hierarchical structure is exactly the same as that shown in Ref.~\cite{wang2011multiparty} when quantum communication is allowed.
\begin{figure}
    \centering
    \includegraphics[keepaspectratio, scale=0.15, angle=0]{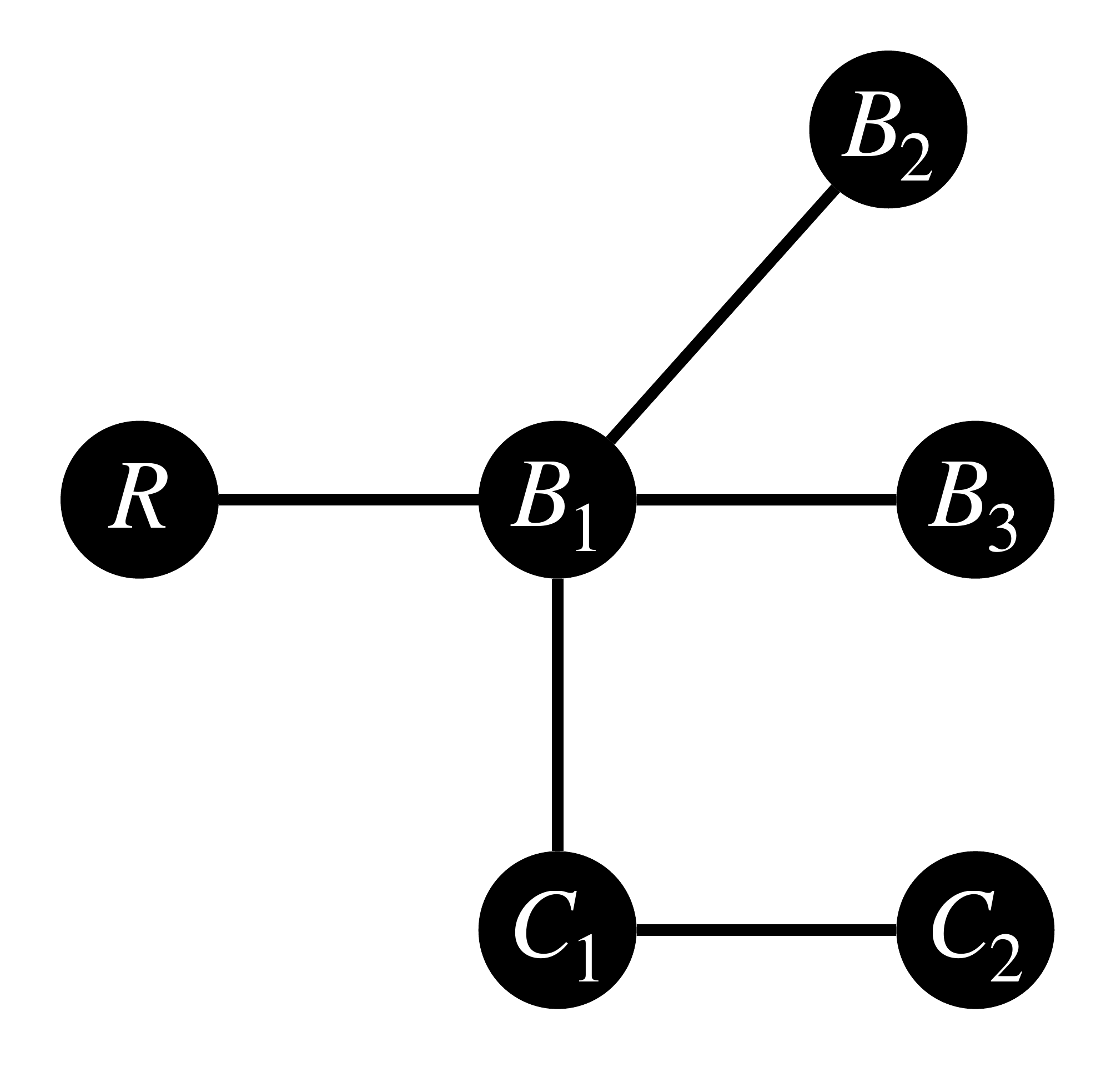}
    \caption{A graph corresponding to the multiparty HQIS~\cite{wang2011multiparty}. $B_j$ can extract the shared quantum information with all the other $B$-labeled parties and one party of $C$-labeled parties. However, $C_j$ needs not only all the other $C$-labeled parties' cooperations but also all the $B$-labeled parties'. That is, $B$-labeled parties have stronger authorities than $C$-labeled parties.}
    \label{graph:MHQIS}
\end{figure}
\end{exa}

We can also consider the $(n,n)$-threshold QSS schemes using GHZ states proposed in Ref.~\cite{Markham2008} in terms of the graph representation.
\begin{exa}[The $(n,n)$-threshold LOCC-QSS using GHZ states~\cite{Markham2008}]
\label{exa:GHZ}
The tree shown in Fig.~\ref{graph:nnTQSS} corresponds to the $(n,n)$-threshold QSS using the GHZ states~\cite{greenberger1990bell}.
In this scheme, every party needs all the other parties' cooperation to extract the quantum information.
This scheme is relevant to the GHZ states since $\ket{\Phi}^{(RS)}$ is represented by
\begin{equation}
\begin{split}
    &\ket{\Phi}^{(RS)}\\
    =&\frac{1}{\sqrt{2}}\{\ket{0}^{(R)}(\ket{00\cdots 0}^{(12\cdots n)}+\ket{11\cdots 1}^{(12\cdots n)})\\
    &+\ket{1}^{(R)}(\ket{00\cdots 0}^{(12\cdots n)}-\ket{11\cdots 1}^{(12\cdots n)})\}.
\end{split}
\end{equation}
The logical states of this case in the stabilizer code are GHZ states.
The state $\ket{\Phi}^{(RS)}$ is LC equivalent to the GHZ state on $n+1$ qubits shown in Fig.~\ref{graph:nnTQSS}.
\begin{figure}
    \centering
    \includegraphics[keepaspectratio, scale=0.15, angle=0]{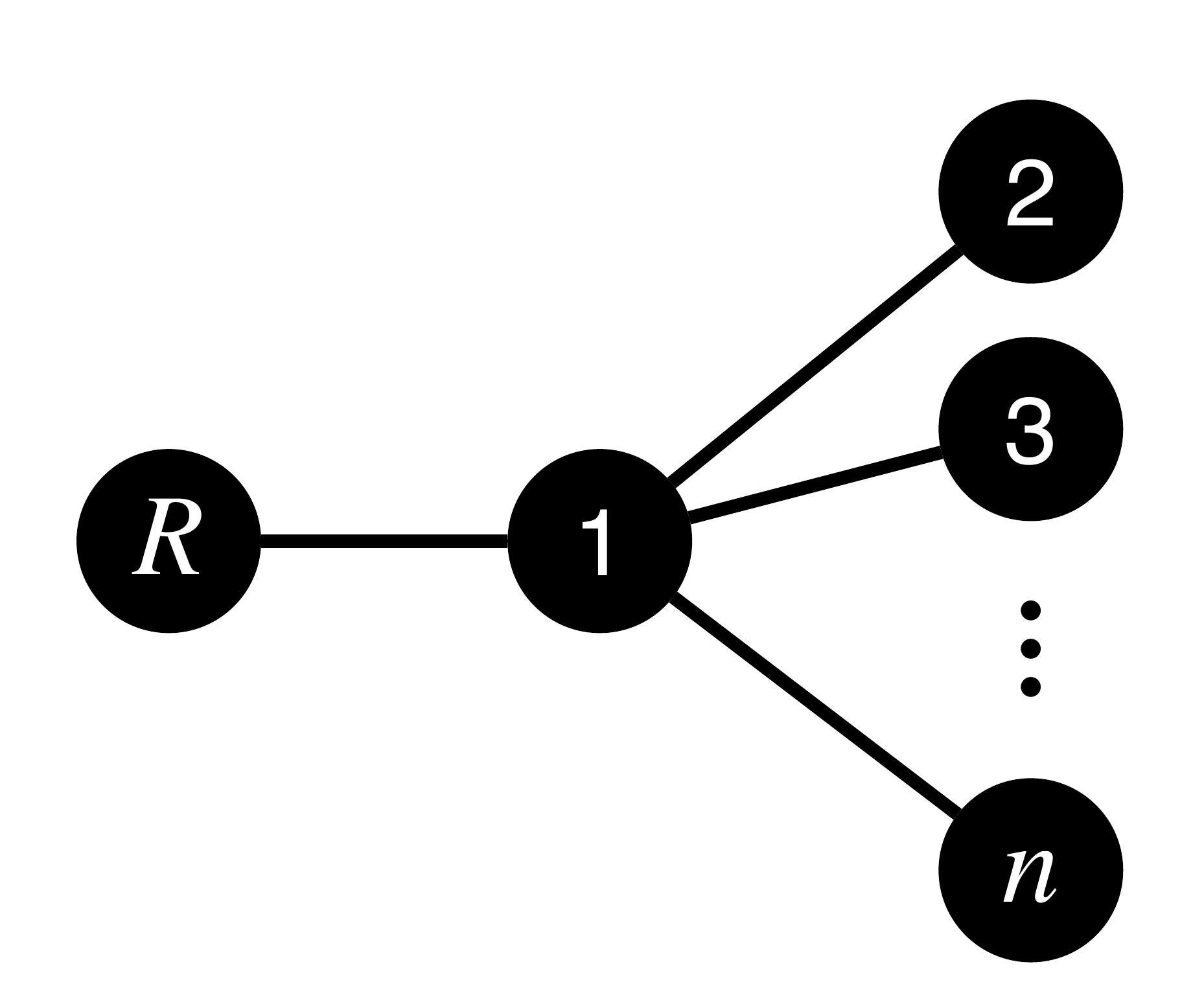}
    \caption{A graph corresponding to the $(n,n)$-threshold QSS using the GHZ states.}
    \label{graph:nnTQSS}
\end{figure}

From the graph represented in Fig.~\ref{graph:nnTQSS}, according to Theorem~\ref{thm:gQSS}, we can see that all the parties are needed to cooperate for the LOCC extraction of quantum information wherever the information is extracted.
The LOCC-QIS scheme corresponding to the GHZ state also works as an $(n,n)$-threshold LOCC-QSS scheme, since no information is accessible if one of the parties does not collaborate even when quantum communication is allowed as shown in Ref.~\cite{Markham2008}.
\end{exa}

As shown in these examples, our results can serve as a unified theoretical tool for studying the QIS and QSS schemes, especially based on the tree-topology graphs.
Using our results in Theorem~\ref{thm:gQSS}, we can design QIS and QSS schemes with various hierarchical information access structures.

\section{conclusion}
\label{sec:5}

In this paper, we analyzed the task of extracting single-qubit quantum information encoded in the stabilizer code of multiple qubits by introducing parties associated with each qubit.
We showed a protocol of LOCC extraction of quantum information to a fixed party given as desired in $\order{n^3}$ time in terms of the number of qubits $n$ in Theorem~\ref{thm:thm10}, whenever possible.
The algorithm to find the protocol for extraction of quantum information is given in Algorithm~\ref{alg2}.
As a consequence of Theorem~\ref{thm:thm10}, the quantum information encoded in the stabilizer code can always be extracted to some party by LOCC in $\order{n^3}$ time.
This algorithm gives an efficient way to decode stabilizer codes, even if we can perform only LOCC rather than multiqubit gates.

Our analysis showed that to extract quantum information, not all parties need to cooperate.
The number of parties needed to cooperate depends on which party quantum information is decoded into.
Theorem~\ref{thm:gQSS} reveals this dependency in cases where the stabilizer code can be represented as a tree-topology graph state up to LC equivalence.
Investigating examples, we also showed that this result can be applied to the LOCC-HQIS schemes including the schemes proposed in the past such as the multiparty HQIS~\cite{wang2011multiparty}, which have the hierarchical information access structures among parties in terms of the minimum numbers of parties needed to cooperate.
For future works, it is open to investigating the hierarchical structure in the stabilizer codes not corresponding to the tree-topology graphs and characterizing distributed quantum information by these structures.

For future works, it is open to showing if the extraction of more than one-qubit quantum information encoded in the stabilizer code is performed only by LOCC in polynomial time in terms of the number of qubits.
Extraction of multi-qubit information is defined as an isometry map which transforms the multi-qubit logical basis $\ket{ab\ldots}^{(S)}$ to the computational basis $\ket{a}^{(j)}\ket{b}^{(k)}\ldots$ of given multiple qubits for extraction.  
The stabilizer codes with more than one logical qubits have applications to low-overhead fault-tolerant quantum computation~\cite{https://doi.org/10.48550/arxiv.2207.08826,gottesman2014faulttolerant,8555154}.
The problem of LOCC extraction of multiple-qubit information from stabilizer codes can be viewed as the problem of extracting multiple Bell states between the reference system consisting of \textit{multiple} qubits and extraction qubits from a given graph state by LOCC without accessing the reference qubits, similarly to the case of LOCC extraction of single-qubit information.
%The maximally entangled state between the reference system and the stabilizer code with a multi-qubit logical space is LC equivalent to the graph state with multiple reference qubits.
%LOCC extraction of multi-qubit information encoded in the stabilizer code corresponds to the LOCC transformation of such a graph state to the Bell pairs between each of the given qubits and each reference qubit.}

However, the possibility of extraction of multi-qubit information cannot be determined solely from the connectivity in the graph between the reference vertices (qubits) and the vertices (qubits) where information is extracted unlike the single-qubit case as explained in Appendix~\ref{app:C}, and therefore, further study is needed to deal with multi-qubit information.
The problem to determine the ability to extract multiple Bell states from a given general graph state by LOCC is NP complete.~\cite{Dahlberg2020transforminggraph}.
Nevertheless, there is still a possibility that the problem can be solved efficiently by utilizing the fact that the initial graph state is restricted to the maximum entanglement state between the reference qubits and code space in LOCC extraction, or by imposing an additional restriction to the initial graph state to describe a particular class of stabilizer codes.

For further applications, it would be interesting to investigate the applicability of our LOCC decoding algorithm to implementing faulty non-Clifford measurement used for preparing a magic state encoded in a quantum error-correcting code in fault-tolerant quantum computation (FTQC)~\cite{Litinski2019magicstate,Lodyga2015}.
In the faulty non-Clifford measurement, one needs to extract an encoded state of a logical qubit into one of the physical qubits in a similar way to our protocol while the error analysis is also needed for the feasibility of FTQC\@.
Our algorithm may lead to insight into conditions of magic-state-preparation protocols that work for general stabilizer codes.

\section*{Acknowledement}
KS is supported by Forefront Physics and Mathematics Program to Drive Transformation (FoPM), a World-leading Innovative Graduate Study (WINGS) Program, the University of Tokyo. MM is supported by MEXT Quantum Leap Flagship Program (MEXT QLEAP) JPMXS0118069605, and
Japan Society for the Promotion of Science (JSPS) KAKENHI Grant No. 21H03394.
HY acknowledges JST PRESTO Grant Number JPMJPR201A and MEXT Quantum Leap Flagship Program (MEXT QLEAP) JPMXS0118069605, JPMXS0120351339.

\appendix
\section{Examples}
\label{sec:appendix}
\subsection{A simple example of extraction of quantum information encoded in two qubits}
\label{app:Ex1}
We consider single-qubit information encoded in a two-qubit system $\mathcal{H}^{(1)} \otimes \mathcal{H}^{(2)} =\mathbb{C}^2 \otimes \mathbb{C}^2$ according to the logical basis states of $\mathcal{H}^{(S)}$ given by
\begin{equation}
    \ket{a}^{(S)}=\ket{a}^{(1)}\ket{a}^{(2)}~(a=0,1),
\end{equation}
where $\ket{a}^{(1)} \in \mathcal{H}^{(1)}$ and $\ket{a}^{(2)} \in \mathcal{H}^{(2)}$ are the computational basis states and a tensor product sign $\otimes$ for pure states is omitted for brevity.
Let $\ket{\pm}^{(2)}$ be the eigenstates of $X^{(2)}\coloneqq\ket{0}\bra{1}^{(2)}+\ket{1}\bra{0}^{(2)}$ with the eigenvalues $\pm1$, respectively.
To extract quantum information encoded in $\rho^{(S)} \in \mathcal{D}(\mathcal{H}^{(S)})$ to $\rho^{(1)} \in \mathcal{D}(\mathcal{H}^{(1)})$ of party $1$, party $2$ performs a measurement in $\{\ket{\pm}^{(2)}\}$ and communicate the outcome of the measurement to party $1$ as shown in the following.

Since $\rho^{(S)}$ can be represented by
\begin{align}
\nonumber
    \rho^{(S)}=&\sum_{a,b=0,1}\rho_{ab}\ket{a}\bra{b}^{(1)}\otimes\ket{a}\bra{b}^{(2)}\\\nonumber
     =&\frac{1}{2}\Big(
      \sum_{a,b=0,1}\rho_{ab}\ket{a}\bra{b}^{(1)}\otimes \ket{+}\bra{+}^{(2)}\\\nonumber
     &+\sum_{a,b=0,1}(-1)^{b}\rho_{ab}\ket{a}\bra{b}^{(1)}\otimes\ket{+}\bra{-}^{(2)}\\\nonumber
     &+\sum_{a,b=0,1}(-1)^{a}\rho_{ab}\ket{a}\bra{b}^{(1)}\otimes\ket{-}\bra{+}^{(2)}\\
     &+\sum_{a,b=0,1}(-1)^{a+b}\rho_{ab}\ket{a}\bra{b}^{(1)}\otimes\ket{-}\bra{-}^{(2)}
     \Big),
\end{align}
the state after the measurement on  $\mathcal{H}^{(2)}$ in $\{\ket{\pm}^{(2)}\}$ with the outcome $\ket{+}^{(2)}$ is given by
\begin{equation}
    \sum_{a,b=0,1}\rho_{ab}\ket{a}\bra{b}^{(1)},
\end{equation}
and with the outcome $\ket{-}^{(2)}$ is given by
\begin{equation}
    \sum_{a,b=0,1}(-1)^{a+b}\rho_{ab}\ket{a}\bra{b}^{(1)}.
\end{equation}
If the outcome is $\ket{+}^{(2)}$, the state is extracted to party $1$, and
if the outcome is $\ket{-}^{(2)}$, the state can be extracted to party $1$ by performing the $Z$ gate on qubit $1$, where $Z^{(1)}\coloneqq \ket{0}\bra{0}^{(1)}-\ket{1}\bra{1}^{(1)}$. This protocol only uses LOCC, and thus the protocol implements an LOCC map for extracting quantum information.

\subsection{A simple example of a state transformation equivalent to an extraction}
\label{app:Ex2}
%\magenta{We} show an example of state transformation equivalent \magenta{to} LOCC extraction of quantum information \magenta{presented} in Appendix~\ref{app:Ex1}.
According to Lemma~\ref{lem:1}, LOCC extraction of quantum information in Appendix~\ref{app:Ex1} is equivalent to the state transformation by LOCC from
\begin{equation}
    |\Phi\rangle^{(RS)}=(\ket{0}^{(R)}\ket{0}^{(1)}\ket{0}^{(2)}+\ket{1}^{(R)}\ket{1}^{(1)}\ket{1}^{(2)})/\sqrt{2}
\end{equation}
to 
\begin{equation}
    |\Phi\rangle^{(R1)}=(\ket{0}^{(R)}\ket{0}^{(1)}+\ket{1}^{(R)}\ket{1}^{(1)})/\sqrt{2}.
\end{equation}
This transformation can be achieved by the following LOCC protocol.
First, party 2 performs a measurement in $\{\ket{\pm}^{(2)}\}$ and communicates the outcome to party 1.  Then $|\Phi\rangle^{(RS)}$ is transformed to
\begin{equation}
\frac{1}{\sqrt{2}}(\ket{0}^{(R)}\ket{0}^{(1)}\pm\ket{1}^{(R)}\ket{1}^{(1)})
\end{equation}
where the $\pm$ sign of $\ket{1}^{(R)}\ket{1}^{(1)}$ corresponds to the outcome $\ket{\pm}$ of the measurement. If the outcome is $\ket{+}$, the state is transformed to $|\Phi\rangle^{(R1)}$. If the outcome of the measurement is $\ket{-}$, party 1 performs the $Z$ gate. Then we have $|\Phi\rangle^{(R1)}$. The sequence of operations in this protocol is the same as the one in Appendix~\ref{app:Ex1}.

\section{Proofs}
\label{sec:appendix_proof}

\subsection{Proof of Theorem~\ref{thm:thm10}}
\label{app:thm5}

We state a useful lemma about quantum state merging~\cite{Yamasaki2019} to prove Theorem~\ref{thm:thm10}.
The task in Theorem~\ref{thm:thm10} can be regarded as extracting quantum information to $\mathcal{H}^{(j)}$ from the bipartite system $\mathcal{H}^{(j)}\otimes\mathcal{H}^{(k\neq j)}~(\mathcal{H}^{(k\neq j)}=\bigotimes_{k\neq j}\mathcal{H}^{(k)})$.
The following Lemma~\ref{lem:cor6} gives a necessary condition for LOCC extraction from a bipartite system $\mathcal{H}^{(A)}\otimes\mathcal{H}^{(B)}$ to $\mathcal{H}^{(B)}$.
\begin{lem}Corollary~6 in Ref.~\cite{Yamasaki2019}:
\label{lem:cor6}
Consider a tripartite pure state
\begin{equation}
\begin{split}
    \ket{\Phi}&\coloneqq\frac{1}{\sqrt{2}}(\ket{0}^{(R)}\ket{\psi_0}^{(AB)}+\ket{1}^{(R)}\ket{\psi_1}^{(AB)})\\
    &\in\mathcal{H}^{(R)}\otimes\mathcal{H}^{(A)}\otimes\mathcal{H}^{(B)},
\end{split}
\end{equation}
where $\ket{\psi_0}^{(AB)}$ and $\ket{\psi_1}^{(AB)}$ are orthogonal pure states of $\mathcal{H}^{(A)}\otimes\mathcal{H}^{(B)}$. 
If $\ket{\Phi}$ is represented by
\begin{equation}
    \ket{\Phi}=\frac{1}{\sqrt{2}}\left(\ket{0}^{(R)}\ket{0}^{(A)}+\ket{1}^{(R)}\ket{1}^{(A)}\right)\otimes\ket{\psi}^{(B)},
\end{equation}
where $\ket{j}^{(R)}~(j=0,1)$ $(\ket{j}^{(A)}~(j=0,1))$ are orthogonal pure states in $\mathcal{H}^{(R)}$ $(\mathcal{H}^{(A)})$ and $\ket{\psi}$ is a pure state in $\mathcal{H}^{(B)}$, there exists no LOCC map which transforms $\ket{\Phi}$ to $(\ket{0}^{(R)}\ket{0}^{(B)}+\ket{1}^{(R)}\ket{1}^{(B)})/\sqrt{2}$.
\end{lem}

In the following, after introducing some notations, we will show Algorithm~\ref{alg2} indeed provides an LOCC map which performs extraction of single-qubit quantum information.

\begin{proof}[Proof of Theorem~\ref{thm:thm10}]
As shown in Lemma~\ref{lem:1}, to construct an LOCC extracting map $\mathcal{C}^{(S\to j)}$ to party $j$, it suffices to construct an LOCC map which transforms $\ket{\Phi}\bra{\Phi}^{(RS)}$ to $\ket{\Phi}\bra{\Phi}^{(Rj)}$ with its action on the reference system being the identity,
\begin{equation}
    \mathrm{id}^{(R)}\otimes\mathcal{C}^{(S\to j)}\ket{\Phi}\bra{\Phi}^{(RS)}=\ket{\Phi}\bra{\Phi}^{(Rj)}
\end{equation}
where $\ket{\Phi}^{(RS)}$ and $\ket{\Phi}^{(Rj)}$ are maximally entangled states defined in Eq.~(\ref{eq:RS}) and (\ref{eq:Rk}).
This $\ket{\Phi}^{(RS)}$ is a stabilizer state since $\ket{\Phi}^{(RS)}$ is a stabilizer state stabilized by the stabilizer generated by $I^{(R)}\otimes S$, $Z^{(R)}\otimes\tilde{Z}^{(S)}$ and $X^{(R)}\otimes\tilde{X}^{(S)}$, where $\tilde{X}^{(S)}(\tilde{Z}^{(S)})$ is the logical $X(Z)$ operator on $\mathcal{H}^{(S)}$.

According to Proposition~\ref{prop:2}, $\ket{\Phi}^{(RS)}$ is LC equivalent to some graph state $\ket{G}$,
that is,
\begin{equation}
\label{eq:initial}
    \ket{G}^{(R,1,\ldots,n)}=\left(U^{(R)}_R\otimes U^{(1)}_1\otimes\cdots\otimes U^{(n)}_n\right)\ket{\Phi}^{(RS)}
\end{equation}
where $U^{(k)}_k$ is a local Clifford map on qubit $k$.

    Now the problem of LOCC construction of $\mathcal{C}^{(S\to j)}$ is reduced to the transformation of a graph state $\ket{\Phi}^{(RS)}$ to another graph state $\ket{\Phi}^{(Rj)}$, which can be considered by using Proposition~\ref{prop:graph} about LPMs on graph states.
    The conversion procedure can be summarized as follows
(i) $Z$-measurements at parties other than the path from $R$ to $j$ on the graph and correction at parties adjacent to that party except $R$.
(ii) $X$-measurements at parties other than $R$ and $j$ in the path from $R$ to $j$ and corrections at parties adjacent to that party except for $R$
(iii) LC operation on the party $j$ to cancel the Clifford operation on the reference, which is intact.

If the vertices $R$ and $j$, which correspond to the reference system and party $j$ respectively, are not connected by a path on the graph, then $j$ does not satisfy the condition in Theorem~\ref{thm:thm10}, and LOCC extraction of quantum information is impossible according to Lemma~\ref{lem:cor6}.
Thus, if party $j$ satisfies the condition, $R$ and $j$ are connected by a path $P=(V_P,E_P)$ from vertex $R$ to vertex $j$.
Let $v_l\in V_P=\{v_l:l=1,2,\cdots,|V_P|\}$ denote the vertices of $P$, where $v_1=j$, $v_{|V_P|}=R$ and $\{v_{l},v_{l+1}\}\in E_P$ for each $l$ is an edge of $P$.
Let $A$ be the set of vertices adjacent to some vertex of $V_P$, which may include $V_P$ itself, i.e.,
\begin{equation}
    A\coloneqq\bigcup_{v_l\in V_P}N_{v_l},
\end{equation}
where $N_{v_l}$ is the set of adjacent vertices to $v_l$ as in Eq.~\eqref{eq:graph_state}.

In Algorithm~\ref{alg2}, given a state~\eqref{eq:initial}, the parties $k$ of $A$ perform $U^{(k)}$, and we obtain
\begin{equation}
\Big(\bigotimes_{k\in A} U^{(k)}\Big)\ket{\Phi}^{(RS)}=(U^{(R)}_R)^\dagger\otimes W^{(V\setminus A)}\ket{G}^{(R1...n)},
\end{equation}
where $W^{(V\setminus A)}=\bigotimes_{k\in V\setminus A}{U_k^{(k)}}^\dagger$ is the unitary acting on the parties corresponding to the parties not included in $A$.
Let $A'$ be the set of the parties of $A$ which does not correspond to a vertex of $V_P$, i.e., $A'=A\setminus V_P$.
Then, each party $k$ of $A'$ performs measurement in $\{\ket{0}^{(k)},\ket{1}^{(k)}\}$, and communicates the outcome of measurement to the parties to all the vertices of $V_P$ adjacent to vertex $k$ except for $v_{|V_P|}=R$.
Since qubit $R$ is a reference qubit, party $R$ is just a virtual one and cannot perform any operation.
Let $m_l$ denote the number of parties who communicate the outcome $\ket{1}$ to party $v_l$ for each $l$.
Then, according to Proposition~\ref{prop:graph}, we obtain
\begin{equation}
    \left(\bigotimes_{l=1}^{|V_P|}(Z^{(v_l)})^{m_l}\right) (U^{(R)}_R)^\dagger \ket{P},
\end{equation}
where vertices in the superscripts represent their corresponding parties.
Thus, each party $v_l$ except for $v_{|V_P|}=R$ performs $(Z^{(v_l)})^{m_l}$, and we obtain
\begin{equation}
    (U^{(R)}_R)^\dagger(Z^{(R)})^{m_{|V_P|}}\ket{P}.
\end{equation}
Next, we consider transforming $P$ by the local complementation shown in Fig.~\ref{graph:Xmeas} in Corollary~\ref{cor:linegraph}, which corresponds to measurement in $\{\ket{\pm}\}$ and LU operations.
Party $v_2$ performs the measurement in $\{\ket{\pm}^{(v_2)}\}$.
Then, according to Proposition~\ref{prop:graph}, if the outcome is $\ket{+}^{(v_2)}$, we obtain
\begin{equation}
    (U^{(R)}_R)^\dagger(Z^{(R)})^{m_{|V_P|}}\otimes Z^{(v_3)} \otimes \sqrt{iY}^{(v_1)} \ket{P},
\end{equation}
and if the outcome is $\ket{-}^{(v_2)}$, we obtain
\begin{equation}
    (U^{(R)}_R)^\dagger(Z^{(R)})^{m_{|V_P|}}\otimes \sqrt{-iY}^{(v_1)}\ket{P}.
\end{equation}
Therefore, if the outcome is $\ket{+}^{(v_2)}$, party $v_1$ and party $v_3$ perform $(\sqrt{iY}^{(v_1)})^{-1}$ and $Z^{(v_3)}$, respectively.
If the outcome is $\ket{-}^{(v_2)}$, party $v_1$ performs $(\sqrt{-iY}^{(v_1}))^{-1}$.
Then, we obtain
\begin{equation}
    (U^{(R)}_R)^\dagger(Z^{(R)})^{m_{|V_P|}}\ket{P'},
\end{equation}
where
\begin{equation}
    P'=\left(V_P\setminus \{v_2\}, (E_P\setminus\{\{v_1,v_2\},\{v_2,v_3\}\})\cup\{\{v_1,v_3\}\}\right).
\end{equation}
By recursively performing the above operation for $l=2,3,\cdots,|V_P|-2$, we obtain
\begin{equation}
    (U^{(R)}_R)^\dagger(Z^{(R)})^{m_{|V_P|}}\ket{P''},
\end{equation}
where
\begin{equation}
\begin{split}
    P''=(&\{v_1,v_{|V_P|-1},v_{|V_P|}\},\\
    &\{\{v_1,v_{|V_P|-1}\},\{v_{|V_P|-1},v_{|V_P|}\}\}).
\end{split}
\end{equation}
The case of $l=|V_P|-1$ is exceptional because we cannot perform a local unitary operation on the reference system.
Finally, party $v_{|V_P|-1}$ performs the measurement in $\{\ket{\pm}^{v_{|V_P|-1}}\}$, and, according to Proposition~\ref{prop:graph}, we obtain
\begin{align}
    (U^{(R)}_R)^\dagger(Z^{(R)})^{m_{|V_P|}+1}\otimes& \sqrt{iY}^{(v_1)} \ket{G'}\\\nonumber
    &(\text{if the outcome is $\ket{+}^{v_{|V_P|-1}}$}),\\
   (U^{(R)}_R)^\dagger(Z^{(R)})^{m_{|V_P|}}\otimes& \sqrt{-iY}^{(v_1)}\ket{G'}\\\nonumber
   &(\text{if the outcome is $\ket{-}^{v_{|V_P|-1}}$}),
\end{align}
where $G'$ is a connected graph with the two vertices $R$ and $v_1=j$.
Thus, by applying an appropriate LU operation on $\mathcal{H}^{(j)}$ depending on the outcome, we obtain $\ket{\Phi}^{(Rj)}$ in both cases of the outcomes.
\end{proof}

\subsection{How to calculate the graph state corresponding to the five-qubit code in Example~\ref{exa:5qcode}}
\label{app:Ex3}
We show how to calculate the graph state corresponding to the five-qubit code.
Let $S'$ be the stabilizer generated by $S\otimes I^{R}$, $\tilde{X}\otimes X^{R}$ and $\tilde{Z}\otimes Z^{R}$.
The logical $X$ and $Z$ are respectively represented by $\tilde{X}=X^{(1)}X^{(2)}X^{(3)}X^{(4)}X^{(5)}$ and $\tilde{Z}=Z^{(1)}Z^{(2)}Z^{(3)}Z^{(4)}Z^{(5)}$.
Logical $Z$ and $X$ operators on this stabilizer code can be calculated using a standard check matrix as mentioned in Chap.~4 in Ref.~\cite{gottesman1997stabilizer}. 

Then, $\ket{\Phi}^{(RS)}$ is the stabilizer state, and its stabilizer, $S'$, is generated by the following 6 generators
\begin{equation}
\label{eq:5qubitgen}
\begin{split}
    g_1&=X^{(1)}Z^{(2)}Z^{(3)}X^{(4)}I^{(5)}I^{(R)},\\
    g_2&=I^{(1)}X^{(2)}Z^{(3)}Z^{(4)}X^{(5)}I^{(R)},\\
    g_3&=X^{(1)}I^{(2)}X^{(3)}Z^{(4)}Z^{(5)}I^{(R)},\\
    g_4&=Z^{(1)}X^{(2)}I^{(3)}X^{(4)}Z^{(5)}I^{(R)},\\
    g_5&=X^{(1)}X^{(2)}X^{(3)}X^{(4)}X^{(5)}X^{(R)},\\
    g_6&=Z^{(1)}Z^{(2)}Z^{(3)}Z^{(4)}Z^{(5)}Z^{(R)}.
\end{split}
\end{equation}
We can transform this set of generators to a (nonunique) canonical one with the transformation as follows~\cite{PhysRevA.69.022316}.
The goal of the transformation is to obtain a set of generators where the Pauli operator acting on the $j$-th qubit of generator $g_j$ is $X$, and the others are $Z$ or $I$.
For this transformation, we can multiply a generator by another generator $g_j\gets g_jg_k~(k\neq j)$, and apply the local Clifford operation such as the Hadamard gate $g_j\gets H^{(1)}g_j H^{(1)}~(\text{for all}~j)$.

We first look at the Pauli operators acting on qubit 1 of the generator $g_1$, which is $X$.
In addition, we arbitrarily choose a Pauli operator other than $X$, here, $Z$.
Then, for the other generators, if the Pauli operator on qubit 1 of the generator $g_j$ is $X$ or $Y$, we update the generator to $g_j\gets g_1g_j$, and if the Pauli operator on qubit 1 of the generator $g_j$ is same as the Pauli operator which we chose, then $g_j$ remains same.
With the above transformation, the generators are transformed to
\begin{equation}
\label{eq:5qubitgen2}
\begin{split}
    g_1&=X^{(1)}Z^{(2)}Z^{(3)}X^{(4)}I^{(5)}I^{(R)},\\
    g_2&=I^{(1)}X^{(2)}Z^{(3)}Z^{(4)}X^{(5)}I^{(R)},\\
    g_3&=I^{(1)}Z^{(2)}Y^{(3)}Y^{(4)}Z^{(5)}I^{(R)},\\
    g_4&=Z^{(1)}X^{(2)}I^{(3)}X^{(4)}Z^{(5)}I^{(R)},\\
    g_5&=I^{(1)}Y^{(2)}Y^{(3)}I^{(4)}X^{(5)}X^{(R)},\\
    g_6&=Z^{(1)}Z^{(2)}Z^{(3)}Z^{(4)}Z^{(5)}Z^{(R)}.
\end{split}
\end{equation}
Next, we repeatedly perform the above transformation looking at qubit $k=2,\cdots,5,6 (=R)$ of the generator $g_k$.
Then, the generators are transformed to
\begin{equation}
\label{eq:5qubitgen3}
\begin{split}
    g_1&=X^{(1)}Z^{(2)}Z^{(3)}X^{(4)}I^{(5)}I^{(R)},\\
    g_2&=Z^{(1)}X^{(2)}Z^{(3)}I^{(4)}I^{(5)}X^{(R)},\\
    g_3&=Z^{(1)}Z^{(2)}X^{(3)}I^{(4)}X^{(5)}I^{(R)},\\
    g_4&=Z^{(1)}I^{(2)}I^{(3)}Z^{(4)}X^{(5)}I^{(R)},\\
    g_5&=I^{(1)}I^{(2)}Z^{(3)}X^{(4)}Z^{(5)}X^{(R)},\\
    g_6&=-I^{(1)}Z^{(2)}I^{(3)}X^{(4)}X^{(5)}Z^{(R)}.
\end{split}
\end{equation}

We can show from this form of generators that $\ket{\Phi}^{(RS)}$ is LC equivalent to the graph state $\ket{G_{\text{five-qubit}}}$ whose graph is given in Fig.~\ref{graph:5qubit} as
\begin{equation}
    \ket{\Phi}^{(RS)}=H^{(R)}\otimes H^{(4)}\otimes H^{(5)}\ket{G}.
\end{equation}
This is because $H^{(R)}H^{(4)}H^{(5)}\ket{\Phi}^{(RS)}$ is stabilized by the stabilizer generated by the following generators, with $g^\prime_j=H^{(R)}H^{(4)}H^{(5)}g_jH^{(R)}H^{(4)}H^{(5)}$,
\begin{equation}
\label{eq:5qubitgen4}
\begin{split}
    g_1^\prime&=X^{(1)}Z^{(2)}Z^{(3)}Z^{(4)}I^{(5)}I^{(R)},\\
    g_2^\prime&=Z^{(1)}X^{(2)}Z^{(3)}I^{(4)}I^{(5)}Z^{(R)},\\
    g_3^\prime&=Z^{(1)}Z^{(2)}X^{(3)}I^{(4)}Z^{(5)}I^{(R)},\\
    g_4^\prime&=Z^{(1)}I^{(2)}I^{(3)}X^{(4)}Z^{(5)}I^{(R)},\\
    g_5^\prime&=I^{(1)}I^{(2)}Z^{(3)}Z^{(4)}X^{(5)}Z^{(R)},\\
    g_6^\prime&=-I^{(1)}Z^{(2)}I^{(3)}Z^{(4)}Z^{(5)}X^{(R)},
\end{split}
\end{equation}
According to Proposition 2 in Ref.~\cite{Hein2006}, the graph state $|{G_\text{five-qubit}}\rangle$ is the stabilizer state of the stabilizer generated by $g_j^\prime=X^{(j)}\otimes\bigotimes_{k\in N_j}Z^{(k)}$.

\subsection{Proof of Theorem~\ref{thm:gQSS}}
\label{app:thm7}

\begin{proof}[Proof of Theorem~\ref{thm:gQSS}]
We consider whether a non-cooperative party can block the extraction to party $j$ by performing a measurement in a wrong basis, communicating a wrong measurement result, or not communicating the measurement result.
As discussed in the main text, we show the following statements.
\begin{enumerate}
    \item If any one of the parties corresponding to a vertex in $V_P$ does not cooperate, extraction of quantum information to party $j$ is impossible.
    \item If all the parties corresponding to the vertices in one of the subgraphs ${\{G_{m,l}\}}_{m,l}$ do not cooperate, the extraction of quantum information to party $j$ is impossible.
\end{enumerate}
To show that LOCC extraction of quantum information is impossible in each case, we consider the reduced state of $\ket{G}\bra{G}$ on the system corresponding to the cooperative parties.
We show that the reduced state cannot be transformed to  $\ket{\Phi}\bra{\Phi}^{(Rj)}$ by LOCC using the monotonicity of an entanglement measure under LOCC\@.
Since the reduced state can be considered to be a classical probability mixture of the state after the projective measurement in an orthonormal basis of the system that is traced out, we can use Proposition~\ref{prop:graph} about LPMs on the graph states.

\begin{figure}[h]
    \centering
    \includegraphics[keepaspectratio, scale=0.12, angle=0]{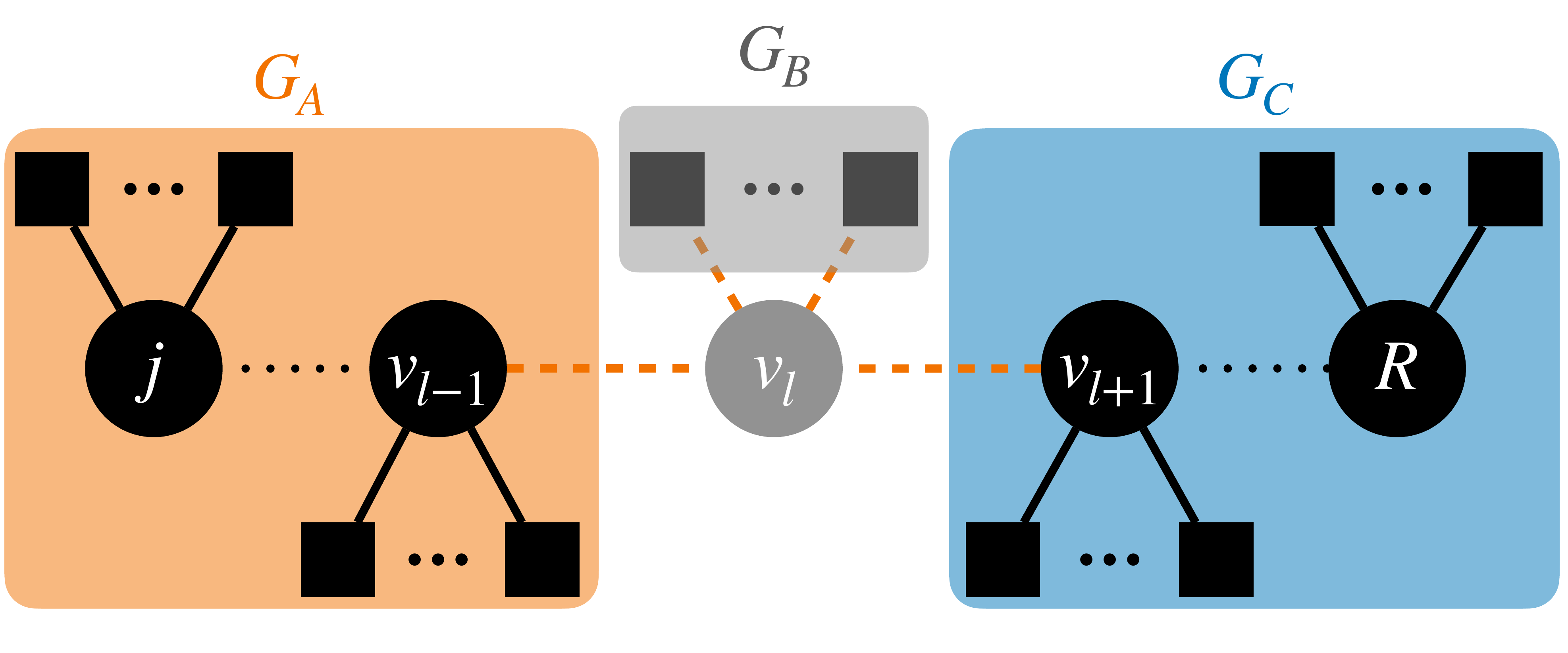}
    \caption{The dotted lines represent the eliminated vertex and edges after the measurement in $\{\ket{0},\ket{1}\}$ on qubit ${(v_l)}$. The graph after the measurement can be divided into three subgraphs $G_A,G_B$, and $G_C$ shaded in orange, gray, and blue, respectively. Since $G_A,G_B$, and $G_C$ are not connected, the state after the measurement has no entanglement between $R$ and $j$.}
    \label{graph:G2}
\end{figure}
\begin{figure}[h]
\centering
    \includegraphics[keepaspectratio, scale=0.12, angle=0]{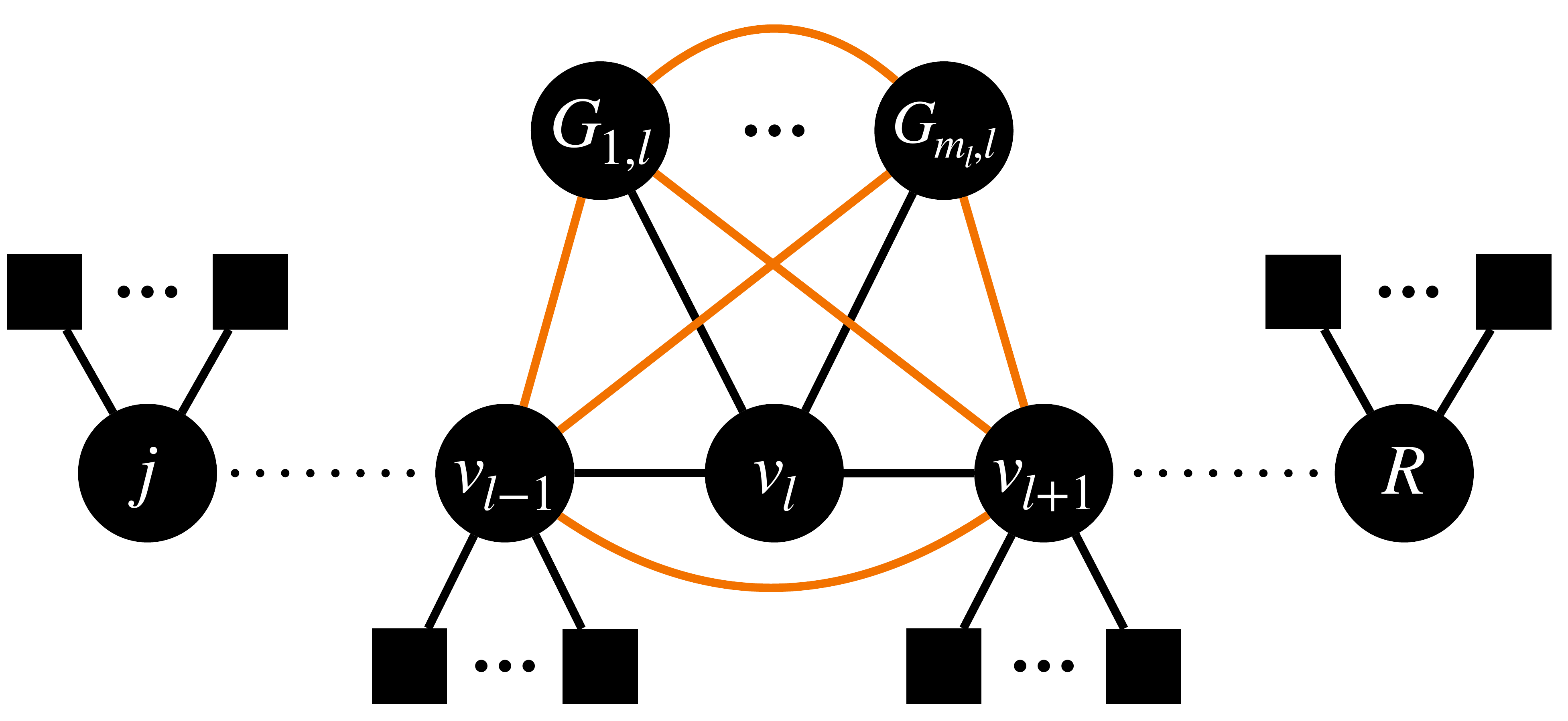}
    \caption{The graph obtained by the local complementation on $v_l$ of the graph shown in Fig.~\ref{graph:G}. The orange lines represent the edges added by the local complementation on $v_l$.}
    \label{graph:G4-a}
\end{figure}

\begin{figure}[h]
    \centering
    \includegraphics[keepaspectratio, scale=0.12, angle=0]{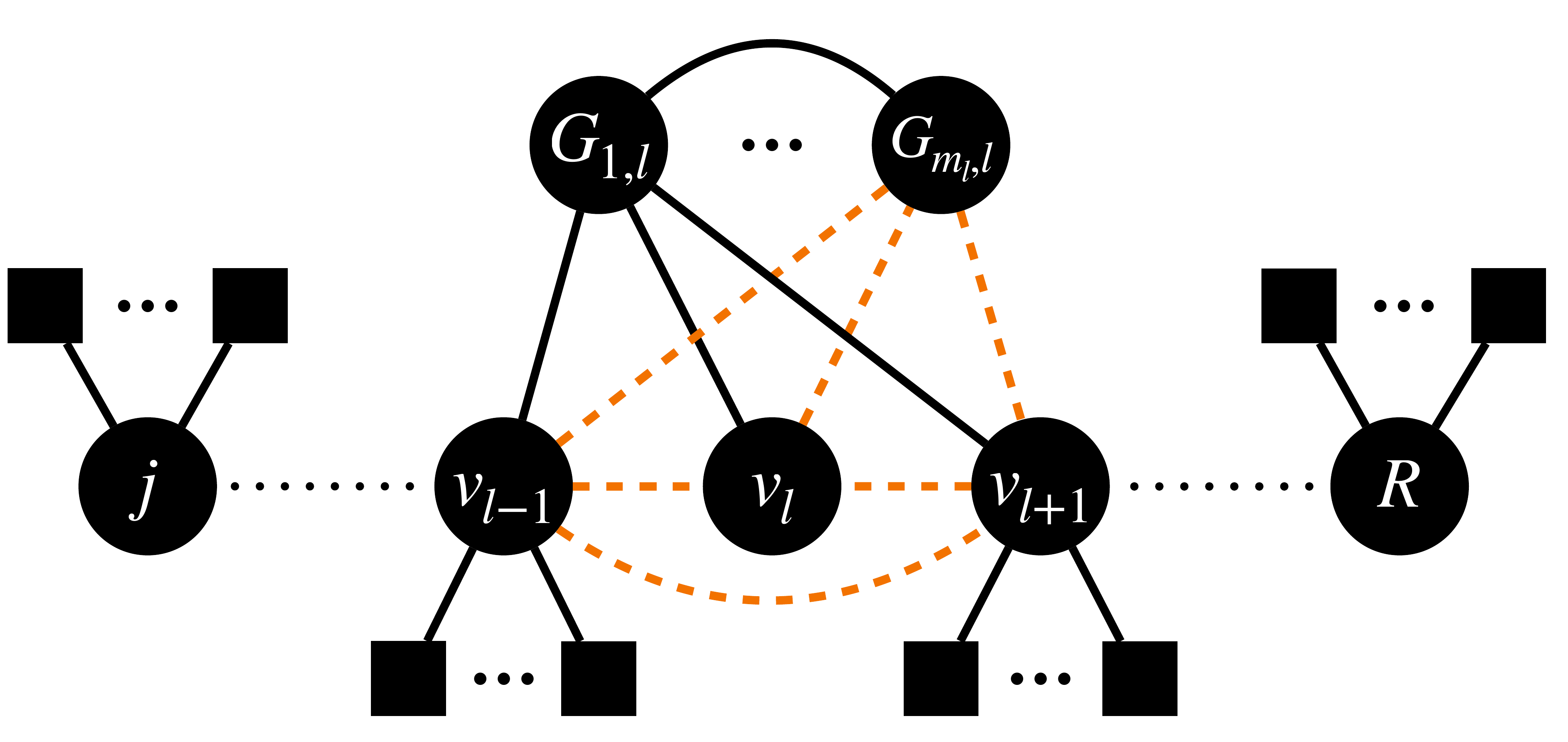}
    \caption{The graph obtained by the local complementation on $G_{1,l}$ of the graph shown in Fig.~\ref{graph:G4-a}. The orange dashed lines represent the edges eliminated by the local complementation on $G_{1,l}$.}
    \label{graph:G4-b}
\end{figure}
\begin{figure}[h]
    \includegraphics[keepaspectratio, scale=0.11, angle=0]{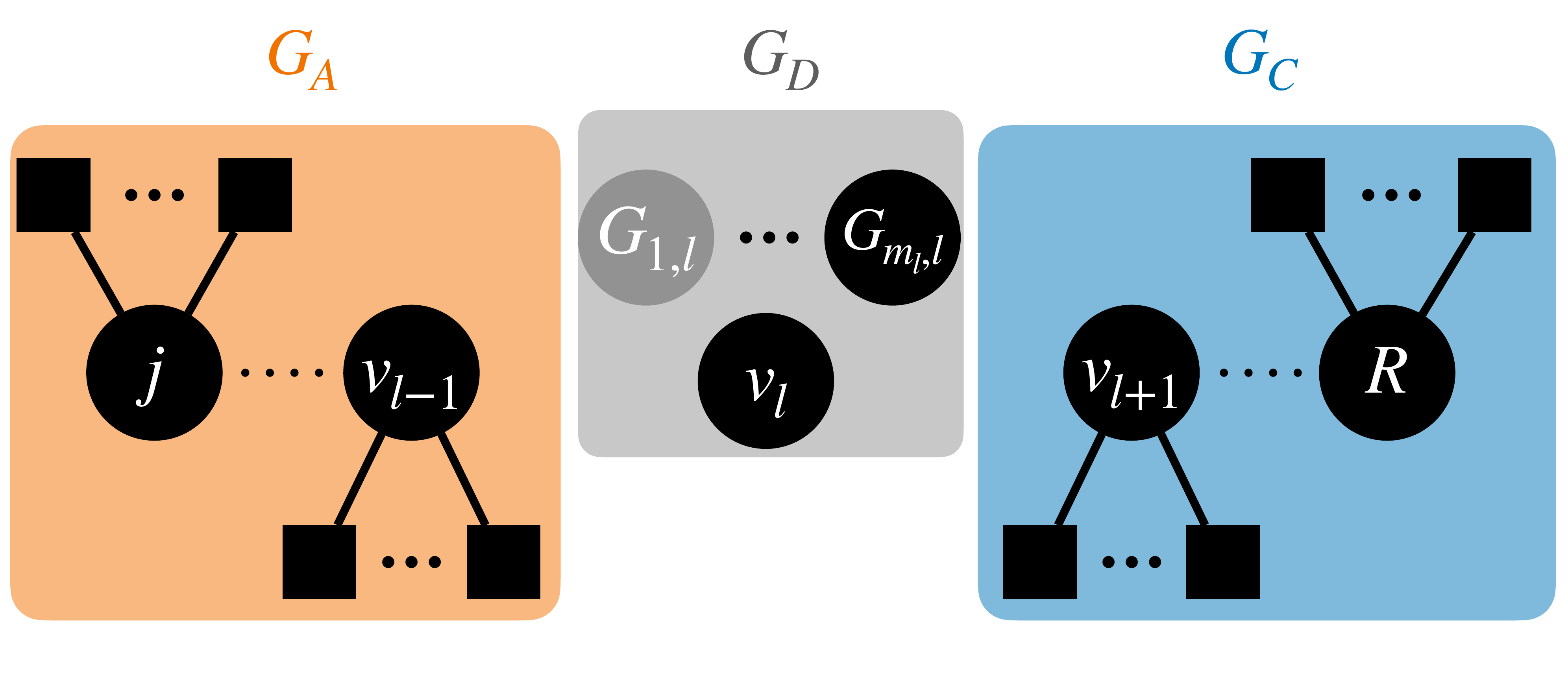}
    \caption{The graph obtained by eliminating $G_{1,l}$ and the local complementation on $v_l$. This graph corresponds to the state after the measurement in $\{\ket{\pm}\}$ on $G_{1,l}$. This graph after the measurement in $\{\ket{\pm}^{(G_{1,l})}\}$ is obtained by the operation written in Proposition~\ref{prop:graph} as follows, (1) Local complementation at $P_m$ shown in Fig.~\ref{graph:G4-a}, (2) Local complementation at $G_{m}^{(1)}$ shown in Fig.~\ref{graph:G4-b}, (3) Eliminating $G_{m}^{(1)}$, and (4) Local complementation at $P_m$ shown in Fig.~\ref{graph:G4-c}. The graph after the measurement can be divided into three subgraphs $G_A,G_D$, and $G_C$ shaded in orange, gray, and blue, respectively. Since $G_A,G_D$ and $G_C$ are not connected with each other, the state after the measurement has no entanglement between $R$ and $j$.}
    \label{graph:G4-c}
\end{figure}
Let us consider the first case where party $v_k\in V_P$ does not cooperate.
Here, as said above, instead of tracing out $\mathcal{H}^{(v_k)}$, we consider the measurement on $\mathcal{H}^{(v_k)}$ since it gives the same reduced state on the composite system of the parties other than $\mathcal{H}^{(v_k)}$.
After the measurement in $\{\ket{0}^{(v_k)},\ket{1}^{(v_k)}\}$ on qubit $v_k$, if the outcome is $\ket{0}^{(v_k)}$, the state of the system $\bigotimes_{v\in V\setminus\{v_k\}}\mathcal{H}^{(v)}$ is represented by $\ket{G-v_k}\bra{G-v_k}$, and if the outcome is $\ket{1}^{(v_k)}$, the state is represented by $U_{v_k,Z,-}\ket{G-v_k}\bra{G-v_k}U_{v_k,Z,-}^\dagger$, where 
\begin{equation}
    U_{v_k,Z,-}=Z^{(u_{k,1})}\cdots Z^{(u_{k,m_l})}Z^{(v_{k-1})}Z^{(v_{k+1})}.
\end{equation}
Since the graph corresponding to the state after the measurement in $\{\ket{0}^{(v_k)},\ket{1}^{(v_k)}\}$ on qubit $v_k$ is represented by Fig.~\ref{graph:G2} according to Proposition~\ref{prop:graph}, we obtain
\begin{equation}
    \ket{G-v_k}=|G_A\rangle \otimes |G_B\rangle \otimes |G_C\rangle,
\end{equation}
where $G_A$, $G_B$, and $G_C$ are the graphs in Fig.~\ref{graph:G2} that are not connected with each other.
Thus the state after the measurement in $\{\ket{0}^{(v_k)},\ket{1}^{(v_k)}\}$ is represented by
\begin{equation}
\begin{split}
    \rho=\frac{1}{2}\sum_{a=0,1} \{&((Z^{(v_{k-1})})^a\ket{G_A}\bra{G_A}(Z^{(v_{k-1})})^a)\\
    &\otimes({(U_{v_k,Z,-}^\prime)}^a\ket{G_B}\bra{G_B}({U_{v_k,Z,-}^\prime}^\dagger)^a)\\
    &\otimes((Z^{(v_{k+1})})^a\ket{G_C}\bra{G_C}(Z^{(v_{k+1})})^a)\},
\end{split}
\end{equation}
where $U_{v_k,Z,-}^\prime\coloneqq Z^{(u_{k,1})}\cdots Z^{(u_{k,m_k})}$.

In the following, we use the distillable entanglement $E_D(\cdot)$ as an entanglement measure to quantify the bipartite entanglement. The distillable entanglement in the bipartite system consisting of the subsystems $A$ and $B$ is defined as~\cite{horodecki2009quantum}
\begin{equation}
\begin{split}
    E_{D}(\rho)=&\sup \left\{r:\right.\\
    &\left.\lim _{n \rightarrow \infty}\left[\inf || \Lambda\left((\rho^{(AB)})^{\otimes n}\right)-\Phi_{2^{r n}}^{ (AB)} \|_{1}\right]=0\right\},
    \end{split}
\end{equation}
where $\Phi_{2^{r n}}^{(AB)}=(\ket{\Phi}\bra{\Phi}^{(AB)})^{\otimes rn}$ and $||\cdot||_1$ is the trace norm.
Let us consider the system as a bipartite system composed of the system corresponding to $G_C$ and the composite system of two systems corresponding to $G_A$ and $G_B$ respectively.
With respect to this bipartition, we obtain $E_{D}(\rho)=0$ since $\rho$ is separable. 

According to Lemma~\ref{lem:1}, if extraction of quantum information is possible, we can transform $\rho$ into $\ket{\Phi}\bra{\Phi}^{(Rj)}$ by LOCC\@. 
Since $E_{D}(\rho)=0$, $E_D(\ket{\Phi}\bra{\Phi}^{(Rj)})=1$ and entanglement measure is not increased by LOCC, $\rho$ cannot be transformed into $\ket{\Phi}\bra{\Phi}^{(Rj)}$ by LOCC\@.
Therefore, extraction of quantum information is impossible without the cooperation of the party corresponding to $v_k$.

Secondly, we consider the case where none of the parties corresponding to the vertices in a subgraph $G_{k,l}$ cooperates. 
Let us define the orthonormal states of the system consisting of qubits in $V_{k,l}$,
\begin{align}
\label{eq:Vkl0}
\ket{0}^{(V_{k,l})}&=\ket{0}^{(u_{k,l})}\otimes|G_{k,l}-u_{k,l}\rangle,\\\label{eq:Vkl1}
\ket{1}^{(V_{k,l})}&=\ket{1}^{(u_{k,l})}\otimes Z^{\left(V_{k,l}\setminus\{u_{k,l}\}\right)}|G_{k,l}-u_{k,l}\rangle.
\end{align}
Since the domain of the reduced state of $\ket{G}\bra{G}$ on the system corresponding to $V_{k,l}$ is spanned by the two states given by Eqs.~(\ref{eq:Vkl0}) and (\ref{eq:Vkl1}), we can regard the system corresponding to $V_{k,l}$ as a qubit. 
Then, $\ket{G}$ can be regarded as the graph state corresponding to the graph in Fig.~\ref{graph:G} by regarding the boxes as vertices.

Without loss of generality, let us consider the case of tracing out the system consisting of $V_{k,1}$. 
In the same way as in the first case, we consider the measurement in $\{\ket{\pm}^{(V_{k,1})}=(\ket{0}^{(V_{k,1})}\pm\ket{1}^{(V_{k,1})})/\sqrt{2}\}$ on the qubit corresponding to $V_{k,1}$ instead of tracing out the system $V_{k,1}$. 
The measurement in $\{\ket{\pm}^{(V_{k,1})}\}$ on the qubit corresponding to $V_{k,1}$ can be represented by the transformation of graphs shown in Fig.~\ref{graph:G4-a},~\ref{graph:G4-b}, and~\ref{graph:G4-c}.
The state obtained by the measurement in $\{\ket{\pm}^{(V_{k,1})}\}$ corresponds to the graph shown in Fig.~\ref{graph:G4-c}.
Then, in the same way as in the first case, the state after the measurement $\rho$ is represented by
\begin{widetext}
\begin{equation}
    \begin{split}
    \rho&=\frac{1}{2}\sum_{a=0,1} \Big\{\left[(Z^{(v_{k-1})})^a\ket{G_A}\bra{G_A}(Z^{(v_{k-1})})^a\right]\\
    &\otimes\left[\sqrt{(-1)^aiY^{(v_k)}}{(U_{v_k,Z,-}^{\prime\prime})}^a\ket{G_D}\bra{G_D}({U_{v_k,Z,-}^{\prime\prime}}^\dagger)^a\sqrt{(-1)^aiY^{(v_k)}}^\dagger\right]\otimes\left[(Z^{(v_{k+1})})^a\ket{G_C}\bra{G_C}(Z^{(v_{k+1})})^a\right]\Big\},
\end{split}
\end{equation}
\end{widetext}
where ${U_{v_k,Z,-}^{\prime\prime}}\coloneqq Z^{(u_{k,2})}\cdots Z^{(u_{k,m_k})}$.
We can see that $\rho$ is separable and hence $E_D(\rho)=0$.
Thus, $\rho$ cannot be transformed into $\ket{\Phi}^{(Rj)}\bra{\Phi}^{(Rj)}$ by LOCC\@.
In the same way as the $v_k$, extraction of quantum information to $j$ is impossible if all the parties in $V_{k,1}$ do not cooperate, which completes the proof.
\end{proof}

\section{Two-qubit extraction is not simple}
\label{app:C}

Extending the analysis of LOCC extraction of single-qubit information to that of multi-qubit quantum information is not straightforward since the multi-qubit case involves an additional network untying problem.  To exhibit this property, let us consider the following graph shown on the left-hand side in Fig.~\ref{graph:twoqbit}, which represents the maximally entangled state between the two-qubit reference and the stabilizer code with two logical qubits.

Consider extracting information into two qubits, $v_3$ and $v_4$. As explained in Lemma~\ref{lem:1}, LOCC extraction to $v_3$ and $v_4$ is equivalent to obtain the maximally entangled state between ${R_1, R_2}$ and ${v_3, v_4}$, which holds two-ebits of entanglement. In the single-qubit case, we only need to check the connectivity between the reference and the party to be extracted. As $R_1$, $R_2$, $v_3$, and $v_4$ are connected, LOCC extraction seems to be possible from the analysis of only the connectivity of the graph. However, two-qubit information cannot be extracted by LOCC in this case because the graph state is actually not maximally entangled for the bipartition drawn in Fig.~\ref{graph:twoqbit}, since this graph state is LC equivalent to the graph state on the right-hand side in Fig.~\ref{graph:twoqbit}, which can only hold one-ebit of entanglement for this bipartition. As LOCC cannot generate entanglement, we cannot obtain the maximally entangled state between ${R_1, R_2}$ and ${v_3, v_4}$ from the initial graph state.

In general, to determine if LOCC extraction is possible and to give the protocol for LOCC extraction, it is necessary to check all the exponential numbers of graphs that can be transformed by local complementation. Finding a class of stabilizer code with multiple logical qubits for efficient LOCC extraction may be possible, but further investigation should be a future task.
 \begin{figure}[b]
    \centering
    \includegraphics[keepaspectratio, scale=0.12, angle=0]{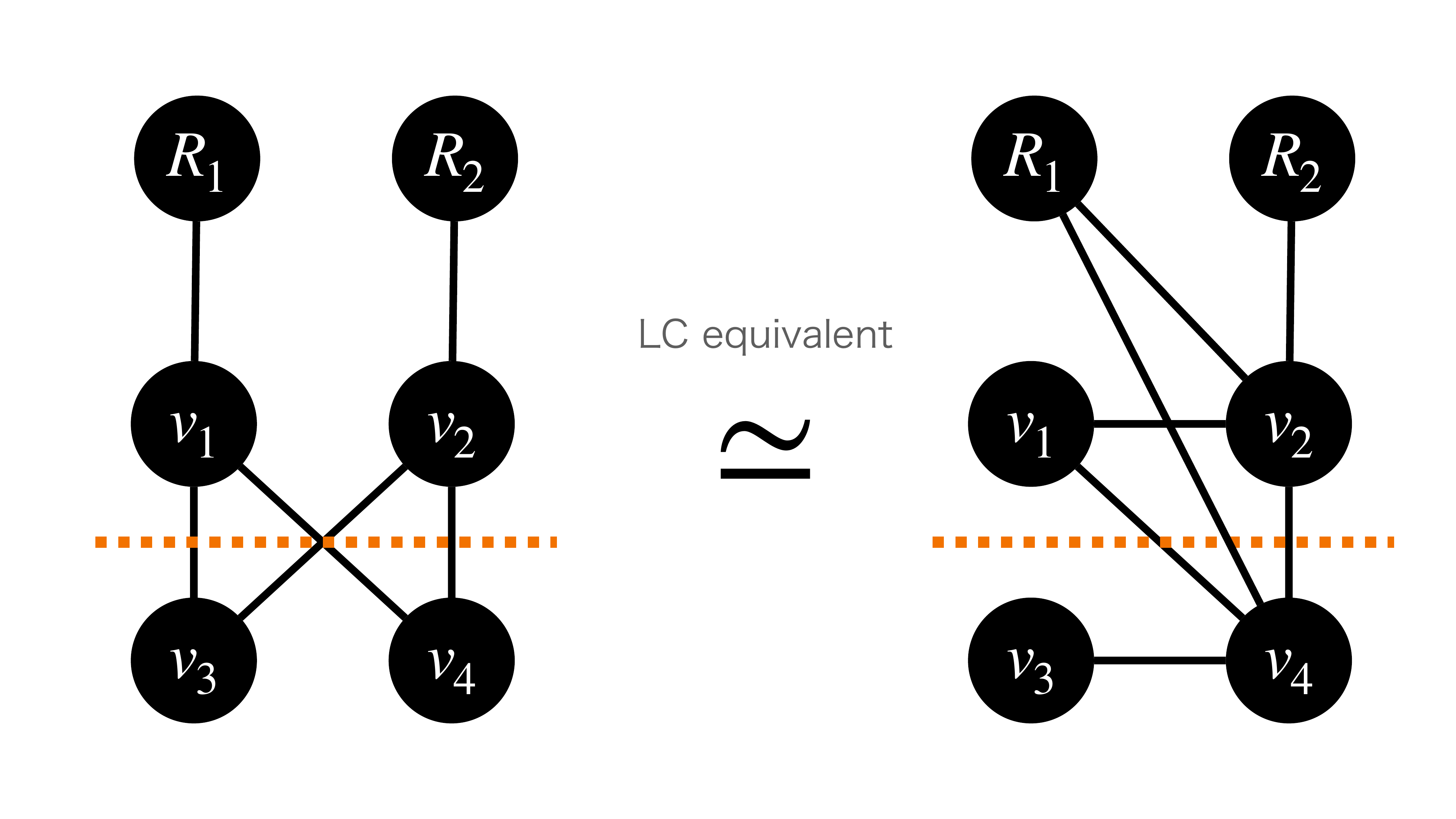}
    \caption{(Left) The graph state corresponding to the maximally entangled state between two reference qubits and a stabilizer code that encodes two qubits. (Right) The graph state that is LC equivalent to the graph state shown on the left-hand side. We can see from the right that the state cannot be maximally entangled for the bipartition drawn by the orange dotted lines. Therefore, two-qubit information cannot be extracted to $v_3$ and $v_4$ by LOCC.}
    \label{graph:twoqbit}
\end{figure}

\bibliography{apssamp}% Produces the bibliography via BibTeX.

\providecommand{\noopsort}[1]{}\providecommand{\singleletter}[1]{#1}%
\begin{thebibliography}{53}%
\makeatletter
\providecommand \@ifxundefined [1]{%
 \@ifx{#1\undefined}
}%
\providecommand \@ifnum [1]{%
 \ifnum #1\expandafter \@firstoftwo
 \else \expandafter \@secondoftwo
 \fi
}%
\providecommand \@ifx [1]{%
 \ifx #1\expandafter \@firstoftwo
 \else \expandafter \@secondoftwo
 \fi
}%
\providecommand \natexlab [1]{#1}%
\providecommand \enquote  [1]{``#1''}%
\providecommand \bibnamefont  [1]{#1}%
\providecommand \bibfnamefont [1]{#1}%
\providecommand \citenamefont [1]{#1}%
\providecommand \href@noop [0]{\@secondoftwo}%
\providecommand \href [0]{\begingroup \@sanitize@url \@href}%
\providecommand \@href[1]{\@@startlink{#1}\@@href}%
\providecommand \@@href[1]{\endgroup#1\@@endlink}%
\providecommand \@sanitize@url [0]{\catcode `\\12\catcode `\$12\catcode
  `\&12\catcode `\#12\catcode `\^12\catcode `\_12\catcode `\%12\relax}%
\providecommand \@@startlink[1]{}%
\providecommand \@@endlink[0]{}%
\providecommand \url  [0]{\begingroup\@sanitize@url \@url }%
\providecommand \@url [1]{\endgroup\@href {#1}{\urlprefix }}%
\providecommand \urlprefix  [0]{URL }%
\providecommand \Eprint [0]{\href }%
\providecommand \doibase [0]{https://doi.org/}%
\providecommand \selectlanguage [0]{\@gobble}%
\providecommand \bibinfo  [0]{\@secondoftwo}%
\providecommand \bibfield  [0]{\@secondoftwo}%
\providecommand \translation [1]{[#1]}%
\providecommand \BibitemOpen [0]{}%
\providecommand \bibitemStop [0]{}%
\providecommand \bibitemNoStop [0]{.\EOS\space}%
\providecommand \EOS [0]{\spacefactor3000\relax}%
\providecommand \BibitemShut  [1]{\csname bibitem#1\endcsname}%
\let\auto@bib@innerbib\@empty
%</preamble>
\bibitem [{\citenamefont {Kimble}(2008)}]{QuantumInternet}%
  \BibitemOpen
  \bibfield  {author} {\bibinfo {author} {\bibfnamefont {H.~J.}\ \bibnamefont
  {Kimble}},\ }\bibfield  {title} {\bibinfo {title} {The quantum internet},\
  }\href {https://doi.org/110.1038/nature07127} {\bibfield  {journal} {\bibinfo
   {journal} {Nature}\ }\textbf {\bibinfo {volume} {453}},\ \bibinfo {pages}
  {1023} (\bibinfo {year} {2008})}\BibitemShut {NoStop}%
\bibitem [{\citenamefont {Wehner}\ \emph {et~al.}(2018)\citenamefont {Wehner},
  \citenamefont {Elkouss},\ and\ \citenamefont
  {Hanson}}]{doi:10.1126/science.aam9288}%
  \BibitemOpen
  \bibfield  {author} {\bibinfo {author} {\bibfnamefont {S.}~\bibnamefont
  {Wehner}}, \bibinfo {author} {\bibfnamefont {D.}~\bibnamefont {Elkouss}},\
  and\ \bibinfo {author} {\bibfnamefont {R.}~\bibnamefont {Hanson}},\
  }\bibfield  {title} {\bibinfo {title} {Quantum internet: A vision for the
  road ahead},\ }\href {https://doi.org/10.1126/science.aam9288} {\bibfield
  {journal} {\bibinfo  {journal} {Science}\ }\textbf {\bibinfo {volume}
  {362}},\ \bibinfo {pages} {eaam9288} (\bibinfo {year} {2018})},\ \Eprint
  {https://arxiv.org/abs/https://www.science.org/doi/pdf/10.1126/science.aam9288}
  {https://www.science.org/doi/pdf/10.1126/science.aam9288} \BibitemShut
  {NoStop}%
\bibitem [{\citenamefont {Cleve}\ \emph {et~al.}(1999)\citenamefont {Cleve},
  \citenamefont {Gottesman},\ and\ \citenamefont {Lo}}]{cleve1999share}%
  \BibitemOpen
  \bibfield  {author} {\bibinfo {author} {\bibfnamefont {R.}~\bibnamefont
  {Cleve}}, \bibinfo {author} {\bibfnamefont {D.}~\bibnamefont {Gottesman}},\
  and\ \bibinfo {author} {\bibfnamefont {H.-K.}\ \bibnamefont {Lo}},\
  }\bibfield  {title} {\bibinfo {title} {How to share a quantum secret},\
  }\href {https://doi.org/10.1103/PhysRevLett.83.648} {\bibfield  {journal}
  {\bibinfo  {journal} {Phys. Rev. Lett.}\ }\textbf {\bibinfo {volume} {83}},\
  \bibinfo {pages} {648} (\bibinfo {year} {1999})}\BibitemShut {NoStop}%
\bibitem [{\citenamefont {Gottesman}(2000)}]{gottesman2000theory}%
  \BibitemOpen
  \bibfield  {author} {\bibinfo {author} {\bibfnamefont {D.}~\bibnamefont
  {Gottesman}},\ }\bibfield  {title} {\bibinfo {title} {Theory of quantum
  secret sharing},\ }\href {https://doi.org/10.1103/PhysRevA.61.042311}
  {\bibfield  {journal} {\bibinfo  {journal} {Phys. Rev. A}\ }\textbf {\bibinfo
  {volume} {61}},\ \bibinfo {pages} {042311} (\bibinfo {year}
  {2000})}\BibitemShut {NoStop}%
\bibitem [{\citenamefont {Singh}\ and\ \citenamefont
  {Srikanth}(2005)}]{singh2005generalized}%
  \BibitemOpen
  \bibfield  {author} {\bibinfo {author} {\bibfnamefont {S.~K.}\ \bibnamefont
  {Singh}}\ and\ \bibinfo {author} {\bibfnamefont {R.}~\bibnamefont
  {Srikanth}},\ }\bibfield  {title} {\bibinfo {title} {Generalized quantum
  secret sharing},\ }\href {https://doi.org/10.1103/PhysRevA.71.012328}
  {\bibfield  {journal} {\bibinfo  {journal} {Phys. Rev. A}\ }\textbf {\bibinfo
  {volume} {71}},\ \bibinfo {pages} {012328} (\bibinfo {year}
  {2005})}\BibitemShut {NoStop}%
\bibitem [{\citenamefont {Noiri}\ \emph {et~al.}(2022)\citenamefont {Noiri},
  \citenamefont {Takeda}, \citenamefont {Nakajima}, \citenamefont {Kobayashi},
  \citenamefont {Sammak}, \citenamefont {Scappucci},\ and\ \citenamefont
  {Tarucha}}]{Noiri2022}%
  \BibitemOpen
  \bibfield  {author} {\bibinfo {author} {\bibfnamefont {A.}~\bibnamefont
  {Noiri}}, \bibinfo {author} {\bibfnamefont {K.}~\bibnamefont {Takeda}},
  \bibinfo {author} {\bibfnamefont {T.}~\bibnamefont {Nakajima}}, \bibinfo
  {author} {\bibfnamefont {T.}~\bibnamefont {Kobayashi}}, \bibinfo {author}
  {\bibfnamefont {A.}~\bibnamefont {Sammak}}, \bibinfo {author} {\bibfnamefont
  {G.}~\bibnamefont {Scappucci}},\ and\ \bibinfo {author} {\bibfnamefont
  {S.}~\bibnamefont {Tarucha}},\ }\bibfield  {title} {\bibinfo {title} {A
  shuttling-based two-qubit logic gate for linking distant silicon quantum
  processors},\ }\href {https://doi.org/10.1038/s41467-022-33453-z} {\bibfield
  {journal} {\bibinfo  {journal} {Nat. Commun.}\ }\textbf {\bibinfo {volume}
  {13}} (\bibinfo {year} {2022})}\BibitemShut {NoStop}%
\bibitem [{\citenamefont {Gottesman}\ and\ \citenamefont
  {Chuang}(1999)}]{Gottesman1999DQC}%
  \BibitemOpen
  \bibfield  {author} {\bibinfo {author} {\bibfnamefont {D.}~\bibnamefont
  {Gottesman}}\ and\ \bibinfo {author} {\bibfnamefont {I.~L.}\ \bibnamefont
  {Chuang}},\ }\bibfield  {title} {\bibinfo {title} {Demonstrating the
  viability of universal quantum computation using teleportation and
  single-qubit operations},\ }\href {https://doi.org/10.1038/46503} {\bibfield
  {journal} {\bibinfo  {journal} {Nature}\ }\textbf {\bibinfo {volume} {402}}
  (\bibinfo {year} {1999})}\BibitemShut {NoStop}%
\bibitem [{\citenamefont {Jiang}\ \emph {et~al.}(2007)\citenamefont {Jiang},
  \citenamefont {Taylor}, \citenamefont {S\o{}rensen},\ and\ \citenamefont
  {Lukin}}]{PhysRevA.76.062323}%
  \BibitemOpen
  \bibfield  {author} {\bibinfo {author} {\bibfnamefont {L.}~\bibnamefont
  {Jiang}}, \bibinfo {author} {\bibfnamefont {J.~M.}\ \bibnamefont {Taylor}},
  \bibinfo {author} {\bibfnamefont {A.~S.}\ \bibnamefont {S\o{}rensen}},\ and\
  \bibinfo {author} {\bibfnamefont {M.~D.}\ \bibnamefont {Lukin}},\ }\bibfield
  {title} {\bibinfo {title} {Distributed quantum computation based on small
  quantum registers},\ }\href {https://doi.org/10.1103/PhysRevA.76.062323}
  {\bibfield  {journal} {\bibinfo  {journal} {Phys. Rev. A}\ }\textbf {\bibinfo
  {volume} {76}},\ \bibinfo {pages} {062323} (\bibinfo {year}
  {2007})}\BibitemShut {NoStop}%
\bibitem [{\citenamefont {Eisert}\ \emph {et~al.}(2000)\citenamefont {Eisert},
  \citenamefont {Jacobs}, \citenamefont {Papadopoulos},\ and\ \citenamefont
  {Plenio}}]{PhysRevA.62.052317}%
  \BibitemOpen
  \bibfield  {author} {\bibinfo {author} {\bibfnamefont {J.}~\bibnamefont
  {Eisert}}, \bibinfo {author} {\bibfnamefont {K.}~\bibnamefont {Jacobs}},
  \bibinfo {author} {\bibfnamefont {P.}~\bibnamefont {Papadopoulos}},\ and\
  \bibinfo {author} {\bibfnamefont {M.~B.}\ \bibnamefont {Plenio}},\ }\bibfield
   {title} {\bibinfo {title} {Optimal local implementation of nonlocal quantum
  gates},\ }\href {https://doi.org/10.1103/PhysRevA.62.052317} {\bibfield
  {journal} {\bibinfo  {journal} {Phys. Rev. A}\ }\textbf {\bibinfo {volume}
  {62}},\ \bibinfo {pages} {052317} (\bibinfo {year} {2000})}\BibitemShut
  {NoStop}%
\bibitem [{\citenamefont {Lemr}\ \emph {et~al.}(2011)\citenamefont {Lemr},
  \citenamefont {\ifmmode~\check{C}\else \v{C}\fi{}ernoch}, \citenamefont
  {Soubusta}, \citenamefont {Kieling}, \citenamefont {Eisert},\ and\
  \citenamefont {Du\ifmmode~\check{s}\else
  \v{s}\fi{}ek}}]{PhysRevLett.106.013602}%
  \BibitemOpen
  \bibfield  {author} {\bibinfo {author} {\bibfnamefont {K.}~\bibnamefont
  {Lemr}}, \bibinfo {author} {\bibfnamefont {A.}~\bibnamefont
  {\ifmmode~\check{C}\else \v{C}\fi{}ernoch}}, \bibinfo {author} {\bibfnamefont
  {J.}~\bibnamefont {Soubusta}}, \bibinfo {author} {\bibfnamefont
  {K.}~\bibnamefont {Kieling}}, \bibinfo {author} {\bibfnamefont
  {J.}~\bibnamefont {Eisert}},\ and\ \bibinfo {author} {\bibfnamefont
  {M.}~\bibnamefont {Du\ifmmode~\check{s}\else \v{s}\fi{}ek}},\ }\bibfield
  {title} {\bibinfo {title} {Experimental implementation of the optimal
  linear-optical controlled phase gate},\ }\href
  {https://doi.org/10.1103/PhysRevLett.106.013602} {\bibfield  {journal}
  {\bibinfo  {journal} {Phys. Rev. Lett.}\ }\textbf {\bibinfo {volume} {106}},\
  \bibinfo {pages} {013602} (\bibinfo {year} {2011})}\BibitemShut {NoStop}%
\bibitem [{\citenamefont {Chou}\ \emph {et~al.}(2018)\citenamefont {Chou},
  \citenamefont {Blumoff}, \citenamefont {Wang}, \citenamefont {Reinhold},
  \citenamefont {Axline}, \citenamefont {Gao}, \citenamefont {Frunzio},
  \citenamefont {Devoret}, \citenamefont {Jiang},\ and\ \citenamefont
  {Schoelkopf}}]{Chou2018}%
  \BibitemOpen
  \bibfield  {author} {\bibinfo {author} {\bibfnamefont {K.~S.}\ \bibnamefont
  {Chou}}, \bibinfo {author} {\bibfnamefont {J.~Z.}\ \bibnamefont {Blumoff}},
  \bibinfo {author} {\bibfnamefont {C.~S.}\ \bibnamefont {Wang}}, \bibinfo
  {author} {\bibfnamefont {P.~C.}\ \bibnamefont {Reinhold}}, \bibinfo {author}
  {\bibfnamefont {C.~J.}\ \bibnamefont {Axline}}, \bibinfo {author}
  {\bibfnamefont {Y.~Y.}\ \bibnamefont {Gao}}, \bibinfo {author} {\bibfnamefont
  {L.}~\bibnamefont {Frunzio}}, \bibinfo {author} {\bibfnamefont {M.~H.}\
  \bibnamefont {Devoret}}, \bibinfo {author} {\bibfnamefont {L.}~\bibnamefont
  {Jiang}},\ and\ \bibinfo {author} {\bibfnamefont {R.~J.}\ \bibnamefont
  {Schoelkopf}},\ }\bibfield  {title} {\bibinfo {title} {Teleportation-based
  realization of an optical quantum two-qubit entangling gate},\ }\href
  {https://doi.org/10.1038/s41586-018-0470-y} {\bibfield  {journal} {\bibinfo
  {journal} {Nature}\ }\textbf {\bibinfo {volume} {561}} (\bibinfo {year}
  {2018})}\BibitemShut {NoStop}%
\bibitem [{\citenamefont {Peres}\ and\ \citenamefont
  {Wootters}(1991)}]{PhysRevLett.66.1119}%
  \BibitemOpen
  \bibfield  {author} {\bibinfo {author} {\bibfnamefont {A.}~\bibnamefont
  {Peres}}\ and\ \bibinfo {author} {\bibfnamefont {W.~K.}\ \bibnamefont
  {Wootters}},\ }\bibfield  {title} {\bibinfo {title} {Optimal detection of
  quantum information},\ }\href {https://doi.org/10.1103/PhysRevLett.66.1119}
  {\bibfield  {journal} {\bibinfo  {journal} {Phys. Rev. Lett.}\ }\textbf
  {\bibinfo {volume} {66}},\ \bibinfo {pages} {1119} (\bibinfo {year}
  {1991})}\BibitemShut {NoStop}%
\bibitem [{\citenamefont {Horodecki}\ \emph {et~al.}(2009)\citenamefont
  {Horodecki}, \citenamefont {Horodecki}, \citenamefont {Horodecki},\ and\
  \citenamefont {Horodecki}}]{horodecki2009quantum}%
  \BibitemOpen
  \bibfield  {author} {\bibinfo {author} {\bibfnamefont {R.}~\bibnamefont
  {Horodecki}}, \bibinfo {author} {\bibfnamefont {P.}~\bibnamefont
  {Horodecki}}, \bibinfo {author} {\bibfnamefont {M.}~\bibnamefont
  {Horodecki}},\ and\ \bibinfo {author} {\bibfnamefont {K.}~\bibnamefont
  {Horodecki}},\ }\bibfield  {title} {\bibinfo {title} {Quantum entanglement},\
  }\href {https://doi.org/10.1103/RevModPhys.81.865} {\bibfield  {journal}
  {\bibinfo  {journal} {Rev. Mod. Phys.}\ }\textbf {\bibinfo {volume} {81}},\
  \bibinfo {pages} {865} (\bibinfo {year} {2009})}\BibitemShut {NoStop}%
\bibitem [{\citenamefont {Chitambar}\ \emph {et~al.}(2014)\citenamefont
  {Chitambar}, \citenamefont {Leung}, \citenamefont {Man{\v{c}}inska},
  \citenamefont {Ozols},\ and\ \citenamefont
  {Winter}}]{chitambar2014everything}%
  \BibitemOpen
  \bibfield  {author} {\bibinfo {author} {\bibfnamefont {E.}~\bibnamefont
  {Chitambar}}, \bibinfo {author} {\bibfnamefont {D.}~\bibnamefont {Leung}},
  \bibinfo {author} {\bibfnamefont {L.}~\bibnamefont {Man{\v{c}}inska}},
  \bibinfo {author} {\bibfnamefont {M.}~\bibnamefont {Ozols}},\ and\ \bibinfo
  {author} {\bibfnamefont {A.}~\bibnamefont {Winter}},\ }\bibfield  {title}
  {\bibinfo {title} {Everything you always wanted to know about locc (but were
  afraid to ask)},\ }\href {https://doi.org/10.1007/s00220-014-1953-9}
  {\bibfield  {journal} {\bibinfo  {journal} {Commun. Math. Phys.}\ }\textbf
  {\bibinfo {volume} {328}},\ \bibinfo {pages} {303} (\bibinfo {year}
  {2014})}\BibitemShut {NoStop}%
\bibitem [{\citenamefont {Nielsen}(1999)}]{PhysRevLett.83.436}%
  \BibitemOpen
  \bibfield  {author} {\bibinfo {author} {\bibfnamefont {M.~A.}\ \bibnamefont
  {Nielsen}},\ }\bibfield  {title} {\bibinfo {title} {Conditions for a class of
  entanglement transformations},\ }\href
  {https://doi.org/10.1103/PhysRevLett.83.436} {\bibfield  {journal} {\bibinfo
  {journal} {Phys. Rev. Lett.}\ }\textbf {\bibinfo {volume} {83}},\ \bibinfo
  {pages} {436} (\bibinfo {year} {1999})}\BibitemShut {NoStop}%
\bibitem [{\citenamefont {Lo}\ and\ \citenamefont
  {Popescu}(2001)}]{PhysRevA.63.022301}%
  \BibitemOpen
  \bibfield  {author} {\bibinfo {author} {\bibfnamefont {H.-K.}\ \bibnamefont
  {Lo}}\ and\ \bibinfo {author} {\bibfnamefont {S.}~\bibnamefont {Popescu}},\
  }\bibfield  {title} {\bibinfo {title} {Concentrating entanglement by local
  actions: Beyond mean values},\ }\href
  {https://doi.org/10.1103/PhysRevA.63.022301} {\bibfield  {journal} {\bibinfo
  {journal} {Phys. Rev. A}\ }\textbf {\bibinfo {volume} {63}},\ \bibinfo
  {pages} {022301} (\bibinfo {year} {2001})}\BibitemShut {NoStop}%
\bibitem [{\citenamefont {Walgate}\ \emph {et~al.}(2000)\citenamefont
  {Walgate}, \citenamefont {Short}, \citenamefont {Hardy},\ and\ \citenamefont
  {Vedral}}]{PhysRevLett.85.4972}%
  \BibitemOpen
  \bibfield  {author} {\bibinfo {author} {\bibfnamefont {J.}~\bibnamefont
  {Walgate}}, \bibinfo {author} {\bibfnamefont {A.~J.}\ \bibnamefont {Short}},
  \bibinfo {author} {\bibfnamefont {L.}~\bibnamefont {Hardy}},\ and\ \bibinfo
  {author} {\bibfnamefont {V.}~\bibnamefont {Vedral}},\ }\bibfield  {title}
  {\bibinfo {title} {Local distinguishability of multipartite orthogonal
  quantum states},\ }\href {https://doi.org/10.1103/PhysRevLett.85.4972}
  {\bibfield  {journal} {\bibinfo  {journal} {Phys. Rev. Lett.}\ }\textbf
  {\bibinfo {volume} {85}},\ \bibinfo {pages} {4972} (\bibinfo {year}
  {2000})}\BibitemShut {NoStop}%
\bibitem [{\citenamefont {Bergou}(2010)}]{doi:10.1080/09500340903477756}%
  \BibitemOpen
  \bibfield  {author} {\bibinfo {author} {\bibfnamefont {J.~A.}\ \bibnamefont
  {Bergou}},\ }\bibfield  {title} {\bibinfo {title} {Discrimination of quantum
  states},\ }\href {https://doi.org/10.1080/09500340903477756} {\bibfield
  {journal} {\bibinfo  {journal} {J. Mod. Opt.}\ }\textbf {\bibinfo {volume}
  {57}},\ \bibinfo {pages} {160} (\bibinfo {year} {2010})}\BibitemShut
  {NoStop}%
\bibitem [{\citenamefont {Rahaman}\ and\ \citenamefont
  {Parker}(2015)}]{PhysRevA.91.022330}%
  \BibitemOpen
  \bibfield  {author} {\bibinfo {author} {\bibfnamefont {R.}~\bibnamefont
  {Rahaman}}\ and\ \bibinfo {author} {\bibfnamefont {M.~G.}\ \bibnamefont
  {Parker}},\ }\bibfield  {title} {\bibinfo {title} {Quantum scheme for secret
  sharing based on local distinguishability},\ }\href
  {https://doi.org/10.1103/PhysRevA.91.022330} {\bibfield  {journal} {\bibinfo
  {journal} {Phys. Rev. A}\ }\textbf {\bibinfo {volume} {91}},\ \bibinfo
  {pages} {022330} (\bibinfo {year} {2015})}\BibitemShut {NoStop}%
\bibitem [{\citenamefont {Wang}\ \emph {et~al.}(2017)\citenamefont {Wang},
  \citenamefont {Li}, \citenamefont {Peng},\ and\ \citenamefont
  {Yang}}]{PhysRevA.95.022320}%
  \BibitemOpen
  \bibfield  {author} {\bibinfo {author} {\bibfnamefont {J.}~\bibnamefont
  {Wang}}, \bibinfo {author} {\bibfnamefont {L.}~\bibnamefont {Li}}, \bibinfo
  {author} {\bibfnamefont {H.}~\bibnamefont {Peng}},\ and\ \bibinfo {author}
  {\bibfnamefont {Y.}~\bibnamefont {Yang}},\ }\bibfield  {title} {\bibinfo
  {title} {Quantum-secret-sharing scheme based on local distinguishability of
  orthogonal multiqudit entangled states},\ }\href
  {https://doi.org/10.1103/PhysRevA.95.022320} {\bibfield  {journal} {\bibinfo
  {journal} {Phys. Rev. A}\ }\textbf {\bibinfo {volume} {95}},\ \bibinfo
  {pages} {022320} (\bibinfo {year} {2017})}\BibitemShut {NoStop}%
\bibitem [{\citenamefont {Yang}\ \emph {et~al.}(2015)\citenamefont {Yang},
  \citenamefont {Gao}, \citenamefont {Wu}, \citenamefont {Qin}, \citenamefont
  {Zuo},\ and\ \citenamefont {Wen}}]{yang2015quantum}%
  \BibitemOpen
  \bibfield  {author} {\bibinfo {author} {\bibfnamefont {Y.-H.}\ \bibnamefont
  {Yang}}, \bibinfo {author} {\bibfnamefont {F.}~\bibnamefont {Gao}}, \bibinfo
  {author} {\bibfnamefont {X.}~\bibnamefont {Wu}}, \bibinfo {author}
  {\bibfnamefont {S.-J.}\ \bibnamefont {Qin}}, \bibinfo {author} {\bibfnamefont
  {H.-J.}\ \bibnamefont {Zuo}},\ and\ \bibinfo {author} {\bibfnamefont {Q.-Y.}\
  \bibnamefont {Wen}},\ }\bibfield  {title} {\bibinfo {title} {Quantum secret
  sharing via local operations and classical communication},\ }\href
  {https://doi.org/10.1038/srep16967} {\bibfield  {journal} {\bibinfo
  {journal} {Sci. Rep.}\ }\textbf {\bibinfo {volume} {5}},\ \bibinfo {pages}
  {16967} (\bibinfo {year} {2015})}\BibitemShut {NoStop}%
\bibitem [{\citenamefont {Yamasaki}\ and\ \citenamefont
  {Murao}(2019{\natexlab{a}})}]{Yamasaki2019}%
  \BibitemOpen
  \bibfield  {author} {\bibinfo {author} {\bibfnamefont {H.}~\bibnamefont
  {Yamasaki}}\ and\ \bibinfo {author} {\bibfnamefont {M.}~\bibnamefont
  {Murao}},\ }\bibfield  {title} {\bibinfo {title} {Quantum state merging for
  arbitrarily small-dimensional systems},\ }\href
  {https://doi.org/10.1109/TIT.2018.2889829} {\bibfield  {journal} {\bibinfo
  {journal} {IEEE Trans. Inf. Theory}\ }\textbf {\bibinfo {volume} {65}},\
  \bibinfo {pages} {3950} (\bibinfo {year} {2019}{\natexlab{a}})}\BibitemShut
  {NoStop}%
\bibitem [{\citenamefont {Yamasaki}\ and\ \citenamefont
  {Murao}(2019{\natexlab{b}})}]{Yamasaki2018}%
  \BibitemOpen
  \bibfield  {author} {\bibinfo {author} {\bibfnamefont {H.}~\bibnamefont
  {Yamasaki}}\ and\ \bibinfo {author} {\bibfnamefont {M.}~\bibnamefont
  {Murao}},\ }\bibfield  {title} {\bibinfo {title} {Distributed encoding and
  decoding of quantum information over networks},\ }\href
  {https://doi.org/https://doi.org/10.1002/qute.201800066} {\bibfield
  {journal} {\bibinfo  {journal} {Adv. Quantum Technol.}\ }\textbf {\bibinfo
  {volume} {2}},\ \bibinfo {pages} {1800066} (\bibinfo {year}
  {2019}{\natexlab{b}})}\BibitemShut {NoStop}%
\bibitem [{\citenamefont {Shor}(1995)}]{shor1995scheme}%
  \BibitemOpen
  \bibfield  {author} {\bibinfo {author} {\bibfnamefont {P.~W.}\ \bibnamefont
  {Shor}},\ }\bibfield  {title} {\bibinfo {title} {Scheme for reducing
  decoherence in quantum computer memory},\ }\href
  {https://doi.org/10.1103/PhysRevA.52.R2493} {\bibfield  {journal} {\bibinfo
  {journal} {Phys. Rev. A}\ }\textbf {\bibinfo {volume} {52}},\ \bibinfo
  {pages} {R2493} (\bibinfo {year} {1995})}\BibitemShut {NoStop}%
\bibitem [{\citenamefont {Steane}(1996)}]{steane1996error}%
  \BibitemOpen
  \bibfield  {author} {\bibinfo {author} {\bibfnamefont {A.~M.}\ \bibnamefont
  {Steane}},\ }\bibfield  {title} {\bibinfo {title} {Error correcting codes in
  quantum theory},\ }\href {https://doi.org/10.1103/PhysRevLett.77.793}
  {\bibfield  {journal} {\bibinfo  {journal} {Phys. Rev. Lett.}\ }\textbf
  {\bibinfo {volume} {77}},\ \bibinfo {pages} {793} (\bibinfo {year}
  {1996})}\BibitemShut {NoStop}%
\bibitem [{\citenamefont {Terhal}(2015)}]{RevModPhys.87.307}%
  \BibitemOpen
  \bibfield  {author} {\bibinfo {author} {\bibfnamefont {B.~M.}\ \bibnamefont
  {Terhal}},\ }\bibfield  {title} {\bibinfo {title} {Quantum error correction
  for quantum memories},\ }\href {https://doi.org/10.1103/RevModPhys.87.307}
  {\bibfield  {journal} {\bibinfo  {journal} {Rev. Mod. Phys.}\ }\textbf
  {\bibinfo {volume} {87}},\ \bibinfo {pages} {307} (\bibinfo {year}
  {2015})}\BibitemShut {NoStop}%
\bibitem [{\citenamefont {Yamasaki}(2019)}]{yamasaki2019entanglement}%
  \BibitemOpen
  \bibfield  {author} {\bibinfo {author} {\bibfnamefont {H.}~\bibnamefont
  {Yamasaki}},\ }\emph {\bibinfo {title} {Entanglement theory in distributed
  quantum information processing}},\ \href@noop {} {Ph.D. thesis},\ \bibinfo
  {school} {The University of Tokyo} (\bibinfo {year} {2019}),\ \Eprint
  {https://arxiv.org/abs/1903.09655} {arXiv:1903.09655 [quant-ph]} \BibitemShut
  {NoStop}%
\bibitem [{\citenamefont {Markham}\ and\ \citenamefont
  {Sanders}(2008)}]{Markham2008}%
  \BibitemOpen
  \bibfield  {author} {\bibinfo {author} {\bibfnamefont {D.}~\bibnamefont
  {Markham}}\ and\ \bibinfo {author} {\bibfnamefont {B.~C.}\ \bibnamefont
  {Sanders}},\ }\bibfield  {title} {\bibinfo {title} {Graph states for quantum
  secret sharing},\ }\href {https://doi.org/10.1103/PhysRevA.78.042309}
  {\bibfield  {journal} {\bibinfo  {journal} {Phys. Rev. A}\ }\textbf {\bibinfo
  {volume} {78}},\ \bibinfo {pages} {042309} (\bibinfo {year}
  {2008})}\BibitemShut {NoStop}%
\bibitem [{\citenamefont {Wang}\ \emph
  {et~al.}(2010{\natexlab{a}})\citenamefont {Wang}, \citenamefont {Zhang},
  \citenamefont {Tang}, \citenamefont {Zhan},\ and\ \citenamefont
  {You}}]{wang2010hierarchical}%
  \BibitemOpen
  \bibfield  {author} {\bibinfo {author} {\bibfnamefont {X.-W.}\ \bibnamefont
  {Wang}}, \bibinfo {author} {\bibfnamefont {D.-Y.}\ \bibnamefont {Zhang}},
  \bibinfo {author} {\bibfnamefont {S.-Q.}\ \bibnamefont {Tang}}, \bibinfo
  {author} {\bibfnamefont {X.-G.}\ \bibnamefont {Zhan}},\ and\ \bibinfo
  {author} {\bibfnamefont {K.-M.}\ \bibnamefont {You}},\ }\bibfield  {title}
  {\bibinfo {title} {Hierarchical quantum information splitting with six-photon
  cluster states},\ }\href
  {https://link.springer.com/article/10.1007/s10773-010-0461-8} {\bibfield
  {journal} {\bibinfo  {journal} {Int. J. Theor. Phys.}\ }\textbf {\bibinfo
  {volume} {49}},\ \bibinfo {pages} {2691} (\bibinfo {year}
  {2010}{\natexlab{a}})}\BibitemShut {NoStop}%
\bibitem [{\citenamefont {Wang}\ \emph
  {et~al.}(2010{\natexlab{b}})\citenamefont {Wang}, \citenamefont {Xia},
  \citenamefont {Wang},\ and\ \citenamefont {Zhang}}]{wang2010hierarchicalQIS}%
  \BibitemOpen
  \bibfield  {author} {\bibinfo {author} {\bibfnamefont {X.-W.}\ \bibnamefont
  {Wang}}, \bibinfo {author} {\bibfnamefont {L.-X.}\ \bibnamefont {Xia}},
  \bibinfo {author} {\bibfnamefont {Z.-Y.}\ \bibnamefont {Wang}},\ and\
  \bibinfo {author} {\bibfnamefont {D.-Y.}\ \bibnamefont {Zhang}},\ }\bibfield
  {title} {\bibinfo {title} {Hierarchical quantum-information splitting},\
  }\href {https://doi.org/https://doi.org/10.1016/j.optcom.2009.11.015}
  {\bibfield  {journal} {\bibinfo  {journal} {Opt. Commun.}\ }\textbf {\bibinfo
  {volume} {283}},\ \bibinfo {pages} {1196} (\bibinfo {year}
  {2010}{\natexlab{b}})}\BibitemShut {NoStop}%
\bibitem [{\citenamefont {Wang}\ \emph {et~al.}(2011)\citenamefont {Wang},
  \citenamefont {Zhang}, \citenamefont {Tang},\ and\ \citenamefont
  {Xie}}]{wang2011multiparty}%
  \BibitemOpen
  \bibfield  {author} {\bibinfo {author} {\bibfnamefont {X.-W.}\ \bibnamefont
  {Wang}}, \bibinfo {author} {\bibfnamefont {D.-Y.}\ \bibnamefont {Zhang}},
  \bibinfo {author} {\bibfnamefont {S.-Q.}\ \bibnamefont {Tang}},\ and\
  \bibinfo {author} {\bibfnamefont {L.-J.}\ \bibnamefont {Xie}},\ }\bibfield
  {title} {\bibinfo {title} {Multiparty hierarchical quantum-information
  splitting},\ }\href {https://doi.org/10.1088/0953-4075/44/3/035505}
  {\bibfield  {journal} {\bibinfo  {journal} {J. Phys. B}\ }\textbf {\bibinfo
  {volume} {44}},\ \bibinfo {pages} {035505} (\bibinfo {year}
  {2011})}\BibitemShut {NoStop}%
\bibitem [{\citenamefont {Yamasaki}\ and\ \citenamefont
  {Murao}(2019{\natexlab{c}})}]{yamasaki2019spread}%
  \BibitemOpen
  \bibfield  {author} {\bibinfo {author} {\bibfnamefont {H.}~\bibnamefont
  {Yamasaki}}\ and\ \bibinfo {author} {\bibfnamefont {M.}~\bibnamefont
  {Murao}},\ }\href@noop {} {\bibinfo {title} {Spread quantum information in
  one-shot quantum state merging}} (\bibinfo {year} {2019}{\natexlab{c}}),\
  \Eprint {https://arxiv.org/abs/1903.03619} {arXiv:1903.03619 [quant-ph]}
  \BibitemShut {NoStop}%
\bibitem [{\citenamefont {Hein}\ \emph {et~al.}(2006)\citenamefont {Hein},
  \citenamefont {D{\"u}r}, \citenamefont {Eisert}, \citenamefont {Raussendorf},
  \citenamefont {Van~den Nest},\ and\ \citenamefont {Briegel}}]{Hein2006}%
  \BibitemOpen
  \bibfield  {author} {\bibinfo {author} {\bibfnamefont {M.}~\bibnamefont
  {Hein}}, \bibinfo {author} {\bibfnamefont {W.}~\bibnamefont {D{\"u}r}},
  \bibinfo {author} {\bibfnamefont {J.}~\bibnamefont {Eisert}}, \bibinfo
  {author} {\bibfnamefont {R.}~\bibnamefont {Raussendorf}}, \bibinfo {author}
  {\bibfnamefont {M.}~\bibnamefont {Van~den Nest}},\ and\ \bibinfo {author}
  {\bibfnamefont {H.}~\bibnamefont {Briegel}},\ }in\ \href
  {http://arxiv.org/abs/quant-ph/0602096} {\emph {\bibinfo {booktitle}
  {Proceedings of the International School of Physics “Enrico Fermi”}}}\
  (\bibinfo  {publisher} {IOS press},\ \bibinfo {year} {2006})\ \Eprint
  {https://arxiv.org/abs/quant-ph/0602096} {arXiv:quant-ph/0602096}
  \BibitemShut {NoStop}%
\bibitem [{\citenamefont {Hillery}\ \emph {et~al.}(1999)\citenamefont
  {Hillery}, \citenamefont {Bu\ifmmode~\check{z}\else \v{z}\fi{}ek},\ and\
  \citenamefont {Berthiaume}}]{PhysRevA.59.1829}%
  \BibitemOpen
  \bibfield  {author} {\bibinfo {author} {\bibfnamefont {M.}~\bibnamefont
  {Hillery}}, \bibinfo {author} {\bibfnamefont {V.}~\bibnamefont
  {Bu\ifmmode~\check{z}\else \v{z}\fi{}ek}},\ and\ \bibinfo {author}
  {\bibfnamefont {A.}~\bibnamefont {Berthiaume}},\ }\bibfield  {title}
  {\bibinfo {title} {Quantum secret sharing},\ }\href
  {https://doi.org/10.1103/PhysRevA.59.1829} {\bibfield  {journal} {\bibinfo
  {journal} {Phys. Rev. A}\ }\textbf {\bibinfo {volume} {59}},\ \bibinfo
  {pages} {1829} (\bibinfo {year} {1999})}\BibitemShut {NoStop}%
\bibitem [{\citenamefont {Bandyopadhyay}(2000)}]{PhysRevA.62.012308}%
  \BibitemOpen
  \bibfield  {author} {\bibinfo {author} {\bibfnamefont {S.}~\bibnamefont
  {Bandyopadhyay}},\ }\bibfield  {title} {\bibinfo {title} {Teleportation and
  secret sharing with pure entangled states},\ }\href
  {https://doi.org/10.1103/PhysRevA.62.012308} {\bibfield  {journal} {\bibinfo
  {journal} {Phys. Rev. A}\ }\textbf {\bibinfo {volume} {62}},\ \bibinfo
  {pages} {012308} (\bibinfo {year} {2000})}\BibitemShut {NoStop}%
\bibitem [{\citenamefont {Deng}\ \emph {et~al.}(2005)\citenamefont {Deng},
  \citenamefont {Li}, \citenamefont {Li}, \citenamefont {Zhou},\ and\
  \citenamefont {Zhou}}]{PhysRevA.72.044301}%
  \BibitemOpen
  \bibfield  {author} {\bibinfo {author} {\bibfnamefont {F.-G.}\ \bibnamefont
  {Deng}}, \bibinfo {author} {\bibfnamefont {X.-H.}\ \bibnamefont {Li}},
  \bibinfo {author} {\bibfnamefont {C.-Y.}\ \bibnamefont {Li}}, \bibinfo
  {author} {\bibfnamefont {P.}~\bibnamefont {Zhou}},\ and\ \bibinfo {author}
  {\bibfnamefont {H.-Y.}\ \bibnamefont {Zhou}},\ }\bibfield  {title} {\bibinfo
  {title} {Multiparty quantum-state sharing of an arbitrary two-particle state
  with einstein-podolsky-rosen pairs},\ }\href
  {https://doi.org/10.1103/PhysRevA.72.044301} {\bibfield  {journal} {\bibinfo
  {journal} {Phys. Rev. A}\ }\textbf {\bibinfo {volume} {72}},\ \bibinfo
  {pages} {044301} (\bibinfo {year} {2005})}\BibitemShut {NoStop}%
\bibitem [{\citenamefont {Gaertner}\ \emph {et~al.}(2007)\citenamefont
  {Gaertner}, \citenamefont {Kurtsiefer}, \citenamefont {Bourennane},\ and\
  \citenamefont {Weinfurter}}]{PhysRevLett.98.020503}%
  \BibitemOpen
  \bibfield  {author} {\bibinfo {author} {\bibfnamefont {S.}~\bibnamefont
  {Gaertner}}, \bibinfo {author} {\bibfnamefont {C.}~\bibnamefont
  {Kurtsiefer}}, \bibinfo {author} {\bibfnamefont {M.}~\bibnamefont
  {Bourennane}},\ and\ \bibinfo {author} {\bibfnamefont {H.}~\bibnamefont
  {Weinfurter}},\ }\bibfield  {title} {\bibinfo {title} {Experimental
  demonstration of four-party quantum secret sharing},\ }\href
  {https://doi.org/10.1103/PhysRevLett.98.020503} {\bibfield  {journal}
  {\bibinfo  {journal} {Phys. Rev. Lett.}\ }\textbf {\bibinfo {volume} {98}},\
  \bibinfo {pages} {020503} (\bibinfo {year} {2007})}\BibitemShut {NoStop}%
\bibitem [{\citenamefont {Zheng}(2006)}]{PhysRevA.74.054303}%
  \BibitemOpen
  \bibfield  {author} {\bibinfo {author} {\bibfnamefont {S.-B.}\ \bibnamefont
  {Zheng}},\ }\bibfield  {title} {\bibinfo {title} {Splitting quantum
  information via w states},\ }\href
  {https://doi.org/10.1103/PhysRevA.74.054303} {\bibfield  {journal} {\bibinfo
  {journal} {Phys. Rev. A}\ }\textbf {\bibinfo {volume} {74}},\ \bibinfo
  {pages} {054303} (\bibinfo {year} {2006})}\BibitemShut {NoStop}%
\bibitem [{\citenamefont {Muralidharan}\ and\ \citenamefont
  {Panigrahi}(2008)}]{PhysRevA.78.062333}%
  \BibitemOpen
  \bibfield  {author} {\bibinfo {author} {\bibfnamefont {S.}~\bibnamefont
  {Muralidharan}}\ and\ \bibinfo {author} {\bibfnamefont {P.~K.}\ \bibnamefont
  {Panigrahi}},\ }\bibfield  {title} {\bibinfo {title} {Quantum-information
  splitting using multipartite cluster states},\ }\href
  {https://doi.org/10.1103/PhysRevA.78.062333} {\bibfield  {journal} {\bibinfo
  {journal} {Phys. Rev. A}\ }\textbf {\bibinfo {volume} {78}},\ \bibinfo
  {pages} {062333} (\bibinfo {year} {2008})}\BibitemShut {NoStop}%
\bibitem [{\citenamefont {Nielsen}\ and\ \citenamefont
  {Chuang}(2010)}]{nielsen2002quantum}%
  \BibitemOpen
  \bibfield  {author} {\bibinfo {author} {\bibfnamefont {M.~A.}\ \bibnamefont
  {Nielsen}}\ and\ \bibinfo {author} {\bibfnamefont {I.~L.}\ \bibnamefont
  {Chuang}},\ }\href {https://doi.org/10.1017/CBO9780511976667} {\emph
  {\bibinfo {title} {Quantum Computation and Quantum Information: 10th
  Anniversary Edition}}}\ (\bibinfo  {publisher} {Cambridge University Press},\
  \bibinfo {year} {2010})\BibitemShut {NoStop}%
\bibitem [{\citenamefont {Gottesman}(1998)}]{gottesman1998heisenberg}%
  \BibitemOpen
  \bibfield  {author} {\bibinfo {author} {\bibfnamefont {D.}~\bibnamefont
  {Gottesman}},\ }\bibfield  {title} {\bibinfo {title} {{The Heisenberg
  representation of quantum computers}},\ }in\ \href@noop {} {\emph {\bibinfo
  {booktitle} {{Group22: proceedings of the XXII International Colloquium on
  Group Theoretical Methods in Physics, Hobart, July 13-17, 1998}}}}\ (\bibinfo
  {year} {1998})\ pp.\ \bibinfo {pages} {32--43},\ \Eprint
  {https://arxiv.org/abs/quant-ph/9807006} {arXiv:quant-ph/9807006}
  \BibitemShut {NoStop}%
\bibitem [{\citenamefont {Aaronson}\ and\ \citenamefont
  {Gottesman}(2004)}]{Aaronson2004}%
  \BibitemOpen
  \bibfield  {author} {\bibinfo {author} {\bibfnamefont {S.}~\bibnamefont
  {Aaronson}}\ and\ \bibinfo {author} {\bibfnamefont {D.}~\bibnamefont
  {Gottesman}},\ }\bibfield  {title} {\bibinfo {title} {Improved simulation of
  stabilizer circuits},\ }\href {https://doi.org/10.1103/PhysRevA.70.052328}
  {\bibfield  {journal} {\bibinfo  {journal} {Phys. Rev. A}\ }\textbf {\bibinfo
  {volume} {70}},\ \bibinfo {pages} {052328} (\bibinfo {year}
  {2004})}\BibitemShut {NoStop}%
\bibitem [{\citenamefont {Gottesman}(1997)}]{gottesman1997stabilizer}%
  \BibitemOpen
  \bibfield  {author} {\bibinfo {author} {\bibfnamefont {D.}~\bibnamefont
  {Gottesman}},\ }\emph {\bibinfo {title} {Stabilizer Codes and Quantum Error
  Correction}},\ \href {https://doi.org/10.7907/rzr7-dt72} {Ph.D. thesis},\
  \bibinfo  {school} {California Institute of Technology} (\bibinfo {year}
  {1997})\BibitemShut {NoStop}%
\bibitem [{\citenamefont {Van~den Nest}\ \emph {et~al.}(2004)\citenamefont
  {Van~den Nest}, \citenamefont {Dehaene},\ and\ \citenamefont
  {De~Moor}}]{PhysRevA.69.022316}%
  \BibitemOpen
  \bibfield  {author} {\bibinfo {author} {\bibfnamefont {M.}~\bibnamefont
  {Van~den Nest}}, \bibinfo {author} {\bibfnamefont {J.}~\bibnamefont
  {Dehaene}},\ and\ \bibinfo {author} {\bibfnamefont {B.}~\bibnamefont
  {De~Moor}},\ }\bibfield  {title} {\bibinfo {title} {Graphical description of
  the action of local clifford transformations on graph states},\ }\href
  {https://doi.org/10.1103/PhysRevA.69.022316} {\bibfield  {journal} {\bibinfo
  {journal} {Phys. Rev. A}\ }\textbf {\bibinfo {volume} {69}},\ \bibinfo
  {pages} {022316} (\bibinfo {year} {2004})}\BibitemShut {NoStop}%
\bibitem [{\citenamefont {Moore}(1959)}]{10009005306}%
  \BibitemOpen
  \bibfield  {author} {\bibinfo {author} {\bibfnamefont {E.~F.}\ \bibnamefont
  {Moore}},\ }\bibfield  {title} {\bibinfo {title} {The shortest path through a
  maze},\ }in\ \href {https://doi.org/10.2307/3611677} {\emph {\bibinfo
  {booktitle} {Proc. Int. Symp. Switching Theory, 1959}}}\ (\bibinfo {year}
  {1959})\ pp.\ \bibinfo {pages} {285--292}\BibitemShut {NoStop}%
\bibitem [{\citenamefont {DiVincenzo}\ and\ \citenamefont
  {Shor}(1996)}]{divincenzo1996fault}%
  \BibitemOpen
  \bibfield  {author} {\bibinfo {author} {\bibfnamefont {D.~P.}\ \bibnamefont
  {DiVincenzo}}\ and\ \bibinfo {author} {\bibfnamefont {P.~W.}\ \bibnamefont
  {Shor}},\ }\bibfield  {title} {\bibinfo {title} {Fault-tolerant error
  correction with efficient quantum codes},\ }\href
  {https://doi.org/10.1103/PhysRevLett.77.3260} {\bibfield  {journal} {\bibinfo
   {journal} {Phys. Rev. Lett.}\ }\textbf {\bibinfo {volume} {77}},\ \bibinfo
  {pages} {3260} (\bibinfo {year} {1996})}\BibitemShut {NoStop}%
\bibitem [{\citenamefont {Greenberger}\ \emph {et~al.}(1990)\citenamefont
  {Greenberger}, \citenamefont {Horne}, \citenamefont {Shimony},\ and\
  \citenamefont {Zeilinger}}]{greenberger1990bell}%
  \BibitemOpen
  \bibfield  {author} {\bibinfo {author} {\bibfnamefont {D.~M.}\ \bibnamefont
  {Greenberger}}, \bibinfo {author} {\bibfnamefont {M.~A.}\ \bibnamefont
  {Horne}}, \bibinfo {author} {\bibfnamefont {A.}~\bibnamefont {Shimony}},\
  and\ \bibinfo {author} {\bibfnamefont {A.}~\bibnamefont {Zeilinger}},\
  }\bibfield  {title} {\bibinfo {title} {Bell’s theorem without
  inequalities},\ }\href {https://doi.org/10.1119/1.16243} {\bibfield
  {journal} {\bibinfo  {journal} {Am. J. Phys.}\ }\textbf {\bibinfo {volume}
  {58}},\ \bibinfo {pages} {1131} (\bibinfo {year} {1990})}\BibitemShut
  {NoStop}%
\bibitem [{\citenamefont {Yamasaki}\ and\ \citenamefont
  {Koashi}(2022)}]{https://doi.org/10.48550/arxiv.2207.08826}%
  \BibitemOpen
  \bibfield  {author} {\bibinfo {author} {\bibfnamefont {H.}~\bibnamefont
  {Yamasaki}}\ and\ \bibinfo {author} {\bibfnamefont {M.}~\bibnamefont
  {Koashi}},\ }\href {https://doi.org/10.48550/ARXIV.2207.08826} {\bibinfo
  {title} {Time-efficient constant-space-overhead fault-tolerant quantum
  computation}} (\bibinfo {year} {2022})\BibitemShut {NoStop}%
\bibitem [{\citenamefont {Gottesman}(2014)}]{gottesman2014faulttolerant}%
  \BibitemOpen
  \bibfield  {author} {\bibinfo {author} {\bibfnamefont {D.}~\bibnamefont
  {Gottesman}},\ }\bibfield  {title} {\bibinfo {title} {Fault-tolerant quantum
  computation with constant overhead},\ }\href
  {https://doi.org/10.26421/QIC14.15-16-5} {\bibfield  {journal} {\bibinfo
  {journal} {Quantum Info. Comput.}\ }\textbf {\bibinfo {volume} {14}},\
  \bibinfo {pages} {1338–1372} (\bibinfo {year} {2014})}\BibitemShut
  {NoStop}%
\bibitem [{\citenamefont {Fawzi}\ \emph {et~al.}(2020)\citenamefont {Fawzi},
  \citenamefont {Grospellier},\ and\ \citenamefont {Leverrier}}]{8555154}%
  \BibitemOpen
  \bibfield  {author} {\bibinfo {author} {\bibfnamefont {O.}~\bibnamefont
  {Fawzi}}, \bibinfo {author} {\bibfnamefont {A.}~\bibnamefont {Grospellier}},\
  and\ \bibinfo {author} {\bibfnamefont {A.}~\bibnamefont {Leverrier}},\
  }\bibfield  {title} {\bibinfo {title} {Constant overhead quantum fault
  tolerance with quantum expander codes},\ }\href
  {https://doi.org/10.1145/3434163} {\bibfield  {journal} {\bibinfo  {journal}
  {Commun. ACM}\ }\textbf {\bibinfo {volume} {64}},\ \bibinfo {pages}
  {106–114} (\bibinfo {year} {2020})}\BibitemShut {NoStop}%
\bibitem [{\citenamefont {Dahlberg}\ \emph {et~al.}(2020)\citenamefont
  {Dahlberg}, \citenamefont {Helsen},\ and\ \citenamefont
  {Wehner}}]{Dahlberg2020transforminggraph}%
  \BibitemOpen
  \bibfield  {author} {\bibinfo {author} {\bibfnamefont {A.}~\bibnamefont
  {Dahlberg}}, \bibinfo {author} {\bibfnamefont {J.}~\bibnamefont {Helsen}},\
  and\ \bibinfo {author} {\bibfnamefont {S.}~\bibnamefont {Wehner}},\
  }\bibfield  {title} {\bibinfo {title} {Transforming graph states to
  {B}ell-pairs is {NP}-{C}omplete},\ }\href
  {https://doi.org/10.22331/q-2020-10-22-348} {\bibfield  {journal} {\bibinfo
  {journal} {{Quantum}}\ }\textbf {\bibinfo {volume} {4}},\ \bibinfo {pages}
  {348} (\bibinfo {year} {2020})}\BibitemShut {NoStop}%
\bibitem [{\citenamefont {Litinski}(2019)}]{Litinski2019magicstate}%
  \BibitemOpen
  \bibfield  {author} {\bibinfo {author} {\bibfnamefont {D.}~\bibnamefont
  {Litinski}},\ }\bibfield  {title} {\bibinfo {title} {Magic {S}tate
  {D}istillation: {N}ot as {C}ostly as {Y}ou {T}hink},\ }\href
  {https://doi.org/10.22331/q-2019-12-02-205} {\bibfield  {journal} {\bibinfo
  {journal} {{Quantum}}\ }\textbf {\bibinfo {volume} {3}},\ \bibinfo {pages}
  {205} (\bibinfo {year} {2019})}\BibitemShut {NoStop}%
\bibitem [{\citenamefont {{\L}odyga}\ \emph {et~al.}(2015)\citenamefont
  {{\L}odyga}, \citenamefont {Mazurek}, \citenamefont {Grudka},\ and\
  \citenamefont {Horodecki}}]{Lodyga2015}%
  \BibitemOpen
  \bibfield  {author} {\bibinfo {author} {\bibfnamefont {J.}~\bibnamefont
  {{\L}odyga}}, \bibinfo {author} {\bibfnamefont {P.}~\bibnamefont {Mazurek}},
  \bibinfo {author} {\bibfnamefont {A.}~\bibnamefont {Grudka}},\ and\ \bibinfo
  {author} {\bibfnamefont {M.}~\bibnamefont {Horodecki}},\ }\bibfield  {title}
  {\bibinfo {title} {Simple scheme for encoding and decoding a qubit in unknown
  state for various topological codes},\ }\href
  {https://doi.org/10.1038/srep08975} {\bibfield  {journal} {\bibinfo
  {journal} {Sci. Rep.}\ }\textbf {\bibinfo {volume} {5}},\ \bibinfo {pages}
  {8975} (\bibinfo {year} {2015})}\BibitemShut {NoStop}%
\end{thebibliography}%
\end{document}